\documentclass[11pt]{article}

\usepackage[T1]{fontenc}
\usepackage{lmodern}
\usepackage{fullpage}
\usepackage{microtype}
\usepackage{amsmath,amssymb,amsthm,mathtools}
\usepackage{enumitem}
\usepackage{tikz}
\usepackage{float}
\usepackage{caption}
\usepackage{mdframed}
\IfFormatAtLeastTF{2026-06-01}{}{\usepackage{aliascnt}}
\usepackage[colorlinks=true,citecolor=blue,linkcolor=blue,urlcolor=blue,pdftitle={Approximating two-terminal network reliability},pdfauthor={Weiming Feng, Yucheng Fu, Heng Guo}]{hyperref}
\usepackage[nameinlink]{cleveref}
\usepackage{xcolor}

\allowdisplaybreaks

\newtheorem{theorem}{Theorem}

\IfFormatAtLeastTF{2026-06-01}{%
  \newtheorem{lemma}[theorem]{Lemma}
  \newtheorem{proposition}[theorem]{Proposition}
  \newtheorem{corollary}[theorem]{Corollary}
  \newtheorem{observation}[theorem]{Observation}
  \newtheorem{claim}[theorem]{Claim}
}{%
  \newaliascnt{lemma}{theorem}
  \newtheorem{lemma}[lemma]{Lemma}
  \aliascntresetthe{lemma}

  \newaliascnt{proposition}{theorem}
  
  \aliascntresetthe{proposition}

  \newaliascnt{corollary}{theorem}
  
  \aliascntresetthe{corollary}

  \newaliascnt{observation}{theorem}
  \newtheorem{observation}[observation]{Observation}
  \aliascntresetthe{observation}

  \newaliascnt{claim}{theorem}
  
  \aliascntresetthe{claim}
}

\theoremstyle{definition}

\IfFormatAtLeastTF{2026-06-01}{%
  \newtheorem{definition}[theorem]{Definition}
  \newtheorem{problem}[theorem]{Problem}
  \newtheorem{remark}[theorem]{Remark}
}{%
  \newaliascnt{definition}{theorem}
  \newtheorem{definition}[definition]{Definition}
  \aliascntresetthe{definition}

  \newaliascnt{problem}{theorem}
  
  \aliascntresetthe{problem}

  \newaliascnt{remark}{theorem}
  \newtheorem{remark}[remark]{Remark}
  \aliascntresetthe{remark}

}

\newtheorem{algorithm}{Algorithm}
\newfloat{algorithmfloat}{htbp}{loa}
\surroundwithmdframed[
  nobreak=true,
  linewidth=0.6pt,
  innerleftmargin=10pt,
  innerrightmargin=10pt,
  innertopmargin=6pt,
  innerbottommargin=6pt,
  skipabove=0pt,
  skipbelow=0pt
]{algorithm}

\crefname{theorem}{theorem}{theorems}
\Crefname{theorem}{Theorem}{Theorems}
\crefname{lemma}{lemma}{lemmas}
\Crefname{lemma}{Lemma}{Lemmas}
\crefname{proposition}{proposition}{propositions}
\Crefname{proposition}{Proposition}{Propositions}
\crefname{corollary}{corollary}{corollaries}
\Crefname{corollary}{Corollary}{Corollaries}
\crefname{observation}{observation}{observations}
\Crefname{observation}{Observation}{Observations}
\crefname{claim}{claim}{claims}
\Crefname{claim}{Claim}{Claims}
\crefname{definition}{definition}{definitions}
\Crefname{definition}{Definition}{Definitions}
\crefname{problem}{problem}{problems}
\Crefname{problem}{Problem}{Problems}
\crefname{remark}{remark}{remarks}
\Crefname{remark}{Remark}{Remarks}
\crefname{algorithm}{algorithm}{algorithms}
\Crefname{algorithm}{Algorithm}{Algorithms}

\newcommand{\E}{\mathbb E}
\newcommand{\Var}{\operatorname{Var}}

\newcommand{\defeq}{:=}

\renewcommand{\Pr}{\mathop{\mathrm{Pr}}\nolimits}
\newcommand{\eps}{\varepsilon}
\newcommand{\abs}[1]{\ensuremath{\left\vert#1\right\vert}}

\newcommand{\ind}{\mathbf 1}
\newcommand{\reach}{\rightsquigarrow}

\newcommand{\Rel}{\textnormal{\textsf{Rel}}}
\newcommand{\UnRel}{\textnormal{\textsf{UnRel}}}
\newcommand{\numP}{\#\textnormal{\textsf{P}}}
\newcommand{\SpanL}{\textnormal{\textsf{SpanL}}}
\newcommand{\NP}{\textnormal{\textsf{NP}}}
\newcommand{\BIS}{\#\textnormal{\textsf{BIS}}}

\def\*#1{\mathbf{#1}} % Use \*A for \mathbf{A}
\def\+#1{\mathcal{#1}} % Use \+A for \mathcal{A}
\def\-#1{\mathrm{#1}} % Use \-A for \mathrm{A}
\def\^#1{\mathbb{#1}} % Use \^A for \mathbb{A}

\definecolor{HGcolor}{RGB}{255,50,50}

\title{Approximating two-terminal network reliability}

\author{
Weiming Feng\footnote{School of Computing and Data Science, The University of Hong Kong, HK, China} 
\and
Yucheng Fu\footnotemark[1]
\and
Heng Guo\footnote{School of Informatics, University of Edinburgh, Informatics Forum, Edinburgh, EH8 9AB, UK}
}

\date{}

\begin{document}

\maketitle

\begingroup
\renewcommand{\thefootnote}{}
\makeatletter
\def\Hy@footnote@currentHref{email-addresses}
\makeatother
\footnotetext{%
{E-mail address:} \texttt{wfeng@hku.hk},
\texttt{fyc0130@connect.hku.hk}, and \texttt{hguo@inf.ed.ac.uk}.}
\endgroup

\begin{abstract}
  We present fully polynomial-time randomised approximation schemes (FPRAS) for the two-terminal reliability problem on general graphs, both directed and undirected.
  For directed graphs with $n$ vertices and $m$ arcs, the running time is $\widetilde{O}(nm\varepsilon^{-2})$, where $\varepsilon$ is the relative error.
  For undirected graphs with $n$ vertices and $m$ edges, it runs in $\widetilde{O}(m+n^2\varepsilon^{-2})$.
  We also show that the complementary unreliability problem is \BIS-hard.
\end{abstract}

\section{Introduction}

In 1979, Valiant \cite{Valiant1979} introduced the computational complexity class \numP{}, a counting analogue of \NP{}.
One of his main motivations was to capture the intrinsic hardness of computing network reliability measures.
Here, the setting is that we are given a stochastic network, which can be a directed or undirected graph.
Each arc in a directed graph, or edge in an undirected graph, opens or closes independently with some probability, and reliability is measured by the probability of connectivity events in the open subgraph.
One of the most prominent variants is two-terminal reliability,
which asks for the probability that a given terminal $s$ can reach another given terminal $t$ in the stochastic graph.
Indeed, this is listed as one of the first thirteen \numP-complete problems by Valiant~\cite{Valiant1979}.

Reliability measures have been extensively studied, and it turns out that most variants are $\numP$-complete \cite{Ball80,Jerrum1981,BP83,PB83,Ball86,Col87},
which means that efficient exact computation is unlikely.
Two-terminal reliability remains \numP-complete when every failure probability is $1/2$ \cite{Valiant1979}, on directed acyclic graphs~\cite{PB83}, or on planar graphs with maximum degree three~\cite{Provan1986}.

Due to the apparent intractability, it is then natural to ask for approximation. 
However, progress on the approximation complexity of reliability measures has been much slower.
For these problems, the appropriate notion of approximation is fully polynomial-time randomised approximation schemes (FPRAS),
which provide a relative approximation in time polynomial in the input size and in the inverse relative error.
The first FPRAS in this context is due to Karger \cite{Kar01} for all-terminal undirected unreliability,
whereas the complementary all-terminal reliability approximation was obtained much later by Guo and Jerrum \cite{GuoJerrum2019}.\footnote{While reliability and unreliability sum to $1$, either quantity can be exponentially small, which prevents a relative approximation of one from translating into a relative approximation of the other.}
For two-terminal reliability, much less is known.
The special case restricted to directed acyclic graphs (DAGs) admits an FPRAS,
because this special case belongs to a certain complexity class called \SpanL{}.
Arenas, Croquevielle, Jayaram, and Riveros~\cite{ACJR21} showed that all problems in \SpanL{} admit FPRASes.
(See also \cite{AvBGM25,FG24}.)

The approximation complexity of two-terminal reliability in general, despite its old history~\cite{KarpLuby1985,ZenklusenLaumanns2011}, remains unresolved.
It was explicitly asked in the survey by Kannan \cite{Kannan1994} if an FPRAS exists.
In this paper we answer this question positively, by giving algorithms for both the directed and undirected cases.

Next we formally introduce the directed version.
Let $G = (V,E)$ be a directed graph with two distinguished vertices $s,t$. Each arc $e \in E$ is associated with an \emph{open probability} $p_e \in (0,1]$. The \emph{two-terminal network reliability} of $G$ asks the following question: suppose that each arc is independently open with probability $p_e$, what is the probability that there is a path from $s$ to $t$?
Represent an arc configuration by $\omega\in\{0,1\}^E$, where $\omega_e=1$ means that $e$ is open.
The probability can be written as
\begin{align*}
  \Rel(s,t) = \Pr_{G_p}[s \reach t] = \sum_{\substack{\omega\in\{0,1\}^E:\\ s \reach t \text{ in } \omega}} \ \prod_{e:\,\omega_e=1} p_e \prod_{e:\,\omega_e=0} (1 - p_e),
\end{align*}
where $G_p$ denotes the random subgraph of open arcs.
In the special case when all $p_e = 1/2$, the problem amounts to counting the spanning subgraphs of $G$ in which $s$ can reach $t$.

By \emph{FPRAS} (fully polynomial-time randomised approximation scheme), we mean a randomised algorithm that, given a directed graph $G = (V,E)$, two terminals $s,t \in V$, rational open probabilities $p_e \in (0,1]$ for each arc $e \in E$, and an error parameter $\varepsilon > 0$, outputs a random variable $\widehat{R}$ such that
\begin{align*}
    \Pr\left[(1-\varepsilon) \textsf{Rel}(s,t) \leq \widehat{R} \leq (1+\varepsilon) \textsf{Rel}(s,t)\right] \geq \frac{3}{4},
\end{align*}
and whose running time is polynomial in the input length and in $\varepsilon^{-1}$.

%Throughout, $\log$ denotes the natural logarithm unless a base is shown.
As usual for approximate counting algorithms, running-time bounds count unit-cost arithmetic operations and %comparisons on numbers, together with 
independent draws of uniform random variables on $(0,1)$. % In particular, the cost of an arithmetic operation does not depend on the bit lengths of the numbers involved.
%The case $m=0$ is trivial, as then $\Rel(s,t)=1$ if $s=t$ and $\Rel(s,t)=0$ otherwise.
We state our bounds for $m\geq 1$ and $\varepsilon\in(0,1]$.
For $\varepsilon>1$ one runs the algorithm with $\varepsilon=1$.
Our first theorem is as follows.

\begin{theorem}
\label{thm:main}
There is an FPRAS for two-terminal reliability on directed graphs. On an input with $n = |V|$ vertices, $m = |E|$ arcs, and error parameter $\varepsilon\in(0,1]$, its running time is
$\widetilde{O}(nm\varepsilon^{-2})$,
where $\widetilde{O}(\cdot)$ hides factors polylogarithmic in $n$, $m$, and $\varepsilon^{-1}$.
\end{theorem}

%In the unit-cost model above, the running time does not depend on the open probabilities, and in particular not on $\Rel(s,t)$ or on the smallest open probability.
The algorithm in fact estimates the probabilities of $s$ reaching the vertices $v$ in $G_p$, roughly in decreasing order, until it has estimated the one for $t$.
This is realised by DFS-type explorations that repeatedly restart from saved sets of discovered vertices.
It is a significant departure from previously known reliability algorithms, such as \cite{Kar01,GuoJerrum2019,CGZZ24,ACJR21}.
In the first version of our paper \cite{FFG26v1}, we gave a slower algorithm using ideas of simulated annealing and Markov chains, in the spirit of \cite{JSV04}.
However, to achieve the fast runtime in \Cref{thm:main}, the algorithm is now completely different.
We give a high-level overview of the algorithm in \Cref{sec:overview}.

The algorithm in \Cref{thm:main} actually runs in time $\widetilde{O}(m+n(n+P)\varepsilon^{-2})$, where $P\defeq\sum_{e\in E}p_e$ is the total open probability (\Cref{thm:expdeg}), and \Cref{thm:main} is a corollary of this bound.
This bound is better when $n+P\ll m$, which is more relevant for the undirected case.

In the \emph{undirected} version, we ask the probability that $s$ and $t$ are connected in the random subgraph.
For this variant, efficient algorithms are known only for graphs of low tree-width~\cite{SatyanarayanaWood1985,GoharshadyMohammadi2020}, and no efficient approximation algorithm in general is known.
The undirected version reduces to the directed one.
One could use a gadget by Jerrum \cite{Jerrum1981}, or, more simply, replace every edge $\{u,v\}$ by the two arcs $(u,v)$ and $(v,u)$, each of which is open independently with the open probability of the edge.
Although the two arcs replacing an edge are now independent, the probability that $s$ reaches $t$ in the bi-directed graph equals the probability that $s$ and $t$ are connected in the undirected graph, by a result of Ball~\cite[Corollary~1]{Ball80} (\Cref{lem:bidirected}). (See also McDiarmid~\cite{McDiarmid1981}, Linusson~\cite[Lemma~2.1]{Linusson2009}, and Guo and Jerrum~\cite[Section~5]{GuoJerrum2019}.)
The directed graph has the same $n$ vertices and at most $2m$ arcs, so \Cref{thm:main} already gives running time $\widetilde{O}(nm\varepsilon^{-2})$ for undirected graphs.
Our next theorem improves this bound to $\widetilde{O}(m+n^2\varepsilon^{-2})$, proved in \Cref{sec:undirected}.
%We do better by first contracting the pieces of the graph that are highly connected with respect to the capacities $p_e$.
%Every such piece is connected in the random subgraph with high probability, and contracting all of them can only increase a two-terminal reliability, by a factor of at most $1/(1-\varepsilon/8)$, however small the reliability is.

\begin{theorem}
\label{thm:undirected}
There is an FPRAS for two-terminal reliability on undirected graphs. On an input with $n$ vertices, $m$ edges, and error parameter $\varepsilon\in(0,1]$, its running time is $\widetilde{O}(m+n^2\varepsilon^{-2})$.
\end{theorem}

The faster runtime is mainly due to a sparsification lemma (\Cref{lem:sparsification}),
which states that we can construct a new graph $G'$ with at most $n$ vertices such that the reliability is preserved but the total open probability is $\widetilde{O}(n)$.
This is inspired by the sparse reliable backbone construction of Chechik, Emek, Patt-Shamir and Peleg~\cite{ChechikEtAl2012},
which gives sparse spanning subgraphs that approximately preserve reliabilities.
However, their construction does not extend to all probability ranges.
This is because with arbitrary open probabilities, a spanning subgraph with the same probabilities may have to keep every edge (\Cref{rem:backbone}).
We take a different route of contracting highly connected components,
which is similar to the sparsification argument of Feng, Guo and Zhang~\cite[Appendix~A]{FengGuoZhang2026} on all-terminal reliability.
Their argument is based on Bencz\'ur and Karger's notion of strengths~\cite{BK15}.
While our goal can also be achieved by strengths, we give a simpler randomised construction instead in \Cref{sec:sparsification-proof}.
Similar ideas are also used in approximation algorithms for all-terminal unreliability~\cite{Kar16,CHLP24,CenLiPanigrahi2025}.

In light of Karger's celebrated algorithm \cite{Kar01},
it is also interesting to ask the approximation complexity for unreliability $\UnRel(s,t)=1-\Rel(s,t)$.
Feng and Guo \cite{FG24} showed that even in DAGs, approximating $s$-$t$ unreliability is \BIS{}-hard.
Here, \BIS{} is the problem of counting the number of independent sets in bipartite graphs.
While \BIS{} is not known to be \NP{}-hard to approximate, it is also conjectured to have no FPRAS \cite{DGGJ03}.
On the other hand, the same problem in undirected graphs is less clear.
The only prior result is due to Karger \cite{Kar01}, where an FPRAS is presented if the minimum $s$-$t$ cut has the same order as the global minimum cut.
We show that in undirected graphs, the problem is \BIS{}-hard as well, even when every edge has the same fixed rational open probability, and thus is unlikely to admit an FPRAS.

\begin{theorem}  \label{thm:unrel-hard}
  For any rational $p\in(0,1)$,
  if undirected $s$-$t$ unreliability at fixed open probability $p$ for all edges admits an FPRAS, then \BIS{} admits an FPRAS.
\end{theorem}

Indeed, Provan and Ball \cite{PB83} already gave a parsimonious reduction from \BIS{} to the number of $s$-$t$ minimum cuts.
The ``ground states'' for $s$-$t$ unreliability are exactly these min-cuts, and thus we just need a boosting argument so that these states contribute a significant portion of the unreliability.
A full proof is given in \Cref{sec:hardness} for completeness.
Note that if we consider $s$-$t$ \emph{minimal} cuts, the problem becomes \NP{}-hard \cite{GGL16}.

The \BIS{}-hardness results in \Cref{thm:unrel-hard} and \cite{FG24} leave open the possibility of $s$-$t$ unreliability being either \NP{}-hard or \BIS{}-easy.
We find both possibilities to be plausible and do not make a conjecture here.

\subsection{High-level overview of \texorpdfstring{\Cref{thm:main}}{the main theorem}}\label{sec:overview}

Given a graph $G=(V,E)$, it is natural to define a so-called Gibbs distribution
on the configurations $\omega$ in which $s$ can reach $t$,
in which the probability of $\omega$ is proportional to its probability in the product distribution (bond percolation).
Approximating $\Rel(s,t)$ can be reduced to drawing random configurations from this Gibbs distribution, and the standard way to do so is Glauber dynamics.
Since this recipe works well for all-terminal reliability (see, e.g.~\cite{CGZZ24}),
it is tempting to try it for two-terminal reliability as well.
However, it fails quite badly for a simple example where $s$ and $t$ are connected by two long, internally disjoint directed paths.
If either path is open, $s$ can reach $t$,
but switching from one path to the other requires both paths to be open at once.
When each arc opens with probability $\frac{1}{2}$, this is exponentially unlikely.
See \Cref{fig:glauber-bottleneck}.
This implies that straightforward Glauber dynamics mixes slowly on graphs like this.

\begin{figure}[htbp]
\centering
\begin{tikzpicture}[
    >=stealth,
    font=\small,
    vertex/.style={circle,draw=black!65,fill=white,minimum size=5.5mm,inner sep=0pt},
    terminal/.style={vertex,draw=black,very thick,minimum size=6.5mm},
    openarc/.style={->,very thick,draw=blue!70!black},
    freearc/.style={->,thick,densely dotted,draw=violet!55!black},
    pics/twopath/.style args={#1/#2}{
      code={
        \node[terminal] (-s) at (-2.2,0) {$s$};
        \node[vertex] (-a1) at (-1.1,0.55) {};
        \node[vertex] (-a2) at (0,0.55) {};
        \node[vertex] (-a3) at (1.1,0.55) {};
        \node[terminal] (-t) at (2.2,0) {$t$};
        \node[vertex] (-b1) at (-1.1,-0.55) {};
        \node[vertex] (-b2) at (0,-0.55) {};
        \node[vertex] (-b3) at (1.1,-0.55) {};
        \draw[#1] (-s) -- (-a1);
        \draw[#1] (-a1) -- (-a2);
        \draw[#1] (-a2) -- (-a3);
        \draw[#1] (-a3) -- (-t);
        \draw[#2] (-s) -- (-b1);
        \draw[#2] (-b1) -- (-b2);
        \draw[#2] (-b2) -- (-b3);
        \draw[#2] (-b3) -- (-t);
      }
    }
]

\pic (left)  at (0,0)    {twopath={openarc/freearc}};
\pic (mid)   at (5.5,0)  {twopath={openarc/openarc}};
\pic (right) at (11.0,0) {twopath={freearc/openarc}};

\node[font=\footnotesize] at (0,-1.25)    {(a) $\omega \in A$: upper path open};
\node[font=\footnotesize] at (5.5,-1.25)  {(b) $\omega \in B = A \cap C$: both paths open};
\node[font=\footnotesize] at (11.0,-1.25) {(c) $\omega \in C$: lower path open};

\end{tikzpicture}
\caption{The Glauber dynamics can pass from $A$ to $C$ only through the rare set $B$. Solid arcs are open, and dotted arcs are subject to no constraint.}
\label{fig:glauber-bottleneck}
\end{figure}

While MCMC can be made to work together with simulated annealing, as in the first version of our paper \cite{FFG26v1}, our current fast algorithm does not use Markov chains.
Instead, we estimate the reliabilities $q_v\defeq\Pr_{G_p}[s\reach v]$ of the vertices $v$, roughly in decreasing order until $q_t$ is estimated, and we reuse the work spent on reaching vertices of large reliability to reach vertices of smaller reliability.
This avoids spending work proportional to $1/q_t$ like in plain Monte Carlo, which is hopeless when $q_t$ is exponentially small.
We first describe a simpler version of the algorithm, \Cref{alg:S}, which already contains the ideas that make the estimates correct.
Then we describe the changes that make it fast.

\paragraph{\texorpdfstring{\Cref{alg:S}}{Algorithm S}.}
Consider a depth-first search from $s$ in $G_p$, implemented via a stack, that reveals the state of an arc only when it examines the arc.
When the search pops a vertex $u$, neither the arcs leaving $u$ nor those leaving undiscovered vertices have been examined, so the rest of the search depends on its past only through the set $D$ of discovered vertices.
Thus the pair $(D,u)$, a \emph{snapshot}, suffices to continue the search with fresh arc states.
\Cref{alg:S} keeps a set $R$ of \emph{resolved} vertices, and, for each $u\in R$, an estimate $\widehat q_u$ of $q_u$ and a \emph{pool} of snapshots at $u$, which approximates the law of the snapshot at $u$ given that $s$ reaches $u$.
Each round chooses an arc $(u,v)$ leaving $R$ that maximises $\widehat q_up_{(u,v)}$, where $u\in R$, and resolves its head $v$.
To do so, it restarts many searches from snapshots drawn from the pools of the current frontier of the resolved vertices.
%A start $u\in R$ is drawn with probability proportional to $\widehat q_u\gamma_u$, where $\gamma_u$ is the probability that some arc from $u$ to an unresolved vertex is open, and the search is conditioned on this event.
A search records its snapshot at $v$ when $v$ is reached along the path of its search tree that avoids resolved vertices other than the start.
The records form the pool of $v$, all with the same weight, and their total weight is $\widehat q_v$.
The pool of $v$ is only formed once when $v$ is resolved,
and never changes after that.
The snapshots in the pool are reused each time the exploration restarts at~$v$.

\paragraph{Why \texorpdfstring{\Cref{alg:S}}{Algorithm S} is correct.}
For every vertex $w$, the pools \emph{predict} $q_w$: taking the pools as given and computing the rest exactly, the prediction splits the paths from $s$ to $w$ in the search tree at their last resolved vertex.
This gives a last-exit decomposition.
The prediction equals $q_w$ at the start and $\widehat q_w$ once $w$ is resolved, and every round leaves it unchanged in expectation, since the new records replace, in expectation, exactly what the old pools stop contributing.
Thus the predictions are martingales, and the output $\widehat q_t$ is unbiased.
The rule choosing $(u,v)$ ensures that a single record changes every prediction by a controlled amount.
Chebyshev's inequality then gives the accuracy with $O(n^3m\varepsilon^{-2}\delta^{-1})$ samples per round, about $n^4m^2\varepsilon^{-2}\delta^{-1}$ operations in total (\Cref{thm:alg-s}).

\paragraph{The main algorithm.}
The algorithm of \Cref{thm:main} keeps these ingredients and makes three changes (\Cref{sec:overview-alg}):
\begin{enumerate}
  \item \emph{Resolving a layer.}
    Instead of resolving one vertex at a time, we try to resolve a whole layer of vertices.
    This layer is selected by a pilot run,
    and due to selection bias, after a layer is selected, we execute a fresh resolving run.
  \item \emph{Skipping resolved vertices.}
    The exploration process skips resolved vertices.
    To accommodate this change, we need to adjust the definition of predictions by using reachability.
    Moreover, to keep them martingales, the explorations defer the vertices being resolved until no other discovered vertex is waiting.
  \item \emph{Implementation.}
    We use data structures that draw the open arcs leaving a vertex directly, without testing the closed ones, and that restart an exploration from a snapshot without copying its discovered set.
    An amortised analysis over the rounds then bounds the total work, which gives the running time in \Cref{thm:main}.
\end{enumerate}
The analysis of the main algorithm follows roughly the same highlevel outline as \Cref{alg:S}.
We first prove that the predictions are martingales (\Cref{lem:martingale}), 
then show that typically the increments are controlled (\Cref{lem:increments}),
and finally bound the error via Freedman's inequality (\Cref{lem:freedman}).

\paragraph{Related work.}
The idea of reusing samples has been explored before, such as in \cite{ACJR21,FG24,MCM24,MdC25,MdC26}.
These algorithms process the instance in an order fixed in advance,
whereas we resolve the vertices in an order adapting to the estimates. 
More importantly, the key new insight is the prediction martingale, which is enabled by the last-exit decomposition.
As a result, reused samples cause no bias in the estimates, and Freedman's inequality controls the error despite correlation caused by reusing samples.
Since the last-exit decomposition is disjoint, we also never need to handle unions of overlapping sets, unlike these algorithms.

After the first version of our paper \cite{FFG26v1} appearing on arXiv, Chen and Liu \cite{CL26} presented an alternative way of approximating two-terminal reliability via common bases of matroids.
Their method is more general, but it requires $\widetilde{O}(m^{13})$ independence oracle calls.
Thus it is much slower than \Cref{thm:main} and \Cref{thm:undirected}.

\subsection{Organisation}
\Cref{sec:algorithm-s} reduces the input, fixes notation, and describes and analyses \Cref{alg:S}.
\Cref{sec:algorithm} explains how the main algorithm modifies \Cref{alg:S}, both in the algorithm and in its analysis, describes the main algorithm in detail, and reduces \Cref{thm:main,thm:expdeg} to two lemmas, on the accuracy and on the running time.
\Cref{sec:searches} shows that the resolving runs keep the predictions unbiased, \Cref{sec:accuracy} proves the accuracy, and \Cref{sec:work} gives the implementation and bounds the work.
\Cref{sec:undirected} proves \Cref{thm:undirected}, and \Cref{sec:hardness} proves \Cref{thm:unrel-hard}.

\section{\texorpdfstring{\Cref{alg:S}}{Algorithm S}}\label{sec:algorithm-s}

Here we introduce a simple but slow algorithm, termed \Cref{alg:S} (S for simple) first.
It estimates the reliabilities $q_v$ roughly in decreasing order until $q_t$ is estimated, and it reuses the work spent on reaching vertices of large reliability to reach vertices of smaller reliability.
\Cref{sec:reduction} reduces the input and fixes notation, which the main algorithm in \Cref{sec:algorithm} shares.
\Cref{alg:S} itself is described in \Cref{sec:s-dfs,sec:simple-alg}, and analysed in \Cref{sec:s-analysis}.

\subsection{Reduction and notation}\label{sec:reduction}

Let $G=(V,E)$, $s$, $t$ and $\boldsymbol{p}=(p_e)_{e\in E}\in(0,1]^E$ be the input.
If $s=t$, then $\Rel(s,t)=1$, so assume that $s\neq t$.
Delete every vertex that cannot be reached from $s$ in $G$, together with its incident arcs.
As every open probability is positive, these are exactly the vertices that $s$ reaches with probability zero, and a path starting at $s$ visits none of them.
Thus the deletion does not change the probability that $s$ reaches any vertex that is kept.
If $t$ is deleted, then $\Rel(s,t)=0$ and we return $0$.
One graph search implements this reduction with $O(n+m)$ operations.
The new numbers of vertices and arcs are at most the old ones, so it suffices to prove \Cref{thm:main,thm:alg-s} for the reduced graph.
From now on we therefore assume the following.

\begin{equation}\label{asm:standing}\tag{SA}
  \parbox{0.85\linewidth}{\emph{Standing assumption.} $G=(V,E)$ is a directed graph with $n=|V|\geq2$ vertices and $m=|E|$ arcs, $\boldsymbol{p}\in(0,1]^E$, $s\neq t$, and every vertex is reachable from $s$ in $G$.
  Parallel arcs and self-loops are allowed.}
\end{equation}

In particular $n-1\leq m$: every vertex other than $s$ has an incoming arc.
Moreover $s$ has at least one outgoing arc, since $t$ is reachable from $s\neq t$.

A \emph{configuration} is a vector $\omega\in\{0,1\}^E$, where $\omega_e=1$ means that the arc $e$ is open.
We always draw $\omega$ from the product measure in which every arc $e$ is open independently with probability $p_e$, and $\Pr$ and $\E$ refer to this measure unless stated otherwise.
The open arcs form the random subgraph $G_p$ of the introduction, and we let $\Pr_{G_p}$ also denote this measure.
For vertices $a,b\in V$, let $a\reach b$ be the event that $a$ reaches $b$ along a directed path of open arcs.
The path may have length zero, so $a\reach a$ always holds.
For every vertex $v$ let
\begin{align*}
  q_v\defeq\Pr[s\reach v]
\end{align*}
be its \emph{reliability}.
Then $q_s=1$, $q_t=\Rel(s,t)$, and $q_v>0$ for every $v$, since, by~\eqref{asm:standing}, some path from $s$ to $v$ exists and all of its arcs are open with positive probability.

We use a classical correlation inequality.
An event of the states of finitely many arcs or edges is \emph{increasing} if it remains true when closed ones are opened.
Every reachability or connectivity event is increasing, because opening arcs or edges does not destroy paths.

\begin{lemma}[Harris's inequality~\cite{Harris1960}]\label{lem:harris}
If $\mathcal{A}$ and $\mathcal{B}$ are increasing events of finitely many independent Bernoulli variables, then $\Pr[\mathcal{A}\cap\mathcal{B}]\geq\Pr[\mathcal{A}]\Pr[\mathcal{B}]$.
\end{lemma}

\paragraph{Computational model.}
We count arithmetic operations at unit cost, as in the introduction.
They include comparisons and taking the integer part of a number.
Independent draws of uniform random variables on $(0,1)$ and the usual operations of a random-access machine on integers, pointers and memory cells also have unit cost.
A coin with a known bias is a comparison of a uniform variable with the bias.
Logarithms are used only in the analysis: integers such as $\lceil\log_2 n\rceil$ are computed by repeated doubling.
The notation $\widetilde{O}(\cdot)$ hides factors that are polylogarithmic in $n$, $m$ and $\varepsilon^{-1}$.

\subsection{A DFS process}\label{sec:s-dfs}
We first introduce a DFS process starting from $s$ in $G_p$.
We use $G$ to denote the underlying deterministic graph, and $G_p$ to denote the random graph.
We maintain a set $D$ of discovered vertices and a stack. Initially, $D=\{s\}$ and the stack contains only $s$.
At each step, we pop the top vertex $u$ and examine its outgoing arcs in a fixed order.
For each arc $e=(u,v)$ with $v\notin D$, we reveal whether $e$ is open. It is open with probability $p_e$, independently of the other arcs.
If it is open, we add $v$ to $D$ and push $v$ onto the stack.
After examining all such arcs from $u$, we repeat with the new top vertex, until the stack is empty.
More generally, \Cref{alg:dfs} starts from a vertex $u$ and a discovered set $D$ with $u\in D$, using fresh independent arc states on each call.
The DFS process from $s$ is the case $u=s$ and $D=\{s\}$.

\begin{algorithmfloat}
\begin{algorithm}[$\mathrm{DFS}(u,D)$]\label{alg:dfs}
\leavevmode\par\noindent
\textbf{Input:} a vertex $u$ and a discovered set $D$ with $u\in D\subseteq V$.
\begin{enumerate}[label=\arabic*.,leftmargin=*]
  \item Initialise a stack $S$ containing only $u$.
  \item \textbf{While} $S$ is not empty:
  \begin{enumerate}[label=(\alph*),leftmargin=*]
    \item Pop the top vertex $z$ from $S$.
    \item \textbf{For each} outgoing arc $e=(z,v)$ in $G$, in the fixed order:
    \begin{enumerate}[label=(\roman*),leftmargin=*]
      \item \textbf{If} $v\in D$, skip this arc.
      \item Reveal a fresh independent state for $e$: it is open with probability $p_e$.
      \item \textbf{If} $e$ is open, set $D\gets D\cup\{v\}$ and push $v$ onto $S$.
    \end{enumerate}
  \end{enumerate}
  \item \textbf{Return} $D$.
\end{enumerate}
\end{algorithm}
\end{algorithmfloat}

We next introduce some definitions related to the DFS process.
\begin{definition}[DFS tree]\label{def:dfs-tree}
The \emph{DFS tree} of a run is the tree rooted at its starting vertex whose arcs are precisely the open arcs $(z,y)$ that add $y$ to the discovered set during the run.
For each such arc, $z$ is the \emph{parent} of $y$.
\end{definition}

\begin{definition}[Snapshot]\label{def:snapshot}
A \emph{snapshot} at a vertex $u$ is a pair $(D,u)$ with $u\in D\subseteq V$.
When the DFS pops $u$, the pair $(D,u)$, with $D$ the discovered set at that moment, is the \emph{snapshot at this pop}.
It is taken before any outgoing arc of $u$ is examined.
\end{definition}

Consider a complete run of $\mathrm{DFS}(s,\{s\})$ and fix a vertex $u$.
Each vertex is pushed at most once, and every pushed vertex is eventually popped.
Thus the run produces exactly one snapshot at $u$ if $s\reach u$ in $G_p$, and none otherwise, where  $s\reach u$ happens with probability $q_u$.

\begin{definition}[Conditional snapshot distribution]\label{def:snapshot-distribution}
For a vertex $u$, let $\pi_u$ be the distribution of the discovered set when $\mathrm{DFS}(s,\{s\})$ pops $u$, conditional on $s\reach u$ in $G_p$.
Thus, for $u\in D\subseteq V$,
\begin{align*}
  \pi_u(D)\defeq\Pr\bigl[\text{the snapshot at $u$ is $(D,u)$}\mid s\reach u\text{ in }G_p\bigr].
\end{align*}
\end{definition}

%When $q_u$ is tiny, obtaining a snapshot at $u$ by starting again from $s$ would take many searches.
%Once a snapshot $(D,u)$ is available, however, the search can be continued from it by a fresh call to $\mathrm{DFS}(u,D)$.

\subsection{\texorpdfstring{\Cref{alg:S}}{Algorithm S}}\label{sec:simple-alg}
\Cref{alg:S}, where $S$ stands for \emph{Simple}, maintains a set $R\subseteq V$ of \emph{resolved} vertices, with $s\in R$.
For each $u\in R$, it maintains a number $\widehat q_u$ and a \emph{pool} at $u$. 
The pool is a multi-set of snapshots $(D,u)$ at $u$. The maintained information satisfies the following properties with high probability.
\begin{enumerate}[label=(\roman*),leftmargin=*]
  \item $\widehat q_u$ is the estimate of $q_u$ with a small relative error;
  \item \emph{drawing} a snapshot $(D,u)$ from the pool chooses one entry uniformly at random, and the drawn discovered set $D$ is intended to be distributed approximately according to $\pi_u$.
\end{enumerate}
Initially, $R=\{s\}$, $\widehat q_s=1$, and the pool of $s$ consists of the single snapshot $(\{s\},s)$.
Each round selects a vertex $v\notin R$, constructs $\widehat q_v$ and the pool of $v$, and adds $v$ to $R$.
For simplicity of exposition, we first ignore the approximation errors at the vertices in $R$: we assume that, for every $u\in R$, $\widehat q_u=q_u$ and drawing from its pool gives exactly a set distributed according to $\pi_u$. We will address the approximation errors in the detailed analysis of the algorithm.

Given the current set $R$, we explain how to choose the next vertex $v \notin R$ to resolve.
Consider the outgoing cut of the current resolved set $R$:
\begin{align*}
  \partial^+R\defeq\{(u,w)\in E:u\in R,\ w\notin R\}.
\end{align*}
We choose an arc $(u,v)\in\partial^+R$ that maximises $q_up_{(u,v)}$, and resolve its head $v$ next.

It remains to construct $\widehat q_v$ and the pool of $v$.
Fix a set $D$ with $v\in D\subseteq V$.
By definition, $q_v\pi_v(D)$ is the probability that $\mathrm{DFS}(s,\{s\})$ pops $v$ with discovered set $D$.
We first express $q_v\pi_v(D)$ through the conditional snapshot distributions at the resolved vertices, and then derive our algorithm from this formula. The value of $q_v\pi_v(D)$ can be analysed as follows.
\begin{itemize}
\item If $\mathrm{DFS}(s,\{s\})$ discovers $v$, consider the path from $s$ to $v$ in its DFS tree.
Since $s\in R$ and $v\notin R$, this path has a last vertex $u$ in $R$, and its arc leaving $u$ lies in $\partial^+R$.

\item Fix such a last resolved vertex $u\in R$, and condition on the snapshot at $u$ being $(A,u)$, which has probability $q_u\pi_u(A)$.
As $v$ lies below $u$ in the DFS tree, $v$ is popped before the search returns to the vertices beneath $u$ in the stack.
Until then, the search examines only arcs leaving $u$ and its descendants, whose states are not yet revealed, and it depends on the earlier history only through $A$.
Hence this part of the run has the same law as $\mathrm{DFS}(u,A)$.
Let $\mathcal E_{v,D}$ be the event that $\mathrm{DFS}(u,A)$ pops $v$ with discovered set $D$ and the path from $u$ to $v$ in its DFS tree meets no vertex of $R$ other than $u$.
Then
\begin{align}\label{eq:last-resolved-split}
q_v\pi_v(D)=\sum_{u\in R}q_u\sum_{A\subseteq V:u\in A}\pi_u(A)\Pr\nolimits_{(A,u)}[\mathcal E_{v,D}],
\end{align}
where $\Pr_{(A,u)}$ refers to the probability space of $\mathrm{DFS}(u,A)$.

\item On $\mathcal E_{v,D}$, the first arc of the tree path from $u$ to $v$ is an open arc of $\partial^+R$ leaving $u$, so $\mathcal E_{v,D}\subseteq \mathrm{Out}_u$, where $\mathrm{Out}_u$ is the event that at least one arc of $\partial^+R$ leaving $u$ is open.
Let
\begin{align}\label{eq:gamma}
\gamma_u\defeq\Pr[\mathrm{Out}_u]=1-\prod_{e=(u,w)\in\partial^+R}(1-p_e).
\end{align}
Then,
\begin{align}\label{eq:condition-on-out}
\Pr\nolimits_{(A,u)}[\mathcal E_{v,D}]=\Pr\nolimits_{(A,u)}[\mathrm{Out}_u] \cdot \Pr\nolimits_{(A,u)}[\mathcal E_{v,D}\mid \mathrm{Out}_u] = \gamma_u \cdot \Pr\nolimits_{(A,u)}[\mathcal E_{v,D}\mid \mathrm{Out}_u].
\end{align}
%Under the conditional law, the states of the arcs of $\partial^+R$ leaving $u$ are drawn conditional on $\mathrm{Out}_u$, and all other arcs are drawn independently with their original probabilities.
\end{itemize}
Combining \eqref{eq:last-resolved-split} with \eqref{eq:condition-on-out} and letting $\Phi\defeq\sum_{u\in R}q_u\gamma_u$, we obtain
\begin{align}\label{eq:snapshot-identity}
q_v\pi_v(D)=\Phi\sum_{u\in R}\frac{q_u\gamma_u}{\Phi}\sum_{A\subseteq V:u\in A}\pi_u(A)\Pr\nolimits_{(A,u)}[\mathcal E_{v,D}\mid \mathrm{Out}_u].
\end{align}

The RHS of \eqref{eq:snapshot-identity} suggests \Cref{alg:generate}, which constructs $\widehat q_v$ and the pool of a vertex $v\notin R$.

\begin{algorithmfloat}
\begin{algorithm}[$\mathrm{Generate}(v,M)$]\label{alg:generate}
\leavevmode\par\noindent
\textbf{Input:} a vertex $v\notin R$ and a number $M$ of samples.
\begin{enumerate}[label=\arabic*.,leftmargin=*]
  \item \textbf{For} $i=1,\ldots,M$:
  \begin{enumerate}[label=(\alph*),leftmargin=*]
    \item Choose a vertex $u\in R$ with probability $\widehat q_u\gamma_u/\widehat\Phi$, where $\widehat\Phi=\sum_{u\in R}\widehat q_u\gamma_u$.
    \item Draw a snapshot $(A,u)$ uniformly at random from the pool of $u$.
    \item Run $\mathrm{DFS}(u,A)$ with one change: when $u$ is popped, the states of the arcs of $\partial^+R$ leaving $u$ are drawn jointly from their law conditional on $\mathrm{Out}_u$, that is, conditional on at least one of them being open. All other arcs are revealed as in \Cref{alg:dfs}.
    \item \textbf{If} the run pops $v$, with discovered set $D$, and the path from $u$ to $v$ in its DFS tree meets no vertex of $R$ other than $u$, \emph{record} the snapshot $(D,v)$.
  \end{enumerate}
  \item \textbf{Return} $\widehat q_v=\widehat\Phi N/M$, where $N$ is the number of records, and the list of recorded snapshots as the pool of $v$.
\end{enumerate}
\end{algorithm}
\end{algorithmfloat}

We check that the output of \Cref{alg:generate} has properties (i) and (ii) at $v$ below.
By our assumption, $\widehat q_u=q_u$ for every $u\in R$, so $\widehat\Phi=\Phi$.
Consider one iteration of Step~1, and fix $D$ with $v\in D\subseteq V$.
Step~1(a) chooses $u$ with probability $q_u\gamma_u/\Phi$. By our assumption, the set $A$ drawn in Step~1(b) is then distributed according to $\pi_u$. Given $u$ and $A$, Steps~1(c) and~1(d) record $(D,v)$ exactly on the event $\mathcal E_{v,D}$, under the law $\Pr_{(A,u)}[\,\cdot\mid\mathrm{Out}_u]$.
Hence, by \eqref{eq:snapshot-identity}, the iteration records $(D,v)$ with probability
$
\frac{q_v\pi_v(D)}{\Phi}.
$
Summing over $D$, the iteration records a snapshot with probability $q_v/\Phi$, and conditional on this, the recorded snapshot is $(D,v)$ with probability $\pi_v(D)$.
The $M$ iterations are independent, which gives the two properties.
\begin{enumerate}[label=(\roman*),leftmargin=*]
  \item The number $N$ of records has the binomial distribution with parameters $M$ and $q_v/\Phi$. Hence $\widehat q_v=\widehat\Phi N/M=\Phi N/M$ satisfies $\E[\widehat q_v]=q_v$, and by Chebyshev's inequality, for every $\varepsilon>0$,
  \begin{align*}
  \Pr\bigl[|\widehat q_v-q_v|\geq\varepsilon q_v\bigr]\leq\frac{\Var[\widehat q_v]}{\varepsilon^2q_v^2}=\frac{\Phi^2}{\varepsilon^2M^2q_v^2}\cdot M\frac{q_v}{\Phi}\Bigl(1-\frac{q_v}{\Phi}\Bigr)\leq\frac{\Phi}{\varepsilon^2Mq_v}.
  \end{align*}
  Thus $\widehat q_v$ is within a factor $1\pm\varepsilon$ of $q_v$ with probability at least $1-\delta$ once $M\geq\Phi/(\delta\varepsilon^2q_v)$.
  \item Each recorded snapshot is $(D,v)$ with probability $\pi_v(D)$, that is, each sample in the pool of $v$ is distributed according to $\pi_v$.
\end{enumerate}

The next lemma shows that the choice of $v$ above keeps $\Phi/q_v$ small, so that $M$ is small.

\begin{lemma}\label{lem:large-ratio}
Let $s\in R\neq V$, and let $(u,v)\in\partial^+R$ maximise $q_up_{(u,v)}$ over $\partial^+R$.
Then $\Phi\leq m\,q_v$.
\end{lemma}

\begin{proof}
By the union bound, $\gamma_x\leq\sum_{e=(x,w)\in\partial^+R}p_e$ for every $x\in R$, so
\begin{align*}
\Phi\leq\sum_{e=(x,w)\in\partial^+R}q_xp_e\leq m\,q_up_{(u,v)}.
\end{align*}
If $s$ reaches $u$, it does so along a path without repeated vertices, which uses no arc leaving $u$.
Therefore the event $s\reach u$ depends only on the arcs not leaving $u$, and it is independent of the state of $(u,v)$.
Hence, with probability $q_up_{(u,v)}$, the vertex $s$ reaches $u$ and $(u,v)$ is open, and on this event $s$ reaches $v$.
This gives $q_up_{(u,v)}\leq q_v$, and hence $\Phi\leq m\,q_v$.
\end{proof}

We are now ready to present \Cref{alg:S}. The choice of $M$ is justified in the proof of \Cref{thm:alg-s} in the next subsection.

\begin{algorithmfloat}
% Algorithm S is labelled S instead of a number. The counter is stepped back after it.
\renewcommand{\thealgorithm}{S}\def\theHalgorithm{S}%
\begin{algorithm}\label{alg:S}
\leavevmode\par\noindent
\textbf{Input:} a graph $G=(V,E)$ with open probabilities $\boldsymbol p$,  $s\neq t$, and  $\varepsilon,\delta\in(0,1)$.
\begin{enumerate}[label=\arabic*.,leftmargin=*]
  \item Set $M=\lceil 2n^3m/(\varepsilon^2\delta)\rceil$, $R=\{s\}$ and $\widehat q_s=1$, and let the pool of $s$ be the single snapshot $(\{s\},s)$.
  \item \textbf{While} $t\notin R$:
  \begin{enumerate}[label=(\alph*),leftmargin=*]
    \item Choose an arc $(u,v)\in\partial^+R$ that maximises $\widehat q_up_{(u,v)}$, breaking ties by a fixed order. \textbf{If} every such product is zero, \textbf{return} $0$.
    \item Run $\mathrm{Generate}(v,M)$, and install its output as $\widehat q_v$ and the pool of $v$.
    \item Set $R\gets R\cup\{v\}$.
  \end{enumerate}
  \item \textbf{Return} $\widehat q_t$.
\end{enumerate}
\end{algorithm}
\addtocounter{algorithm}{-1}%
\end{algorithmfloat}

\Cref{alg:S} is well defined even though a pool may be empty.
An empty pool at $v$ arises only when $\mathrm{Generate}(v,M)$ makes no record, in which case $\widehat q_v=0$.
Whenever $\mathrm{Generate}$ is run, Step~2(a) has found $\widehat q_up_{(u,v)}>0$, so $\widehat\Phi>0$, and Step~1(a) of \Cref{alg:generate} chooses only vertices $x$ with $\widehat q_x>0$, whose pools are nonempty.

\subsection{Analysis of \texorpdfstring{\Cref{alg:S}}{Algorithm S}}\label{sec:s-analysis}
To explain the idea of \Cref{alg:S}, the previous subsection assumed that, for every resolved vertex $u$, $\widehat q_u=q_u$ and drawing from the pool of $u$ gives exactly a set distributed according to $\pi_u$.
In this subsection we drop this assumption and prove the following, where \Cref{alg:S} uses $M=\lceil 2n^3m/(\varepsilon^2\delta)\rceil$.

\begin{theorem}\label{thm:alg-s}
Given a directed graph $G$ with $n$ vertices and $m$ arcs, open probabilities $\boldsymbol p$, vertices $s\neq t$, and $\varepsilon,\delta\in(0,1)$, \Cref{alg:S} runs in time $O(n^4m^2/(\varepsilon^2\delta))$ and outputs an unbiased estimate of $q_t$ that is within a factor $1\pm\varepsilon$ of $q_t$ with probability at least $1-\delta$.
\end{theorem}

We remark that the analysis for \Cref{alg:S} can actually be tightened to $O(n^3m^2/(\varepsilon^2\delta))$ by using some heavier martingale machinery.
However, since our goal here is to give a simple analysis and there will be a faster algorithm in any case, we do not pursue improvement.

The idea is to follow, for every vertex $w$, the estimate of $q_w$ predicted by the current pools.
It is a martingale over the rounds, it starts at $q_w$, and for $w=t$ it ends at the output.
Thus the output is unbiased, and it remains to bound the variance of this martingale.

\paragraph{Predictions.}
Number the rounds of the while loop of \Cref{alg:S} by $r=1,2,\ldots$.
Once the algorithm has returned, we let every further round change nothing, so that round $r$ is defined for every $r\geq1$.
Let $R_r$ be the resolved set at the start of round $r$.
A round changes something only if it runs $\mathrm{Generate}$, which resolves a new vertex other than $s$, so at most $n-1$ rounds do.
In particular, $R_n$ and the pools at the start of round $n$ are the final ones.

For $x\in R_r$, let $\widehat\pi_x$ be the empirical distribution of the pool of $x$:
\begin{align*}
\widehat\pi_x(A)\defeq\frac{\text{number of entries of the pool equal to }(A,x)}{\text{number of entries of the pool}},
\end{align*}
and $\widehat\pi_x\defeq0$ if the pool is empty.
Drawing from the pool of $x$ gives $A$ with probability $\widehat\pi_x(A)$, so Steps~1(a) and~1(b) of \Cref{alg:generate} choose the pair $(x,A)$ with probability $\gamma_x\widehat q_x\widehat\pi_x(A)/\widehat\Phi$. As noted after \Cref{alg:S}, this is also correct, namely $0$, when the pool of $x$ is empty.

For a set $R\ni s$, a vertex $x\in R$, a set $A$ with $x\in A\subseteq V$, and a vertex $w$, let
\begin{align}\label{eq:reach-S}
  \mathrm{Reach}^{\mathrm{S}}_R(x,A;w)\defeq\Pr\nolimits_{(A,x)}\bigl[&\text{$\mathrm{DFS}(x,A)$ discovers $w$, and the path from $x$ to $w$}\notag\\
  &\text{in its DFS tree meets no vertex of $R$ other than $x$}\bigr].
\end{align}
Here and below, the superscript $\mathrm{S}$ indicates a quantity for \Cref{alg:S}, to distinguish it from the analogous quantity for the main algorithm.
We set $\mathrm{Reach}^{\mathrm{S}}_R(x,A;x)\defeq1$.
In particular, $\mathrm{Reach}^{\mathrm{S}}_R(x,A;w)=0$ for $w\in R\setminus\{x\}$.
The \emph{prediction} of $q_w$ at round $r$ is
\begin{align}\label{eq:prediction}
Q^{\mathrm{S}}_r(w)\defeq\sum_{x\in R_r}\widehat q_x\sum_{A\subseteq V:x\in A}\widehat\pi_x(A)\,\mathrm{Reach}^{\mathrm{S}}_{R_r}(x,A;w).
\end{align}
Intuitively, $Q^{\mathrm{S}}_r(w)$ is the estimate of $q_w$ obtained by taking the current pools as given and computing the rest exactly: as in \eqref{eq:last-resolved-split}, it partitions all the cases that $\mathrm{DFS}(s,\{s\})$ discovers $w$ according to the last resolved vertex $x$ on the tree path to $w$. The RHS of the above equation equals $q_w$ when $\widehat q_x=q_x$ and $\widehat\pi_x=\pi_x$ for every $x\in R_r$.
We record three facts about $Q^{\mathrm{S}}_r(w)$.
\begin{itemize}
\item For any vertex $w$, $Q^{\mathrm{S}}_1(w)=q_w$. This is because $R_1=\{s\}$, $\widehat q_s=1$ and $\widehat\pi_s$ puts mass one on $\{s\}$, so $Q^{\mathrm{S}}_1(w)$ is the probability that $\mathrm{DFS}(s,\{s\})$ discovers $w$.
\item $Q^{\mathrm{S}}_r(w)=\widehat q_w$ for $w\in R_r$, since then only $x=w$ contributes to the RHS of \eqref{eq:prediction}.
\item The output equals $Q^{\mathrm{S}}_n(t)$.
If the algorithm returns $\widehat q_t$, then $t\in R_n$, and $Q^{\mathrm{S}}_n(t)=\widehat q_t$ by the previous fact.
Otherwise, it returns $0$ in Step~2(a) of some round $r_0$. No round from $r_0$ on changes anything, so $Q^{\mathrm{S}}_n(t)=Q^{\mathrm{S}}_{r_0}(t)$.
In round $r_0$, every $x\in R_{r_0}$ with an arc in $\partial^+R_{r_0}$ has $\widehat q_x=0$, as $p_e>0$ by~\eqref{asm:standing}. For every other $x\in R_{r_0}$, $\mathrm{Reach}^{\mathrm{S}}_{R_{r_0}}(x,A;t)=0$, because a path from $x$ to $t\notin R_{r_0}$ meeting $R_{r_0}$ only at $x$ starts with an arc of $\partial^+R_{r_0}$ leaving $x$. Hence $Q^{\mathrm{S}}_{r_0}(t)=0$.
\end{itemize}

The following lemma shows that the predictions form a martingale.

\begin{lemma}\label{lem:s-martingale}
Let $\mathcal F_{<r}$ be the randomness used before round $r$. It determines $R_r$ and the numbers $\widehat q_x$ and pools of all $x\in R_r$, and hence $Q^{\mathrm{S}}_r(w)$.
For every round $r$ and every vertex $w$, 
\begin{align*}
  \E[Q^{\mathrm{S}}_{r+1}(w)\mid\mathcal F_{<r}]=Q^{\mathrm{S}}_r(w).
\end{align*}
Consequently, it holds that $\E[\widehat q_t]=q_t$.
\end{lemma}

\begin{proof}
Suppose that round $r$ runs $\mathrm{Generate}(v,M)$, as otherwise there is nothing to prove, so that $R_{r+1}=R_r\cup\{v\}$.
The vertex $v$ and the number $\widehat\Phi$ are determined by $\mathcal F_{<r}$.
For $x\in R_r$ and $A\ni x$, recall that $\mathcal E_{v,D}$ is the event that $\mathrm{DFS}(x,A)$ pops $v$ with discovered set $D$ and the path from $x$ to $v$ in its DFS tree meets no vertex of $R_r$ other than $x$.
Let $N$ be the number of records of $\mathrm{Generate}(v,M)$, and let $N_D$ be the number of records equal to $(D,v)$.
Then $\widehat q_v=\widehat\Phi N/M$ and $\widehat\pi_v(D)=N_D/N$, so $\widehat q_v\widehat\pi_v(D)=\widehat\Phi N_D/M$. This also holds when $N=0$, as both sides vanish.
Given $\mathcal F_{<r}$, the $M$ iterations are independent.
In each iteration, Steps~1(a) and~1(b) choose $(x,A)$ with probability $\gamma_x\widehat q_x\widehat\pi_x(A)/\widehat\Phi$, and Steps~1(c) and~1(d) then record $(D,v)$ with probability $\Pr_{(A,x)}[\mathcal E_{v,D}\mid\mathrm{Out}_x]$.
Hence $\E[N_D\mid\mathcal F_{<r}]$ is $M$ times the probability that one iteration records $(D,v)$, and by \eqref{eq:condition-on-out}, we have the following equality
\begin{align}\label{eq:pool-mean}
\E[\widehat q_v\widehat\pi_v(D)\mid\mathcal F_{<r}]
&=\frac{\widehat\Phi}{M}\,\E[N_D\mid\mathcal F_{<r}]
=\widehat\Phi\sum_{x\in R_r}\sum_{A}\frac{\gamma_x\widehat q_x\widehat\pi_x(A)}{\widehat\Phi}\Pr\nolimits_{(A,x)}[\mathcal E_{v,D}\mid\mathrm{Out}_x]\notag\\
&=\sum_{x\in R_r}\widehat q_x\sum_{A}\widehat\pi_x(A)\Pr\nolimits_{(A,x)}[\mathcal E_{v,D}].
\end{align}
Next, we split the event defining $\mathrm{Reach}^{\mathrm{S}}_{R_r}(x,A;w)$ in~\eqref{eq:reach-S} according to whether the tree path from $x$ to $w$ passes through $v$.
If it does not, the event is the one defining $\mathrm{Reach}^{\mathrm{S}}_{R_{r+1}}(x,A;w)$.
If it does, the run pops $v$ with some discovered set $D$, the part of the path up to $v$ gives $\mathcal E_{v,D}$, and, as in the derivation of \eqref{eq:last-resolved-split}, the search below $v$ has the law of $\mathrm{DFS}(v,D)$, in which the rest of the path must meet no vertex of $R_{r+1}$ other than $v$.
Hence
\begin{align}\label{eq:h-split}
\mathrm{Reach}^{\mathrm{S}}_{R_r}(x,A;w)=\mathrm{Reach}^{\mathrm{S}}_{R_{r+1}}(x,A;w)+\sum_{D}\Pr\nolimits_{(A,x)}[\mathcal E_{v,D}]\,\mathrm{Reach}^{\mathrm{S}}_{R_{r+1}}(v,D;w).
\end{align}
Now $Q^{\mathrm{S}}_{r+1}(w)$ is given by \eqref{eq:prediction} with the sum over $x\in R_{r+1}$ and with $\mathrm{Reach}^{\mathrm{S}}_{R_{r+1}}$ in place of $\mathrm{Reach}^{\mathrm{S}}_{R_r}$.
Round $r$ does not change $\widehat q_x$ or the pool of any $x\in R_r$, so $\widehat q_x$ and $\widehat\pi_x$ are the same as at round $r$ for $x\in R_r$, while $\widehat q_v$ and $\widehat\pi_v$ come from $\mathrm{Generate}(v,M)$.
Splitting the sum into the terms with $x\in R_r$ and the term with $x=v$ gives by definition in~\eqref{eq:prediction},
\begin{align}\label{eq:Q-next}
Q^{\mathrm{S}}_{r+1}(w)=\sum_{x\in R_r}\widehat q_x\sum_A\widehat\pi_x(A)\,\mathrm{Reach}^{\mathrm{S}}_{R_{r+1}}(x,A;w)+\widehat q_v\sum_D\widehat\pi_v(D)\,\mathrm{Reach}^{\mathrm{S}}_{R_{r+1}}(v,D;w).
\end{align}
Since $R_{r+1}$ is determined by $\mathcal F_{<r}$, the only random factors on the RHS of \eqref{eq:Q-next} are the numbers $\widehat q_v\widehat\pi_v(D)$.
Taking the conditional expectation and using \eqref{eq:pool-mean} and then \eqref{eq:h-split}, we obtain
\begin{align*}
\E[Q^{\mathrm{S}}_{r+1}(w)\mid\mathcal F_{<r}]
&=\sum_{x\in R_r}\widehat q_x\sum_A\widehat\pi_x(A)\,\mathrm{Reach}^{\mathrm{S}}_{R_{r+1}}(x,A;w)+\sum_D\E[\widehat q_v\widehat\pi_v(D)\mid\mathcal F_{<r}]\,\mathrm{Reach}^{\mathrm{S}}_{R_{r+1}}(v,D;w)\\
&=\sum_{x\in R_r}\widehat q_x\sum_A\widehat\pi_x(A)\Bigl(\mathrm{Reach}^{\mathrm{S}}_{R_{r+1}}(x,A;w)+\sum_D\Pr\nolimits_{(A,x)}[\mathcal E_{v,D}]\,\mathrm{Reach}^{\mathrm{S}}_{R_{r+1}}(v,D;w)\Bigr)\\
&=\sum_{x\in R_r}\widehat q_x\sum_A\widehat\pi_x(A)\,\mathrm{Reach}^{\mathrm{S}}_{R_r}(x,A;w)=Q^{\mathrm{S}}_r(w).
\end{align*}
Finally, $\E[\widehat q_t]=\E[Q^{\mathrm{S}}_n(t)]=Q^{\mathrm{S}}_1(t)=q_t$.
\end{proof}

We bound the conditional variance of the predictions in one round.
Let $\Psi_r\defeq\sum_{w\in V}Q^{\mathrm{S}}_r(w)/q_w$. By \Cref{lem:s-martingale}, $\Psi_r$ is also a martingale, and $\Psi_1=n$.

\begin{lemma}\label{lem:one-round}
Suppose that round $r$ runs $\mathrm{Generate}(v,M)$, where $(u,v)$ is the arc chosen in Step~2(a). Then, for every vertex $w$,
\begin{align*}
\Var[Q^{\mathrm{S}}_{r+1}(w)\mid\mathcal F_{<r}]\leq\frac mM\cdot\frac{\widehat q_u}{q_u}\cdot q_wQ^{\mathrm{S}}_r(w)
\quad\text{and}\quad
\Var[\Psi_{r+1}\mid\mathcal F_{<r}]\leq\frac{nm}M\cdot\frac{\widehat q_u}{q_u}\cdot\Psi_r.
\end{align*}
\end{lemma}

\begin{proof}
Recall from the proof of \Cref{lem:s-martingale} that $R_{r+1}=R_r\cup\{v\}$, that $N_D$ is the number of records equal to $(D,v)$, and that $\widehat q_v\widehat\pi_v(D)=\widehat\Phi N_D/M$.
We first express $Q^{\mathrm{S}}_{r+1}(w)$ as a constant plus a sum of independent terms.
By \eqref{eq:Q-next},
\begin{align*}
Q^{\mathrm{S}}_{r+1}(w)&=Q^{\mathrm{S}}_{\mathrm{old}}(w)+\frac{\widehat\Phi}{M}\sum_DN_D\,\mathrm{Reach}^{\mathrm{S}}_{R_{r+1}}(v,D;w),
\quad\text{where}\\
Q^{\mathrm{S}}_{\mathrm{old}}(w)&\defeq\sum_{x\in R_r}\widehat q_x\sum_A\widehat\pi_x(A)\,\mathrm{Reach}^{\mathrm{S}}_{R_{r+1}}(x,A;w)\geq0.
\end{align*}
Note that $Q^{\mathrm{S}}_{\mathrm{old}}(w)$ is determined by $\mathcal F_{<r}$.
For $i=1,\ldots,M$, let $X_i(w)\defeq \mathrm{Reach}^{\mathrm{S}}_{R_{r+1}}(v,D;w)$ if the $i$-th iteration of $\mathrm{Generate}(v,M)$ records $(D,v)$, and $X_i(w)\defeq0$ if it records nothing.
Each iteration records at most one snapshot, so $\sum_DN_D\,\mathrm{Reach}^{\mathrm{S}}_{R_{r+1}}(v,D;w)=\sum_{i=1}^MX_i(w)$, and hence
\begin{align*}
Q^{\mathrm{S}}_{r+1}(w)=Q^{\mathrm{S}}_{\mathrm{old}}(w)+\frac{\widehat\Phi}{M}\sum_{i=1}^MX_i(w).
\end{align*}
Given $\mathcal F_{<r}$, the $M$ iterations are independent and run the same procedure, so $X_1(w),\ldots,X_M(w)$ are independent and identically distributed.
Taking the conditional expectation and using \Cref{lem:s-martingale}, we get $Q^{\mathrm{S}}_{\mathrm{old}}(w)+\widehat\Phi\,\E[X_1(w)\mid\mathcal F_{<r}]=Q^{\mathrm{S}}_r(w)$, so $\widehat\Phi\,\E[X_1(w)\mid\mathcal F_{<r}]\leq Q^{\mathrm{S}}_r(w)$.

We use two bounds.
\begin{itemize}
\item $X_i(w)\leq q_w/q_v$.
Indeed, if $\mathrm{DFS}(v,D)$ discovers $w$, its tree path is an open path from $v$ to $w$, so $X_i(w)\leq\Pr_{G_p}[v\reach w]$.
The events $s\reach v$ and $v\reach w$ are increasing, so by Harris's inequality (\Cref{lem:harris}), $q_w\geq\Pr[s\reach v\text{ and }v\reach w]\geq q_v\Pr[v\reach w]$.
\item $\widehat\Phi\leq m\,\widehat q_uq_v/q_u$.
Indeed, as in the proof of \Cref{lem:large-ratio}, $\widehat\Phi\leq\sum_{e=(x,y)\in\partial^+R_r}\widehat q_xp_e\leq m\,\widehat q_up_{(u,v)}$ by the choice of $(u,v)$, and the same proof shows $q_up_{(u,v)}\leq q_v$.
\end{itemize}
Since $0\leq X_1(w)\leq q_w/q_v$, we have $\Var[X_1(w)\mid\mathcal F_{<r}]\leq\E[X_1(w)^2\mid\mathcal F_{<r}]\leq(q_w/q_v)\,\E[X_1(w)\mid\mathcal F_{<r}]$.
Hence, by independence,
\begin{align*}
\Var[Q^{\mathrm{S}}_{r+1}(w)\mid\mathcal F_{<r}]&=\frac{\widehat\Phi^2}{M^2}\sum_{i=1}^M\Var[X_i(w)\mid\mathcal F_{<r}]=\frac{\widehat\Phi^2}{M}\Var[X_1(w)\mid\mathcal F_{<r}]
\leq\frac{\widehat\Phi^2}{M}\cdot\frac{q_w}{q_v}\,\E[X_1(w)\mid\mathcal F_{<r}]\\
&\leq\frac{\widehat\Phi}{Mq_v}\,q_wQ^{\mathrm{S}}_r(w)
\leq\frac mM\cdot\frac{\widehat q_u}{q_u}\cdot q_wQ^{\mathrm{S}}_r(w),
\end{align*}
where the last two steps use $\widehat\Phi\,\E[X_1(w)\mid\mathcal F_{<r}]\leq Q^{\mathrm{S}}_r(w)$ and $\widehat\Phi\leq m\,\widehat q_uq_v/q_u$.
The same argument applies to $\Psi_{r+1}$, with $X_i(w)$ replaced by $\widetilde X_i\defeq\sum_{w\in V}X_i(w)/q_w$.
Since $X_i(w)=0$ for $w\in R_r$, we have $\widetilde X_i\leq n/q_v$, and $\widehat\Phi\,\E[\widetilde X_1\mid\mathcal F_{<r}]\leq\Psi_r$. This gives the second bound.
\end{proof}

We prove the concentration part of \Cref{thm:alg-s}.
Unbiasedness is \Cref{lem:s-martingale}.
All predictions are nonnegative, so $\widehat q_u/q_u=Q^{\mathrm{S}}_r(u)/q_u\leq\Psi_r$ and $Q^{\mathrm{S}}_r(w)\leq q_w\Psi_r$.
Therefore \Cref{lem:one-round} gives, for every round $r$ and every fixed outcome of $\mathcal F_{<r}$,
\begin{align*}
\Var[Q^{\mathrm{S}}_{r+1}(t)\mid\mathcal F_{<r}]\leq\frac mM\,q_t^2\Psi_r^2
\quad\text{and}\quad
\Var[\Psi_{r+1}\mid\mathcal F_{<r}]\leq\frac{nm}M\,\Psi_r^2.
\end{align*}
Both hold trivially if round $r$ changes nothing.
Since $\Psi_r$ is a martingale, $\E[\Psi_{r+1}\mid\mathcal F_{<r}]=\Psi_r$, and hence
\begin{align*}
\E[\Psi_{r+1}^2\mid\mathcal F_{<r}]=\Psi_r^2+\Var[\Psi_{r+1}\mid\mathcal F_{<r}]\leq\Bigl(1+\frac{nm}M\Bigr)\Psi_r^2.
\end{align*}
Taking expectations over $\mathcal F_{<r}$ gives $\E[\Psi_{r+1}^2]\leq(1+nm/M)\,\E[\Psi_r^2]$.
As $\Psi_1=n$ and $M\geq2n^3m$, for every $r\leq n$,
\begin{align*}
\E[\Psi_r^2]\leq n^2\Bigl(1+\frac{nm}M\Bigr)^{n-1}\leq n^2e^{n^2m/M}\leq n^2e^{1/2}\leq2n^2.
\end{align*}
Let $\Delta_r\defeq Q^{\mathrm{S}}_{r+1}(t)-Q^{\mathrm{S}}_r(t)$, so that $\widehat q_t-q_t=Q^{\mathrm{S}}_n(t)-Q^{\mathrm{S}}_1(t)=\sum_{r=1}^{n-1}\Delta_r$.
By \Cref{lem:s-martingale}, $\E[\Delta_r\mid\mathcal F_{<r}]=0$, and hence $\E[\Delta_r^2\mid\mathcal F_{<r}]=\Var[Q^{\mathrm{S}}_{r+1}(t)\mid\mathcal F_{<r}]$.
For $r<r'$, the increment $\Delta_r$ is determined by $\mathcal F_{<r'}$, so $\E[\Delta_r\Delta_{r'}]=\E\bigl[\Delta_r\,\E[\Delta_{r'}\mid\mathcal F_{<r'}]\bigr]=0$.
Since $\E[\widehat q_t]=q_t$, expanding the square of $\sum_r\Delta_r$ gives
\begin{align*}
\Var[\widehat q_t]&=\E\biggl[\Bigl(\sum_{r=1}^{n-1}\Delta_r\Bigr)^2\biggr]=\sum_{r=1}^{n-1}\E[\Delta_r^2]=\sum_{r=1}^{n-1}\E\bigl[\Var[Q^{\mathrm{S}}_{r+1}(t)\mid\mathcal F_{<r}]\bigr]\\
&\leq\frac mM\,q_t^2\sum_{r=1}^{n-1}\E[\Psi_r^2]\leq\frac{2n^3m}{M}\,q_t^2\leq\varepsilon^2\delta\,q_t^2.
\end{align*}
By Chebyshev's inequality, $\Pr[|\widehat q_t-q_t|\geq\varepsilon q_t]\leq\Var[\widehat q_t]/(\varepsilon^2q_t^2)\leq\delta$.

We count arithmetic operations at unit cost.
By the standing assumption~\eqref{asm:standing}, every vertex is reachable from $s$, so $n\leq m+1$.
There are at most $n-1$ rounds.
In a round, computing the numbers $\gamma_u$, the number $\widehat\Phi$ and the arc of Step~2(a) takes $O(m)$ time.
Each of the $M$ iterations of \Cref{alg:generate} takes $O(m)$ time. Choosing $u$ and copying $A$ take $O(n)$ time. The states of the arcs $e_1,\ldots,e_k$ of $\partial^+R$ leaving $u$ can be drawn conditional on $\mathrm{Out}_u$ in $O(k)$ time, by first choosing the first open arc, which is $e_j$ with probability $p_{e_j}\prod_{i<j}(1-p_{e_i})/\gamma_u$, and then drawing the states of $e_{j+1},\ldots,e_k$ independently. The search examines every arc at most once, and a record stores $D$ in $O(n)$ time.
Hence the running time is $O(nMm)=O(n^4m^2/(\varepsilon^2\delta))$.
\qed

\section{The main algorithm}\label{sec:algorithm}

In this section we describe the algorithm of \Cref{thm:main}, which we call the \emph{main algorithm}.
\Cref{sec:overview-alg} explains the changes comparing to \Cref{alg:S}.
\Cref{sec:dfs,sec:runs} describe its building blocks: the deferred exploration, and the pilot and resolving runs.
\Cref{sec:alg} states the full algorithm, and \Cref{sec:two-lemmas} reduces \Cref{thm:main,thm:expdeg} to two lemmas, on the accuracy and on the runtime, which \Cref{sec:searches,sec:accuracy,sec:work} prove.
\Cref{sec:roadmap} explains how these proofs are organised.

\subsection{From \texorpdfstring{\Cref{alg:S}}{Algorithm S} to the main algorithm}\label{sec:overview-alg}

The main algorithm keeps some elements from \Cref{alg:S}.
It still resolves vertices gradually through the graph, uses pools of snapshots to restart graph exploration, and guarantees correctness by (a different version of) predictions.
Many other details are different. Here we highlight some key changes.

\begin{enumerate}
  \item \emph{Resolving a layer.} 
    \Cref{alg:S} resolves vertices gradually to avoid resolving vertices with low reliability, which cost too much time.
    However, resolving one vertex each round is also wasteful, as there might be many vertices with sufficiently high reliability.
    In the main algorithm, we resolve a layer of vertices instead.
    These vertices should have high reliability, so to choose them, we do a pilot run first.
    We set a target $\theta=\Theta(\log(n+m))$ ($\theta$ is defined in \eqref{eqn:parameters}), and draw samples until some vertices have been scanned $\theta$ times.
    Then we put all vertices that have been scanned at least $\theta/16$ times into the layer to be resolved next.
    This way, we save a factor $m$ in the number of samples.
  \item \emph{The deferred exploration.} 
    \Cref{alg:S} uses a standard DFS in \Cref{alg:dfs} to do the graph exploration.
    The downside of this is that it would keep exploring vertices that have already been resolved.
    In the main algorithm, we skip the resolved vertices. 
    This change invalidates the previous prediction proof, so it is adjusted to a quantity defined by reachability instead.
    Also, to make the new prediction a martingale, we defer every discovered vertex of the layer until no other discovered vertex is waiting.
  \item \emph{Implementation.}
    In \Cref{alg:S}, every sample copies the discovered set of its snapshot and examines the arcs leaving every vertex it scans, so it may cost $O(m)$ operations.
    In the main algorithm, discovered sets are stored as persistent arrays, so a sample does not copy the set of its snapshot.
    Moreover, we use segment trees over the outgoing arcs to draw the open arcs without examining the closed ones.
    Up to a logarithmic factor, a sample then costs only as many operations as the open arcs that it returns.
  \item \emph{The analysis.} Instead of Chebyshev's inequality, we use Freedman's inequality, a Bernstein-type inequality for martingales, applied to the samples one at a time.
    It bounds the probability that a prediction is inaccurate by $e^{-\Omega(\rho/n)}$, where $\rho$ (defined in \eqref{eqn:parameters}) sets the number of samples of the resolving runs, 
    whereas Chebyshev's inequality gives only $O(n/\rho)$.
    A union bound over the vertices then keeps every prediction accurate with high probability, at the cost of only a logarithmic factor in the number of samples.
\end{enumerate}

In the rest of this section, we describe the building blocks of the main algorithm in more detail, roughly in the order in which those of \Cref{alg:S} were described: the deferred exploration (\Cref{sec:dfs}), the pilot runs and resolving runs (\Cref{sec:runs}), and the full algorithm (\Cref{sec:alg}).

\subsection{The deferred exploration}\label{sec:dfs}

This subsection describes the deferred exploration, and how we generate samples, restarting from snapshots.
Some setups are the same as before.
Snapshots are as in \Cref{def:snapshot}.
%The main algorithm explores the random graph with the \emph{deferred search}, instead of the DFS process of \Cref{sec:s-dfs}.
%Like that process, it scans all arcs from a vertex at once, but with an important change that resolved vertices are skipped.
%As explained in the last subsection, in each round we resolve a layer of vertices, instead of one vertex.
%We defer the vertices of the layer: a discovered vertex of the layer is processed only when no discovered vertex outside the layer is waiting.
For every vertex $u$, fix a \emph{scan order} of $u$ as in \Cref{alg:dfs}.
%\emph{Scanning} $u$ means examining, in this order, the arcs leaving $u$ whose heads are neither discovered nor resolved, as specified below.
For a set $R\subseteq V$, we call the arcs of $\partial^+R$ the \emph{cut arcs} of $R$, and we let $\gamma_u(R)$ denote the probability $\gamma_u$ in~\eqref{eq:gamma} when we need to indicate $R$.
Thus $\gamma_u(R)=0$ if $u$ has no cut arc, as the product is then empty, and $\gamma_u(R)>0$ otherwise, as every open probability is positive by~\eqref{asm:standing}.
The union bound gives $\gamma_u(R)\leq\sum_ep_e$, where $e$ ranges over the cut arcs leaving $u$.

The main subroutine here is called Sample, described in \Cref{alg:sample}.
A sample restarts the deferred exploration from a snapshot $(A,u)$ at a resolved vertex $u$, with a resolved set $R\ni u$ and a layer $L\subseteq V\setminus R$ as inputs.
It maintains a set $D$ of \emph{discovered} vertices, initially $A$, and two stacks: the first holds discovered vertices outside $R\cup L$, and the second holds discovered vertices of $L$.
Note that the first scan at $u$ is conditioned on at least one cut arc leaving $u$ being open. This small optimisation is necessary for the run-time analysis later.
After scanning $u$, it repeatedly \emph{pops} the top vertex of the first stack if that stack is not empty, and otherwise the top vertex of the second stack, and scans the popped vertex. It ends when both stacks are empty.
Thus the exploration \emph{defers} the vertices of $L$: a discovered vertex of $L$ waits on the second stack until the first stack is empty.

As explained in \Cref{sec:overview-alg}, the main algorithm will do a pilot run first to determine the layer $L$, and then a resolving run to estimate reliabilities for the vertices in $L$ and generate their pools.
In the pilot run, the layer is empty: the second stack is never used, nothing is recorded, and the pilot run only counts the unresolved vertices that each sample scans.
In the resolving run, the layer $L$ is determined and given as an input to \Cref{alg:sample}, its vertices are deferred, and the snapshots at their pops are recorded and added to the new pools.
The deferral of the layer is important to make the later last-exit decomposition correct (\Cref{lem:last-exit}(b)).

\begin{algorithmfloat}
\begin{algorithm}[Sample]\label{alg:sample}
\leavevmode\par\noindent
\textbf{Input:} a resolved set $R$, a layer $L\subseteq V\setminus R$, possibly empty, and a snapshot $(A,u)$ with $u\in R$ and $\gamma_u(R)>0$.

\noindent
\textbf{Scan.} The sample keeps a discovered set $D$, a first stack $S_1$ and a second stack $S_2$.
To \emph{scan} a vertex $z$ means: \textbf{for each} arc $f=(z,z')$ leaving $z$, in scan order,
\begin{enumerate}[label=(\roman*),leftmargin=*]
  \item \textbf{if} $z'\in D\cup R$, skip this arc;
  \item otherwise read the state of $f$, and \textbf{if} $f$ is open, set $D\gets D\cup\{z'\}$ and push $z'$ onto $S_2$ if $z'\in L$, and onto $S_1$ otherwise.
\end{enumerate}
\begin{enumerate}[label=\arabic*.,ref=\arabic*,leftmargin=*]
  \item Set $D\gets A$, and let $S_1$ and $S_2$ be empty.
  \item\label{step:first-scan} \emph{First scan.} Draw the states of all cut arcs leaving $u$ jointly, from their product law conditional on at least one of them being open. Scan $u$, reading these states. (The cut arcs whose heads lie in $A$ are skipped by the scan, so their states are drawn but not used.)
  \item\label{step:loop} \emph{Loop.} \textbf{While} $S_1$ or $S_2$ is not empty:
  \begin{enumerate}[label=(\alph*),ref=\theenumi(\alph*)]
    \item\label{step:pop} \textbf{If} $S_1$ is not empty, pop its top vertex $z$. Otherwise, pop the top vertex $z$ of $S_2$.
    \item\label{step:record} \textbf{If} $z\in L$, \emph{record} the snapshot $(D,z)$.
    \item\label{step:scan} Scan $z$, reading for every arc a fresh independent state, open with its own probability.
  \end{enumerate}
  \item \textbf{Return} the recorded snapshots.
\end{enumerate}
\end{algorithm}
\end{algorithmfloat}

A vertex is pushed only when it is added to $D$, so every vertex is pushed at most once and scanned at most once, and every arc is examined at most once.
Every pushed vertex is eventually popped: there are at most $n$ pushes, and every iteration of the loop removes one stack entry.
Within the first stack the exploration is depth-first, and it finishes all the vertices of this stack before popping any vertex of $L$.
In particular, after a vertex $v\in L$ is popped, the sample scans $v$ and all vertices placed on the first stack before another vertex of $L$ is popped, while newly discovered vertices of $L$ wait on the second stack.
This order will make the snapshot at $v$ describe precisely the new contribution supplied by its pool (\Cref{lem:restarting} and \Cref{lem:last-exit}(b)).

A snapshot $(D,z)$ recorded in Step~\ref{step:record} at a pop of $z$ is a \emph{record} at $z$, and $z$ is then a \emph{recorded} vertex of the sample.
Only vertices of the fixed layer $L$ are recorded.
The other unresolved vertices are explored first and may add vertices to the discovered set before a record is made.
Resolved vertices are never pushed or scanned, except for the starting vertex $u$ in the first scan.
Every newly discovered vertex is placed on one of the two stacks exactly once.
In particular, a vertex of $L$ is recorded when it is popped from the second stack, not when it is discovered and placed there.
Every discovered vertex of $L$ is eventually popped.
The order matters for the snapshots: a later record includes all discoveries made before its pop.
\Cref{fig:sample} is an example.

\begin{remark}[Comparison with \Cref{alg:S}]
In \Cref{alg:S}, the DFS may explore resolved vertices, but \Cref{alg:generate} records its target only if the path from the start $u$ to the target in the DFS tree meets no vertex of $R$ other than $u$.
Here the exploration skips all other vertices of $R$, so every newly discovered vertex of $L$ is recorded at its pop without any further test on its tree path.
The set $R$ remains fixed throughout the sample: recording a vertex of $L$ does not add it to $R$, and vertices of $L$ reached through it may also be recorded.
Moreover, one sample can record several vertices of $L$, whose scans are deferred until the first stack is empty.
These changes alter the distribution of the snapshots. Hence, a different analysis will be used for this algorithm.
\end{remark}

%\paragraph{Snapshots and restarted searches.}
%When the deferred search pops a vertex $z$ while its discovered set is $D$, we call $(D,z)$ the \emph{snapshot at this pop}.
%It is taken when $z$ is popped, not when $z$ is discovered.
%In between, the search may discover further vertices. In particular, a vertex of $L$ may wait while the first stack is processed.
%The deferred search \emph{restarted} from a snapshot $(A,u)$, with a resolved set $R$ containing $u$ and a layer $L\subseteq V\setminus R$, is \Cref{alg:sample} without the conditioning in its first scan and without records: it starts with discovered set $A$ and two empty stacks, scans $u$ first, and then follows the loop of \Cref{alg:sample}.
%The initial vertex is scanned even though it is resolved and already belongs to $D$.
%Every examined arc receives a fresh independent state, open with its own probability.
%This is the same law as executing the search on a configuration drawn from the product measure.
%The search from $(\{s\},s)$ with resolved set $\{s\}$ and $L=\emptyset$ is the DFS process $\mathrm{DFS}(s,\{s\})$ of \Cref{alg:dfs}: the second stack remains empty and the search never pushes $s$ again.

\begin{figure}[htbp]
\centering
\begin{tikzpicture}[
  >=stealth, font=\small,
  base/.style={circle,minimum size=6.5mm,inner sep=0pt},
  act/.style={base,draw=blue!70!black,very thick,fill=white},
  pas/.style={base,draw=black!55,fill=white},
  resol/.style={base,draw=black!55,fill=black!18},
  arc/.style={->,thick,draw=black!70},
]
\node[resol] (u) at (0,0) {$u$};
\node[pas] (c1) at (-2.6,-1.3) {$x_1$};
\node[act] (v) at (0,-1.3) {$v$};
\node[pas] (c3) at (2.6,-1.3) {$x_3$};
\draw[arc] (u) -- (c1);
\draw[arc] (u) -- (v);
\draw[arc] (u) -- (c3);
\node[pas] (z1) at (-0.9,-2.6) {$z_1$};
\node[resol] (b) at (0.9,-2.6) {$a$};
\node[pas] (z2) at (-0.9,-3.9) {$z_2$};
\node[pas] (z3) at (0.9,-3.9) {$z_3$};
\node[pas] (z4) at (2.6,-2.6) {$z_4$};
\draw[arc] (v) -- (z1);
\draw[arc] (v) -- (b);
\draw[arc] (z1) -- (z2);
\draw[arc] (b) -- (z3);
\draw[arc] (c3) -- (z4);
\begin{scope}[shift={(5.0,-0.6)}]
  \node[act,minimum size=4.5mm] at (0,0) {};
  \node[anchor=west] at (0.4,0) {vertex of the layer $L$: recorded at its pop};
  \node[pas,minimum size=4.5mm] at (0,-0.75) {};
  \node[anchor=west] at (0.4,-0.75) {other unresolved vertex};
  \node[resol,minimum size=4.5mm] at (0,-1.5) {};
  \node[anchor=west] at (0.4,-1.5) {resolved vertex};
\end{scope}
\end{tikzpicture}
\caption{A sample from a snapshot $(A,u)$ with $L=\{v\}$, drawn with the arcs that are open in one configuration.
The scan order of $u$ is from left to right, and the other vertices of $A$ are not shown. The cut-arc states of $u$ are drawn conditional on at least one being open.
Here the arcs to $x_1$, $v$ and $x_3$ are open. The vertices $x_1$ and $x_3$ enter the first stack, and $v$ enters the second stack. The exploration from $x_3$, including $z_4$, is completed first, and then the exploration from $x_1$. Only after this first stack is empty is $v$ popped and recorded. Its snapshot includes all vertices discovered during the earlier explorations.
After $v$ is scanned, the first stack processes $z_1$ and $z_2$. The arc to the resolved vertex $a$ is ignored, so neither $a$ nor $z_3$ is explored. The states of arcs with resolved heads need not be drawn.
The sample completes all these scans. It records only the fixed layer $L$, so $x_1$, $x_3$, $z_1$, $z_2$ and $z_4$ add to the discovered set but produce no records. The snapshot at $v$ therefore represents the part still reachable after the earlier explorations have finished.}
\label{fig:sample}
\end{figure}

\subsection{Pilot runs and resolving runs}\label{sec:runs}

In \Cref{alg:S}, a round resolves one vertex $v$, and \Cref{alg:generate} builds its pool from $M$ samples.
In the main algorithm, a round resolves a layer of vertices in two steps: a pilot run chooses the layer $L$ and the number $M$ of samples, and a resolving run builds the pools of the vertices of $L$ from $M$ samples.
Together the two runs replace \Cref{alg:generate}.
Both of them consist of samples of the pools, which we define first.

\paragraph{Samples of the pools.}
Like \Cref{alg:S}, the main algorithm keeps a set $R$ of resolved vertices, with $s\in R$, and, for each $u\in R$, a pool at $u$, which is a multiset of snapshots at $u$, and a number $\widehat q_u$, the estimate of $q_u$, which we call the \emph{total weight} of the pool.
Initially, $R=\{s\}$, $\widehat q_s=1$, and the pool of $s$ consists of the single snapshot $(\{s\},s)$.
Once built, the pool and the number $\widehat q_u$ of a resolved vertex $u$ never change.
\emph{Drawing} from a nonempty pool chooses one of its entries uniformly at random.
When we refer to the pools below, their numbers $\widehat q_u$ are included.

A sample runs \Cref{alg:sample} from a snapshot $(A,u)$ in the pools, so the exploration from $s$ to a resolved vertex is never repeated.
As in Steps~1(a) and~1(b) of \Cref{alg:generate}, the snapshot is chosen in two steps: first a resolved vertex $u$, with probability proportional to $\widehat q_u\gamma_u(R)$, and then a snapshot drawn from the pool of $u$.
Here $\gamma_u(R)$ is the probability $\gamma_u$ in~\eqref{eq:gamma} that at least one cut arc leaving $u$ is open; we indicate $R$ because it grows over the rounds.
The factor $\widehat q_u$ is the total weight of the pool of $u$, and the factor $\gamma_u(R)$ makes up for the first scan of \Cref{alg:sample}, which is conditioned on at least one of these cut arcs being open.

Formally, the \emph{outflow} of the pools is $\widehat\Phi\defeq\sum_{u\in R}\widehat q_u\gamma_u(R)$, as in \Cref{alg:generate}.
If $\widehat\Phi>0$, the \emph{start distribution} of the pools is the law $(\widehat q_u\gamma_u(R)/\widehat\Phi)_{u\in R}$, and a draw of a vertex from it is a \emph{start draw}.
Only vertices $u$ with $\widehat q_u>0$ and $\gamma_u(R)>0$ are drawn, so their pools are not empty and the conditioning in the first scan is defined.
Both runs below repeat the same procedure, which we call a \emph{sample of the pools}, with a layer $L\subseteq V\setminus R$: draw a vertex $u$ by a start draw, draw a snapshot $(A,u)$ from the pool of $u$, and run \Cref{alg:sample} with $R$, $L$ and $(A,u)$.
The two runs differ only in the layer: the pilot run uses the empty layer, and the resolving run uses the layer that it resolves.

\paragraph{Pilot runs.}
A pilot run with a threshold $\theta$ draws samples of the current pools with the \emph{empty} layer, that is, it runs \Cref{alg:sample} with $L=\emptyset$, until some unresolved vertex has been scanned $\theta$ times.
It then selects as its output the layer of the vertices scanned at least $\theta/16$ times.
\begin{algorithmfloat}
\begin{algorithm}[Pilot run]\label{alg:pilot}
The input is a resolved set $R\subseteq V$ with $s\in R$, a pool for every resolved vertex, with positive outflow $\widehat\Phi$, and an integer threshold $\theta\geq1$.
\begin{enumerate}[label=\arabic*.,ref=\arabic*,leftmargin=*]
  \item Set a \emph{count} $N(v)\defeq0$ for every $v\notin R$.
  \item\label{step:pilot-loop} Repeat: draw a sample of the pools with the empty layer $L = \emptyset$, and add $1$ to $N(v)$ for every unresolved vertex $v$ that it scans. Stop if some count equals $\theta$ at the end of a sample.
  \item Return the number $T$ of samples drawn and the layer $L\defeq\{v\notin R:\ N(v)\geq \theta/16\}$.
\end{enumerate}
\end{algorithm}
\end{algorithmfloat}

The outflow is positive, so some resolved vertex has a cut arc, the set $V\setminus R$ is not empty, and the samples are defined.
A sample scans every vertex at most once, so each count increases by at most one per sample.
Thus, when the pilot run stops, some count equals $\theta$, and its vertex lies in $L$, as $\theta\geq \theta/16$. In particular, $L$ is not empty.

The pilot run is a Las Vegas procedure: its number $T$ of samples is not fixed in advance, but random, as it stops only when some count reaches $\theta$.
The samples of the pilot run are not reused: the resolving run draws fresh samples.
The vertices of $L$ are selected because these samples scanned them often, so pools built from these samples would be biased upwards.
The resolving run depends on the pilot run only through $L$ and $T$, which sets its number of samples (\Cref{sec:alg}).
Both are fixed before the resolving run starts, so the analysis can condition on them (used later in the proof of \Cref{lem:martingale}(b)).
%The pilot run may not terminate, for instance if no sample can reach an unresolved vertex; the step bound of the main algorithm handles this case (\Cref{sec:alg}).

\paragraph{Resolving runs.}
A resolving run makes a given number of independent samples and builds new pools from their records.
It plays a similar role to $\mathrm{Generate}(v,M)$ in \Cref{alg:generate}, except that it builds the pools of all vertices of a layer $L$ at once, instead of a single vertex $v$.

\begin{algorithmfloat}
\begin{algorithm}[Resolving run]\label{alg:resolving}
The input is a resolved set $R\subseteq V$ with $s\in R$, a number $\widehat q_u$ and a pool for every resolved vertex $u$, with positive outflow $\widehat\Phi$, a nonempty layer $L\subseteq V\setminus R$, and an integer $M\geq1$.
\begin{enumerate}[label=\arabic*.,ref=\arabic*,leftmargin=*]
  \item For every $v\in L$, let the \emph{new pool} of $v$ be empty.
  \item\label{step:resolving-loop} Repeat $M$ times, with fresh randomness each time: draw a sample of the pools with the layer $L$, and add every snapshot that it records at a vertex $v\in L$ to the new pool of $v$.
  \item\label{step:new-weights} For every $v\in L$, let $\widehat q_v$ be $\widehat\Phi/M$ times the number of entries of the new pool of $v$, so that $\widehat q_v=0$ if this pool is empty. Return the new pools and the numbers $\widehat q_v$.
\end{enumerate}
\end{algorithm}
\end{algorithmfloat}

Different samples use independent randomness for their starts and their arc states.
The pools and the layer $L$ do not change during a resolving run.
A sample records at most one snapshot at a given vertex, because every vertex is pushed at most once.
The unresolved vertices outside $L$ are explored as needed, but they produce no records.

\subsection{The full algorithm}\label{sec:alg}

The main algorithm resolves vertices in rounds, each consisting of a pilot run and a resolving run.

\paragraph{Parameters.}
For the rest of the paper, we fix the two parameters
\begin{align}\label{eqn:parameters}
  \theta\defeq1024\lceil\log_2(n+m)\rceil
  \qquad\text{and}\qquad
  \rho\defeq\bigl\lceil 256\,n\,(\lceil\log_2(n+m)\rceil+\varepsilon^{-2})\bigr\rceil,
\end{align}
where $\varepsilon\in(0,1]$ is the error parameter.
In particular, $\theta=\Theta(\log(n+m))$ and $\rho=\Theta(n(\log(n+m)+\varepsilon^{-2}))$.
By the standing assumption~\eqref{asm:standing}, $n\geq2$ and every vertex is reachable from $s$, so $m\geq n-1\geq1$ and $n+m\geq3$, which gives $\theta\geq2048$.

The pilot runs use the threshold $\theta$, and $\rho$ sets the number of samples of the resolving runs.
The algorithm also runs under a \emph{step bound} $16W_0$, where
\begin{align}\label{eq:W0}
  W_0\defeq c_*\bigl(m+(n+P)\rho\bigr),
\end{align}
$P\defeq\sum_{e\in E}p_e$ is the total open probability, and $c_*\geq1$ is an absolute constant determined by the implementation.
Computing $P$ takes $O(m)$ operations.

\paragraph{The algorithm.}
The full algorithm is given in \Cref{alg:main}.
After the reduction of \Cref{sec:reduction}, it starts with $R=\{s\}$ and the pool of $s$, and performs rounds until $t$ is resolved.
Each round first computes the outflow $\widehat\Phi$ of the current pools, and returns $0$ if it is zero.
Otherwise, a pilot run selects a layer $L$ and reports the number $T$ of samples it drew, and a resolving run with $M=\lceil256\rho T/\theta\rceil$ samples builds the pools of the vertices of $L$, which are then added to $R$.
Thus the number of samples of the resolving run is proportional to that of the pilot run.
The output is the estimate $\widehat q_t$.

\begin{algorithmfloat}
\begin{algorithm}[Main algorithm]\label{alg:main}
The input is $G$, $\boldsymbol{p}$, $s$, $t$ and $\varepsilon\in(0,1]$.
\begin{enumerate}[label=\arabic*.,ref=\arabic*,leftmargin=*]
  \item If $s=t$, return $1$. Otherwise reduce the instance as in \Cref{sec:reduction}.
  If $t$ is deleted, return $0$. Compute $P$.
  \item\label{step:params} Compute $\theta$, $\rho$ and the parameter $W_0$ in~\eqref{eq:W0}. Run the remaining steps while counting their steps.
  Before executing a step that would bring the count to at least the \emph{step bound} $16W_0$, stop and return $0$.
  \item Let $R\defeq\{s\}$, with the pool of $s$.
  \item\label{step:rounds} While $t\notin R$, perform a \emph{round}:
  \begin{enumerate}[label=(\alph*),ref=\theenumi(\alph*)]
    \item\label{step:outflow} \emph{Outflow.} Compute the outflow $\widehat\Phi\defeq\sum_{u\in R}\widehat{q}_u\gamma_u(R)$ of the current pools. If $\widehat\Phi=0$, return $0$.
    \item\label{step:pilot} \emph{Pilot run.} Run a pilot run (\Cref{alg:pilot}) with $R$, the pools and the threshold $\theta$, and let $T$ and $L$ be its output.
    \item\label{step:resolve} \emph{Resolving run.} Run a resolving run (\Cref{alg:resolving}) with $R$, the pools, the layer $L$ and $M\defeq\lceil256\rho T/\theta\rceil$.
    \emph{Install} its output: for every $v\in L$, make the new pool of $v$, possibly empty, and the number $\widehat q_v$ the pool and the number of $v$. Then set $R\gets R\cup L$.
  \end{enumerate}
  \item\label{step:return} Return the total weight $\widehat{R}\defeq\widehat{q}_t$ of the pool of $t$.
\end{enumerate}
\end{algorithm}
\end{algorithmfloat}

The layer of every round is not empty, so every round resolves at least one vertex, and the algorithm returns after at most $n-1$ rounds, unless a pilot run does not stop, in which case the step bound stops it.

\subsection{The proof of \texorpdfstring{\Cref{thm:main}}{Theorem 1}}\label{sec:two-lemmas}

The proof of \Cref{thm:main} is based on the following two lemmas about \Cref{alg:main} without the bound on the number of steps.
We call this the \emph{uncapped} algorithm.
The first lemma says that, with high probability, the uncapped algorithm returns after at most $n-1$ rounds with an accurate output. Its proof (\Cref{sec:accuracy}) shows that, with high probability, the predictions stay accurate and the pilot runs are good throughout.
The second one bounds the expected runtime.

\begin{lemma}[Accuracy]\label{lem:accuracy}
  Under the standing assumption~\eqref{asm:standing}, let $\varepsilon\in(0,1]$, and let $\theta$ and $\rho$ be as in \eqref{eqn:parameters}.
There is an event $\mathcal{C}$, determined by the random choices of the uncapped algorithm, with $\Pr[\mathcal{C}]\geq 1-\frac{1}{64}$, such that the following hold on $\mathcal C$.
\begin{enumerate}[label=(\alph*)]
  \item The uncapped \Cref{alg:main} returns in Step~\ref{step:return} after at most $n-1$ rounds.
  \item Its output $\widehat{R}$ satisfies $\abs{\widehat{R}-q_t}\leq\varepsilon q_t$.
\end{enumerate}
\end{lemma}

\begin{lemma}[Runtime]\label{lem:runtime}
Under the standing assumption~\eqref{asm:standing}, let $\varepsilon\in(0,1]$.
There is an implementation of \Cref{alg:main} (given in \Cref{sec:work}) and an absolute constant $c_*\geq1$ such that, with $W_0$ as in~\eqref{eq:W0}, there is a random variable $\mathcal{W}\geq 0$ with $\E\,\mathcal{W}\leq W_0$ such that, on the event $\mathcal{C}$ of \Cref{lem:accuracy}, the uncapped algorithm performs at most $\mathcal{W}$ steps after Step~\ref{step:params}.
\end{lemma}

The constant $c_*$ comes from the constant numbers of steps that the operations of the implementation in \Cref{sec:work} take, through the constants $c_1$ and $c_2$ in the proof of \Cref{lem:runtime}.
It is a fixed number, which \Cref{alg:main} uses in Step~\ref{step:params} to compute $W_0$.

We first prove the following bound in terms of the total open probability $P$, and then derive \Cref{thm:main} from it.

\begin{theorem}\label{thm:expdeg}
There is an FPRAS for two-terminal reliability on directed graphs whose running time, on an input with $n$ vertices, $m$ arcs, total open probability $P$, and error parameter $\varepsilon\in(0,1]$, is
$\widetilde{O}(m+n(n+P)\varepsilon^{-2})$.
\end{theorem}

\begin{proof}
We run \Cref{alg:main} with the implementation of \Cref{sec:work}.
The first step of \Cref{alg:main} is exact, so assume that the algorithm reaches Step~\ref{step:params}.
Couple \Cref{alg:main} with the uncapped algorithm, using the same randomness.
\Cref{alg:main} checks the count before executing every step.
If it stops at this check, that next step of the uncapped algorithm would bring its count to at least $16W_0$, so the uncapped algorithm performs at least $16W_0$ steps after Step~\ref{step:params}.
Until an early stop, the two algorithms coincide.
On the event $\mathcal{C}\cap\{\mathcal{W}<16W_0\}$, the uncapped algorithm returns its output $\widehat{R}$, by \Cref{lem:accuracy}(a), after fewer than $16W_0$ steps, by \Cref{lem:runtime}, so \Cref{alg:main} returns the same output, and by \Cref{lem:accuracy}(b) it is within a factor $1\pm\varepsilon$ of $q_t$.

We bound the probability of the complementary event.
By \Cref{lem:accuracy}, $\Pr[\mathcal{C}^c]\leq\frac{1}{64}$, and by Markov's inequality and \Cref{lem:runtime}, $\Pr[\mathcal{W}\geq 16W_0]\leq\frac1{16}$.
Altogether, the output $\widehat{R}$ of \Cref{alg:main} satisfies
\begin{align}\label{eq:success}
  \Pr\Bigl[\abs{\widehat{R}-q_t}\leq\varepsilon q_t\Bigr]\geq1-\frac1{64}-\frac1{16}=1-\frac{5}{64}>\frac34.
\end{align}

For the number of operations, the reduction costs $O(n+m)$.
Computing $P$ takes $O(m)$ operations.
The parameters of Step~\ref{step:params} cost $O(\log(n+m))$ operations.
The remaining steps are at most $16W_0$ steps on every execution, and by \Cref{lem:step-cost} each of them costs $O(\log(n+m))$ operations.
By~\eqref{eq:W0} and~\eqref{eqn:parameters},
\begin{align*}
  W_0=c_*\bigl(m+(n+P)\rho\bigr)=\widetilde{O}\bigl(m+n(n+P)\varepsilon^{-2}\bigr),
\end{align*}
and so is the running time.
The reduction only deletes vertices and their incident arcs, so it does not increase $n$, $m$ or $P$, and this bound also holds for the input.
\end{proof}

\begin{proof}[Proof of \Cref{thm:main}]
Reduce the input as in \Cref{sec:reduction}, which preserves $\Rel(s,t)$ and costs $O(n+m)$ operations, and, unless it returns, run the algorithm of \Cref{thm:expdeg} on the reduced graph.
Let the reduced graph have $n'\leq n$ vertices, $m'\leq m$ arcs and total open probability $P'$.
Then $P'\leq m'$, and $n'\leq  2m'$ by~\eqref{asm:standing}, so $n'(n'+P')\leq3n'm'$, and \Cref{thm:expdeg} gives time $\widetilde{O}(m'+n'm'\varepsilon^{-2})$.
With the reduction, the running time is $\widetilde{O}(n+m+nm\varepsilon^{-2})$.
As $n+m\leq 2nm$ and $\varepsilon\leq1$, this is $\widetilde{O}(nm\varepsilon^{-2})$.
\end{proof}

\begin{remark}
The success probability $\frac34$ in the definition of an FPRAS can be raised to $1-\delta$ for any $\delta\in(0,1)$ in the standard way: run the algorithm $O(\log(1/\delta))$ times independently and return the median of the outputs.
\end{remark}

\subsection{Proof outline}\label{sec:roadmap}

The rest of the proof is organised as follows.
%The analysis revolves around the \emph{predictions} $Q_R(y)$ of the pools, defined in \Cref{sec:predictions}.
%At the start, $Q_R(y)=q_y$ (see ~\eqref{eq:initial-prediction}).
%Then, \Cref{lem:unbiased} shows that they form martingales, as in \Cref{alg:S}.
%For a resolved vertex $y$, the prediction is the total weight of its pool, which is the algorithm's estimate of $q_y$.

%Given the resolved set $R$ and its pools, $Q_R(y)$ adds up, over their snapshots $(A,u)$, the weights $\widehat q_u\widehat\pi_u(A)$ times the probability $\mathrm{Reach}_R(u,A;y)$ of an open path from $u$ to $y$ whose vertices after $u$ avoid $A\cup R$.
%At the start, $R=\{s\}$ and the only snapshot is $(\{s\},s)$ with weight one, so $Q_R(y)=q_y$ by~\eqref{eq:initial-prediction}.
%When a resolving run adds a layer $L$ to the resolved set, the old pools count only the reachability that also avoids $L$.
%Its explorations finish this part before popping a vertex of $L$, and the new pools replace the remaining contribution in expectation, as \Cref{lem:unbiased} proves.
%Thus the predictions are martingales starting at the reliabilities (\Cref{lem:unbiased}), as in \Cref{alg:S}.
%For a resolved vertex $y$, the prediction is the total weight of its pool, which is the algorithm's estimate of $q_y$.
\begin{itemize}
  \item \Cref{sec:searches} defines the \emph{predictions} $Q_R(y)$ and shows that they form martingales, sample by sample (\Cref{lem:martingale}).
    This rests on the last-exit lemma (\Cref{lem:last-exit}), whose proof (\Cref{sec:last-exit-proof}) studies the deferred exploration restarted from a snapshot $(A,u)$.
    %, determines which vertices the exploration reaches, and shows that a snapshot needs to store only its discovered set. 
    The pops of the vertices of the layer partition the exploration into segments, and the segments starting at these pops partition the contribution of the new pools. 
    This is the last-exit decomposition. %\Cref{lem:sample-mean} also gives the probability that a sample scans an unresolved vertex, which is proportional to its prediction (\Cref{lem:scans}).
  \item \Cref{sec:accuracy} proves \Cref{lem:accuracy}. 
    For the analysis only, we introduce a monitored algorithm in \Cref{sec:monitored}, which halts when one of two tests fails.
    When both tests pass, the run is good, and we can show that the differences of the predictions are suitably bounded (\Cref{lem:increments}).
      Then, we bound the accumulated error by Freedman's inequality (\Cref{lem:freedman}).
      The proof is complete by bounding the failure probabilities (\Cref{sec:pilot}).

    %For the analysis only, the algorithm is halted once a pilot run is not good or a prediction is inaccurate. After a good pilot run, a record changes the prediction for $y$ by at most $2q_y/\rho$, by Harris's inequality (\Cref{lem:harris}). The layers are disjoint, so the conditional variances of all samples add up to $O(nq_y^2/\rho)$. Freedman's inequality then keeps every prediction accurate up to a factor $1\pm\frac13$, and the output up to $1\pm\varepsilon$, with high probability. After any fixed number of samples of a pilot run, the count of an unresolved vertex is a sum of independent indicators, each with mean its scan probability. The Chernoff bound, applied to these counts, shows that the pilot runs are good.
  \item \Cref{sec:work} proves \Cref{lem:runtime}.
    We use data structures that sample the open arcs leaving a vertex directly, without testing the closed ones, which speeds up the algorithm.
    Finally, an amortised analysis bounds the total work over all rounds (\Cref{lem:round-sums}): the work of each round is charged to the unresolved vertices, and the charges of a vertex over the rounds form a geometric sum.
  \item \Cref{sec:undirected} proves \Cref{thm:undirected} using \Cref{thm:expdeg}.
    We first turn the undirected graph into a bi-directed graph while preserving the reliability, by a result of Ball \cite{Ball80}.
    Then we prove a sparsification lemma, \Cref{lem:sparsification}, so that the call of \Cref{thm:expdeg} is efficient enough. 
    The proof of the sparsification lemma uses Harris's inequality (\Cref{lem:harris}).
\end{itemize}
\Cref{fig:proof-map} summarises these dependencies.

\begin{figure}[htbp]
\centering
\begin{tikzpicture}[
  >=stealth, font=\small,
  box/.style={draw=black!60,rounded corners=2pt,align=center,inner sep=4pt,text width=3.05cm},
  dep/.style={->,thick,draw=black!65},
]
\node[box] (search) at (0,-1.8) {The prediction martingale\\(\Cref{sec:searches})};
\node[box] (freedman) at (3.95,0) {Freedman's inequality\\(\Cref{lem:freedman})};
\node[box] (incr) at (3.95,-1.8) {Increments of a resolving run\\(\Cref{lem:increments})};
\node[box] (layer) at (3.95,-3.6) {Harris's inequality\\(\Cref{lem:harris})};
\node[box] (acc) at (7.9,-1.8) {Accuracy\\\Cref{lem:accuracy}\\(\Cref{sec:accuracy})};
\node[box] (work) at (7.9,-3.6) {Runtime\\\Cref{lem:runtime}\\(\Cref{sec:work})};
\node[box] (contr) at (7.9,-5.4) {Sparsification\\\Cref{lem:sparsification}\\(\Cref{sec:undirected})};
\node[box,very thick] (thm1) at (11.85,-1.8) {\Cref{thm:main}};
\node[box] (expdeg) at (11.85,-3.6) {Total open probability\\\Cref{thm:expdeg}\\(\Cref{sec:two-lemmas})};
\node[box,very thick] (thm2) at (11.85,-5.4) {\Cref{thm:undirected}};
\draw[dep] (search) -- (incr);
\draw[dep] (incr) -- (acc);
\draw[dep] (freedman) -- (acc);
\draw[dep] (layer) -- (incr);
\draw[dep] (layer) -- (contr);
\draw[dep] (acc) -- (work);
\draw[dep] (work) -- (expdeg);
\draw[dep] (expdeg) -- (thm1);
\draw[dep] (expdeg) -- (thm2);
\draw[dep] (contr) -- (thm2);
\end{tikzpicture}
\caption{The structure of the proofs of \Cref{thm:main,thm:expdeg,thm:undirected}. %: an arrow points from a part of the analysis to a part that uses it. A dependency that follows a path of arrows, such as that of \Cref{thm:main} on \Cref{lem:accuracy}, is not drawn separately.
}
\label{fig:proof-map}
\end{figure}

\section{The prediction martingale}\label{sec:searches}

The main result of this section is that predictions, a key notion in the analysis, form martingales, one sample at a time (\Cref{lem:martingale}).
\Cref{sec:predictions} defines the predictions, and \Cref{sec:sample-martingales} states this result.
\Cref{sec:resolving-mean,sec:last-exit-proof} prove it, by computing the contribution of one sample to a prediction.
Throughout, the standing assumption~\eqref{asm:standing} holds, and $R$ is a resolved set with $s\in R$.

%prove that the predictions of the pools form a martingale over the rounds of the main algorithm, starting at the reliabilities (\Cref{lem:unbiased}).
%\Cref{sec:predictions} defines the predictions and shows that their initial values are the reliabilities.
%\Cref{sec:resolving-mean} states the last-exit lemma, \Cref{lem:last-exit}, about a single exploration restarted from a snapshot, and derives the martingale property from it.
%The rest of the section proves \Cref{lem:last-exit}.
%\Cref{sec:last-exit-proof} determines which vertices the restarted exploration reaches, overall and below each vertex of the layer, and completes the proof.
%The proof of accuracy in \Cref{sec:accuracy} uses the same mean identity~\eqref{eq:resolving-mean} together with the deterministic bound of \Cref{lem:last-exit}(a) on the contribution of each sample.
%Throughout, the standing assumption holds, and $R$ is a resolved set with $s\in R$.

\subsection{Last exits and predictions}\label{sec:predictions}

For a set $R\subseteq V$, a snapshot $(A,u)$ with $u\in R$ and $u\in A$, and a vertex $y$, define
\begin{align*}
  \mathrm{Reach}_R(u,A;y)\defeq
  \Pr\bigl[\text{there is an open path from $u$ to $y$ whose vertices after $u$ avoid $A\cup R$}\bigr].
\end{align*}
This quantity is the counterpart of $\mathrm{Reach}^{\mathrm{S}}_R(u,A;y)$ in~\eqref{eq:reach-S} for \Cref{alg:S}.
It uses reachability instead of paths in the DFS tree,
which means it does not use an exploration order.
Paths of length zero are allowed, so $\mathrm{Reach}_R(u,A;u)=1$.
For $y\neq u$, this is reachability after forbidding the other vertices of $A\cup R$.
In particular, $\mathrm{Reach}_R(u,A;y)=0$ when $y\in R\setminus\{u\}$.
An open path counted by $\mathrm{Reach}_R$ makes its last exit from $R$ at $u$.
\Cref{lem:reach} shows that the deferred exploration with any fixed layer $L\subseteq V\setminus R$ finds exactly these vertices.
If $R\subseteq R''$, every path allowed by $\mathrm{Reach}_{R''}(u,A;y)$ is allowed by $\mathrm{Reach}_R(u,A;y)$, so $\mathrm{Reach}_{R''}(u,A;y)\leq\mathrm{Reach}_R(u,A;y)$.

For $u\in R$, let $\widehat\pi_u$ be the empirical distribution of the pool of $u$, as in \Cref{sec:s-analysis}.
The pool gives a snapshot $(A,u)$ the \emph{weight} $\widehat q_u\widehat\pi_u(A)$, and these weights add up to $\widehat q_u$.
For $u\neq s$, the weight is $\widehat\Phi/M$ times the number of entries equal to $(A,u)$, where $\widehat\Phi$ and $M$ are the outflow and the number of samples of the resolving run that built the pool (Step~\ref{step:new-weights} of \Cref{alg:resolving}).
The \emph{prediction} of the pools for a vertex $y$ is
\begin{align*}
  Q_R(y)\defeq\sum_{u\in R}\widehat q_u\sum_{A\subseteq V:\,u\in A}\widehat\pi_u(A)\,\mathrm{Reach}_R(u,A;y).
\end{align*}
Here and below, a sum over $A$ ranges over the sets $A$ with $u\in A\subseteq V$, and similarly for the sets of other snapshots.
As $\mathrm{Reach}_R(y,A;y)=1$, and $\mathrm{Reach}_R(u,A;y)=0$ for $u\in R\setminus\{y\}$, only the snapshots at $y$ contribute to $Q_R(y)$ for $y\in R$, so
\begin{align}\label{eq:prediction-resolved}
  Q_R(y)=\widehat q_y\qquad\text{for every }y\in R.
\end{align}
Initially, $R=\{s\}$, and the pool of $s$ consists of the single snapshot $(\{s\},s)$, with $\widehat q_s=1$.
For this snapshot, the forbidden set $(A\cup R)\setminus\{s\}$ is empty.
An open path from $s$ to $y$ can be taken to visit $s$ only at its start, by removing cycles.
Thus the definition of $\mathrm{Reach}_{\{s\}}$ gives $\mathrm{Reach}_{\{s\}}(s,\{s\};y)=q_y$, also for $y=s$, when both sides are one.
As the pool gives weight one to this sole snapshot, the initial predictions are the reliabilities:
\begin{align}\label{eq:initial-prediction}
  Q_{\{s\}}(y)=q_y\qquad\text{for every vertex }y.
\end{align}
The prediction is a tool of the analysis.
The algorithm never computes it.
It plays the role of the prediction~\eqref{eq:prediction} of \Cref{alg:S}, with $\mathrm{Reach}_R(u,A;y)$ replacing $\mathrm{Reach}^{\mathrm{S}}_R(u,A;y)$ of~\eqref{eq:reach-S}.

\subsection{Sample-by-sample martingales}\label{sec:sample-martingales}

The prediction for a vertex $y$ changes only when a resolving run installs its new pools, and this change is the total effect of all samples of the run.
To control it by a concentration bound, we split it into one small step for each sample.
More precisely, we define a process $\widetilde Q_i(y)$ that starts at $q_y$, moves at each sample of a resolving run by the contribution of this sample minus its mean, and equals the prediction for $y$ at the start of every round.
The goal of this subsection is to show that this process is a martingale.
To define it, we first fix notation for the rounds and for the random choices of the algorithm.
We work with the uncapped algorithm, that is, the main algorithm without the step bound.

\paragraph{Notation for the rounds.}
The rounds of \Cref{alg:main} are numbered $r=1,2,\ldots$, as in \Cref{sec:s-analysis}.
If round $r$ is executed, let $R_r$ be the resolved set at its start and $\widehat\Phi_r$ the outflow of its pools, whose prediction for a vertex $y$ is $Q_r(y)=Q_{R_r}(y)$.
If the resolving run of round $r$ is executed, let $R_{r+1}$ and $Q_{r+1}(y)$ be the resolved set and the prediction of its pools after the resolving run has installed its new pools, also when round $r+1$ is not executed.
If the pilot run of round $r$ is executed and stops, let $T_r$ and $L_r$ be its number of samples and its layer, and let $M_r\defeq\lceil256\rho T_r/\theta\rceil$ be the number of samples of the resolving run of round $r$.

\paragraph{Sources of randomness.}
To define filtrations precisely, we realise all random choices of the uncapped algorithm, over all rounds, by a single sequence $U_1,U_2,\ldots$ of independent random sources, each an infinite sequence of independent uniform random variables on $(0,1)$.
The algorithm takes these sources in order as it runs, and each of the following takes the next source that has not been used before.
\begin{itemize}
  \item Every pilot run takes one source for all of its samples: a random source can be split into infinitely many independent random sources, one for each potential sample, and the $k$-th sample of the pilot run uses the $k$-th of them.
  \item Every sample of a resolving run takes one source.
\end{itemize}
Thus, if the resolving run of round $r$ is executed, then round $r$ uses $1+M_r$ consecutive sources: the first for its pilot run, and then one for each sample of its resolving run.
For example, if $M_1=3$, then round $1$ uses $U_1$ for its pilot run and $U_2,U_3,U_4$ for the samples of its resolving run, and the pilot run of round $2$ uses $U_5$.
As $M_r$ is random, the round to which a given source belongs is not fixed in advance.
The sources after the last one used remain unused.
Let $\mathcal H_0$ be the trivial $\sigma$-algebra, and $\mathcal H_i\defeq\sigma(U_1,\ldots,U_i)$ for $i\geq1$.
By induction over $i$, the following are determined by $\mathcal H_{i-1}$: whether $U_i$ is used, and if so, whether by a pilot run or by a sample of a resolving run, the round, the index of the sample, and the inputs of the pilot run or of the resolving run.
These inputs are the resolved set and the pools, and, for a resolving run, also its layer and its number of samples, which the earlier random sources determine.
As $U_i$ is independent of $\mathcal H_{i-1}$, conditional on $\mathcal H_{i-1}$, the pilot run or the sample that uses $U_i$ has the law of a pilot run or of a sample of a resolving run with these inputs.

\paragraph{The martingales.}
Fix a vertex $y$.
For a round $r$ whose resolving run is executed, let
\begin{align}\label{eqn:Q_r-mu_r}
  Q_{\mathrm{old},r}(y)\defeq\sum_{u\in R_r}\widehat q_u\sum_A\widehat\pi_u(A)\,\mathrm{Reach}_{R_r\cup L_r}(u,A;y)
  \qquad\text{and}\qquad
  \mu_r(y)\defeq\frac{Q_r(y)-Q_{\mathrm{old},r}(y)}{M_r}.
\end{align}
Here $Q_{\mathrm{old},r}(y)$ is what the pools of the vertices of $R_r$ still contribute to the prediction for $y$ once $L_r$ is resolved, since the paths they count may then no longer visit $L_r$.
Thus, $Q_r(y)-Q_{\mathrm{old},r}(y)$ is the part of the prediction that these pools lose, and the new pools have to make up for it in expectation.
The quantity $\mu_r(y)$ is the share of this loss for each of the $M_r$ samples of the resolving run, and it will turn out to be the mean contribution of a sample.
Let $Y_{r,k}(y)$ be the contribution of the $k$-th sample of this resolving run to the prediction for $y$.
The records of the $k$-th sample are the snapshots at its pops of vertices of $L_r$ (Step~\ref{step:record} of \Cref{alg:sample}), and each of them enters the new pool of its vertex with weight $\lambda_r\defeq\widehat\Phi_r/M_r$ (Step~\ref{step:new-weights} of \Cref{alg:resolving}), where $\widehat\Phi_r=\sum_{u\in R_r}\widehat q_u\gamma_u(R_r)$ is the outflow of the pools at the start of round $r$, so
\begin{align}\label{eq:sample-contribution}
  Y_{r,k}(y)\defeq\lambda_r\sum_{(D,v)}\mathrm{Reach}_{R_r\cup L_r}(v,D;y),
\end{align}
where the sum is over the records $(D,v)$ of the $k$-th sample.
The pools of the vertices of $R_r$ do not change and contribute $Q_{\mathrm{old},r}(y)$, and the new pools consist of the records of the $M_r$ samples, so
\begin{align}\label{eq:resolving-sum}
  Q_{r+1}(y)=Q_{\mathrm{old},r}(y)+\sum_{k=1}^{M_r}Y_{r,k}(y).
\end{align}
Let $\widetilde Q_0(y)\defeq q_y$.
For $i\geq1$, if $U_i$ is used by the $k$-th sample of the resolving run of a round $r$, let
\begin{align}\label{eq:tilde-Q}
  \widetilde Q_i(y)\defeq \widetilde Q_{i-1}(y)+Y_{r,k}(y)-\mu_r(y),
\end{align}
and otherwise, that is, if $U_i$ is used by a pilot run or not used at all, let $\widetilde Q_i(y)\defeq \widetilde Q_{i-1}(y)$.
Let $\Delta_i(y)\defeq \widetilde Q_i(y)-\widetilde Q_{i-1}(y)$ be its differences.

\begin{lemma}[Prediction martingale]\label{lem:martingale}
Fix a vertex $y$.
\begin{enumerate}[label=(\alph*)]
  \item For a round $r$, let $i_r$ be the number of random sources used before it. Then $\widetilde Q_{i_r}(y)=Q_r(y)$. 

    If the pilot run of round $r$ is executed, it uses the random source $U_{i_r+1}$, and $\widetilde Q_{i_r+1}(y)=\widetilde Q_{i_r}(y)$.

    If the resolving run of round $r$ is executed, then its $k$-th sample uses the random source $U_{i_r+k+1}$, and, for $0\leq k\leq M_r$,
  \begin{align*}
    \widetilde Q_{i_r+k+1}(y)=Q_{\mathrm{old},r}(y)+\sum_{j=1}^kY_{r,j}(y)+(M_r-k)\mu_r(y).
  \end{align*}
  In particular, $\widetilde Q_{i_r+M_r+1}(y)=Q_{r+1}(y)$.
  \item The process $(\widetilde Q_i(y))_{i\geq0}$ is a nonnegative martingale for $(\mathcal H_i)_{i\geq0}$.
\end{enumerate}
\end{lemma}

\begin{proof}[Proof of \Cref{lem:martingale}(a)]
We use induction on $r$.
No random source is used before round $1$, so $i_1=0$ and $\widetilde Q_0(y)=q_y=Q_1(y)$, by~\eqref{eq:initial-prediction}.
Assume round $r$ is executed and $\widetilde Q_{i_r}(y)=Q_r(y)$.
If the pilot run of round $r$ is executed, it uses the next random source $U_{i_r+1}$, and otherwise this random source is not used.
In both cases, $\widetilde Q_{i_r+1}(y)=\widetilde Q_{i_r}(y)=Q_r(y)$.
If the resolving run of round $r$ is executed, its samples use the next random sources $U_{i_r+2},\ldots,U_{i_r+M_r+1}$, in order.
After $k$ of them,
\begin{align*}
  \widetilde Q_{i_r+k+1}(y)
  =Q_r(y)+\sum_{j=1}^k\bigl(Y_{r,j}(y)-\mu_r(y)\bigr)
  =Q_{\mathrm{old},r}(y)+\sum_{j=1}^kY_{r,j}(y)+(M_r-k)\mu_r(y),
\end{align*}
where the second equality uses $M_r\mu_r(y)=Q_r(y)-Q_{\mathrm{old},r}(y)$ by \eqref{eqn:Q_r-mu_r}.
For $k=M_r$, this is $Q_{\mathrm{old},r}(y)+\sum_{j=1}^{M_r}Y_{r,j}(y)=Q_{r+1}(y)$, by~\eqref{eq:resolving-sum}, as the pools of $R_{r+1}$ are those after this resolving run.
If round $r+1$ is executed, then $i_{r+1}=i_r+M_r+1$, which completes the induction.
\end{proof}

\subsection{The contribution of a sample}\label{sec:resolving-mean}

In this subsection, we prove \Cref{lem:martingale}(b).
It requires the mean of $Y_{r,k}(y)$, which comes down to the contribution of a single exploration to the prediction.
Its mean is given by the last-exit identity, \Cref{lem:last-exit} below, whose proof we defer to \Cref{sec:last-exit-proof}.
To state this identity, we use a variant of a sample whose first scan is not conditioned: the cut arcs leaving $u$ follow the product measure, like all other arcs.

\begin{definition}[Restarted exploration]\label{def:restarted}
For a snapshot $(A,u)$, a resolved set $R$ containing $u$ and a layer $L\subseteq V\setminus R$, the \emph{deferred exploration restarted from $(A,u)$}, with $R$ and $L$, is \Cref{alg:sample} with an unconditioned first scan and without records: it starts with discovered set $A$ and two empty stacks, scans $u$, even though $u$ is resolved and belongs to $A$, and then follows the loop of \Cref{alg:sample}.
On a fixed configuration $\omega\in\{0,1\}^E$, it reads $\omega_e$ whenever it examines an arc $e$, and it uses no other randomness.
\end{definition}

For a layer $L\subseteq V\setminus R$, a vertex $y$ and a snapshot $(A,u)$ with $u\in R$, let $R'\defeq R\cup L$,
namely the new resolved set after a resolving run.
Fix a configuration of arcs $\omega$, and run the deferred exploration restarted from $(A,u)$, with $R$ and $L$, on $\omega$.
Every vertex is pushed at most once, so the exploration pops each vertex of $L$ at most once.
Let $v_1,\dots,v_\ell$ be the vertices of $L$ that it pops, in the order of their pops, where $\ell\geq0$, and for $1\leq i\leq\ell$ let $D_i$ be the discovered set at the pop of $v_i$.
Thus $(D_1,v_1),\dots,(D_\ell,v_\ell)$ are the snapshots at these pops, that is, the snapshots that Step~\ref{step:record} of \Cref{alg:sample} would record.
The number $\ell$, the vertices $v_i$ and the sets $D_i$ are determined by $\omega$, and they do not depend on $y$.
The \emph{contribution} of the exploration to the prediction for $y$ is
\begin{align}\label{eq:contribution}
  S\defeq\sum_{i=1}^{\ell}\mathrm{Reach}_{R'}(v_i,D_i;y).
\end{align}
Once $L$ is resolved, a snapshot $(D,v)$ at $v\in L$ contributes $\mathrm{Reach}_{R'}(v,D;y)$ times $\widehat{\Phi}/M$ to the prediction for $y$ (see \Cref{alg:resolving}), and $S$ adds up these contributions.
The next lemma is the last-exit identity.
Before $L$ is resolved, the starting snapshot $(A,u)$ contributes $\mathrm{Reach}_R(u,A;y)$ to the prediction for $y$, through the open paths from $u$ to $y$ whose vertices after $u$ avoid $A\cup R$.
Once $L$ is resolved, these vertices must also avoid $L$, so the same snapshot contributes only $\mathrm{Reach}_{R'}(u,A;y)$.
The difference $\mathrm{Reach}_R(u,A;y)-\mathrm{Reach}_{R'}(u,A;y)$ is the reachability that $(A,u)$ loses.
It is the probability that $y$ is reachable from $u$ by an open path whose vertices after $u$ avoid $A\cup R$, but not by one whose vertices after $u$ also avoid $L$.
On this event, every open path from $u$ to $y$ whose vertices after $u$ avoid $A\cup R$ has a last vertex $v$ in $L$, where it makes its last exit from $R'$, and after $v$ it avoids $R'$.
This remaining part of the path is what the snapshots at the pops of the vertices $v\in L$ account for, through the terms $\mathrm{Reach}_{R'}(v_i,D_i;y)$ of $S$.
\Cref{lem:last-exit}(b) states that, in expectation, these snapshots supply exactly the lost reachability.
%Weighted as in the pools, the new snapshots therefore replace, in expectation, what the old snapshots lose, and this gives the martingale property (\Cref{lem:unbiased}).
%\Cref{lem:last-exit}(a) bounds $S$ on every configuration, which is used for the concentration in \Cref{sec:accuracy}.

\begin{lemma}[Last-exit contributions]\label{lem:last-exit}
Fix a layer $L\subseteq V\setminus R$, a vertex $y$ and a snapshot $(A,u)$ with $u\in R$, let $R'\defeq R\cup L$, and let $S$ be the contribution~\eqref{eq:contribution} of the deferred exploration restarted from $(A,u)$, with $R$ and $L$.
\begin{enumerate}[label=(\alph*)]
  \item On every configuration $\omega$, the contribution satisfies $0\leq S\leq\sum_{v\in L}\Pr_{G_p}[v\reach y]$, and $S=0$ if no cut arc leaving $u$ is open.
  Here $\Pr_{G_p}$ denotes the unconditioned product measure, with each arc $e$ open independently with probability $p_e$.
  \item Under the product measure of $\omega$, $\E S=\mathrm{Reach}_R(u,A;y)-\mathrm{Reach}_{R'}(u,A;y)$.
\end{enumerate}
\end{lemma}

We prove \Cref{lem:last-exit} in \Cref{sec:last-exit-proof}, after studying the restarted exploration.
To apply it to the samples, we relate a sample to the restarted exploration.
Consider pools for the vertices of $R$ with positive outflow $\widehat\Phi$.
For the analysis it is convenient to combine the random choices of a sample of these pools.
A sample draws its \emph{start}, a vertex $u\in R$ and a snapshot $(A,u)$ from the pool of $u$, with probability
\begin{align}\label{eq:start-law}
  \frac{\gamma_u(R)\,\widehat q_u\widehat\pi_u(A)}{\widehat\Phi},
\end{align}
and different samples draw their starts and their arc states independently.
Only vertices with $\gamma_u(R)>0$ and nonempty pools are drawn.

For ease of analysis, instead of drawing a fresh state each time an arc is examined, we may draw a whole configuration ahead of time. We have the following simple observation.
\begin{observation}\label{lem:deferred}
Fix a start $(A,u)$ with $\gamma_u(R)>0$ and a layer $L\subseteq V\setminus R$.
Draw a whole configuration $\omega\in\{0,1\}^E$, before the exploration starts, from the product measure conditioned on the event $\mathrm{Out}_u$ that some cut arc leaving $u$ is open.
Then the deferred exploration restarted from $(A,u)$ (\Cref{def:restarted}), with $R$ and $L$, executed on $\omega$, has the same law as the sample from $(A,u)$ with the layer $L$.
\end{observation}

%\begin{proof}
%The sample draws the arc states only when it reads them.
%The event $\mathrm{Out}_u$ involves only the cut arcs leaving $u$, whose states the first scan of the sample draws jointly under exactly this conditioning.
%Thus, under the conditioned measure, the states of all other arcs remain independent with their original probabilities.
%The exploration reads every arc at most once, so reading a state of $\omega$ that has not been read before is the same as drawing a fresh one.
%\end{proof}

Now we can compute the mean contribution of one sample.

\begin{lemma}[Mean contribution of a sample]\label{lem:sample-mean}
  Fix pools for the vertices of $R$ with positive outflow $\widehat\Phi$, a layer $L\subseteq V\setminus R$, and a vertex $y$, let $R'\defeq R\cup L$, and let
  \begin{align*}
    Q_{\mathrm{old}}\defeq\sum_{u\in R}\widehat q_u\sum_A\widehat\pi_u(A)\,\mathrm{Reach}_{R'}(u,A;y)
  \end{align*}
  be the contribution of these pools to the prediction for $y$ once $L$ is resolved (note that the subscript of $\mathrm{Reach}$ here is $R'=R\cup L$, not $R$).
  Let $S$ be the contribution~\eqref{eq:contribution} of the exploration of a sample of the pools with the layer $L$, evaluated at its start and its configuration.
  Then $\widehat\Phi\,\E S=Q_R(y)-Q_{\mathrm{old}}$.
\end{lemma}

\begin{proof}
By \Cref{lem:last-exit}(a), which holds on every configuration, irrespective of its law, $0\leq S\leq\abs{L}$, so $S$ has a finite mean.
Fix a start $(A,u)$ with $\gamma_u(R)>0$.
By \Cref{lem:deferred}, the sample from this start can be executed on a configuration $\omega\in\{0,1\}^E$ drawn from the product measure conditional on the event $\mathrm{Out}_u$ that at least one cut arc leaving $u$ is open.
The contribution of its exploration is zero outside this event, by \Cref{lem:last-exit}(a).
Consequently its conditional mean is its mean under the product measure divided by $\gamma_u(R)$, the probability of this event.
By \Cref{lem:last-exit}(b), the mean under the product measure is $\mathrm{Reach}_R(u,A;y)-\mathrm{Reach}_{R'}(u,A;y)$.
By the start law~\eqref{eq:start-law},
\begin{align*}
  \E S
  &=\sum_{\substack{u\in R:\,\gamma_u(R)>0}}\sum_A
       \frac{\gamma_u(R)\,\widehat q_u\widehat\pi_u(A)}{\widehat\Phi}\,
       \frac{\mathrm{Reach}_R(u,A;y)-\mathrm{Reach}_{R'}(u,A;y)}{\gamma_u(R)}\\
  &=\frac 1{\widehat\Phi}\sum_{u\in R}\widehat q_u\sum_A\widehat\pi_u(A)
       \bigl(\mathrm{Reach}_R(u,A;y)-\mathrm{Reach}_{R'}(u,A;y)\bigr)
   =\frac 1{\widehat\Phi}\bigl(Q_R(y)-Q_{\mathrm{old}}\bigr).
\end{align*}
The second line sums over all $u\in R$, whereas the first only sums over those with $\gamma_u(R)>0$.
This adds no new contribution: if $\gamma_u(R)=0$, then, with probability one under the product measure, no cut arc leaving $u$ is open, so the contribution of the exploration is zero by \Cref{lem:last-exit}(a), and \Cref{lem:last-exit}(b) gives $\mathrm{Reach}_R(u,A;y)-\mathrm{Reach}_{R'}(u,A;y)=0$.
\end{proof}

\begin{proof}[Proof of \Cref{lem:martingale}(b)]
Every $\widetilde Q_i(y)$ is determined by the first $i$ random sources.
By~(a), it is a prediction $Q_r(y)$, or $Q_{\mathrm{old},r}(y)+\sum_{j\leq k}Y_{r,j}(y)+(M_r-k)\mu_r(y)$ for a round $r$ and some $k\leq M_r$.
The predictions, $Q_{\mathrm{old},r}(y)$ and the contributions $Y_{r,j}(y)$ are nonnegative by definition.
Moreover, $Q_{\mathrm{old},r}(y)\leq Q_r(y)$, as $\mathrm{Reach}_{R_r\cup L_r}(u,A;y)\leq\mathrm{Reach}_{R_r}(u,A;y)$ for every snapshot $(A,u)$ (\Cref{sec:predictions}), which implies $\mu_r(y)\geq0$.
Therefore every $\widetilde Q_i(y)$ is nonnegative.
Let $U_i$ be used by the $k$-th sample of the resolving run of round $r$. This event is determined by $\mathcal H_{i-1}$.
Conditional on $\mathcal H_{i-1}$, on this event, the sample has the law of a sample of a resolving run with the inputs $R_r$, the pools of $R_r$, $L_r$ and $M_r$, which are determined by $\mathcal H_{i-1}$.
The records of the sample are the snapshots at the pops of vertices of $L_r$ in its exploration, so~\eqref{eq:sample-contribution} gives $Y_{r,k}(y)=\lambda_rS$, where $S$ is the contribution~\eqref{eq:contribution} of this exploration.
Thus, by \Cref{lem:sample-mean}, applied to these inputs, $Y_{r,k}(y)$ has conditional mean $\lambda_r(Q_r(y)-Q_{\mathrm{old},r}(y))/\widehat\Phi_r=\mu_r(y)$, as $\lambda_r=\widehat\Phi_r/M_r$, and $\E[\Delta_i(y)\mid\mathcal H_{i-1}]=0$.
On the complementary event,\footnote{The event is that $U_i$ is used by a sample of some resolving run, and on it, $\mathcal H_{i-1}$ determines the round $r$ and the index $k$. On the complementary event, $U_i$ is used by a pilot run, or it is not used at all, because the algorithm has already returned, or an earlier pilot run never stops.} $\Delta_i(y)=0$.
As the process is nonnegative, taking expectations gives $\E\widetilde Q_i(y)=\E\widetilde Q_{i-1}(y)$ for every $i\geq1$, so $\E\widetilde Q_i(y)=q_y$, and every $\widetilde Q_i(y)$ is integrable.
\end{proof}

\subsection{Proof of \texorpdfstring{\Cref{lem:last-exit}}{the last-exit lemma}}\label{sec:last-exit-proof}

We first determine which vertices the restarted exploration reaches, overall (\Cref{lem:reach}) and between consecutive pops of vertices of the layer (\Cref{lem:restarting}).

\begin{lemma}\label{lem:reach}
Consider the deferred exploration restarted from a snapshot $(A,u)$, with resolved set $R$ containing $u$ and layer $L\subseteq V\setminus R$ (\Cref{def:restarted}), on a fixed configuration $\omega\in\{0,1\}^E$.
It discovers, and subsequently pops, exactly the vertices $y\notin A\cup R$ for which an open path from $u$ to $y$ has all its vertices after $u$ outside $A\cup R$.
\end{lemma}

\begin{proof}
Every newly discovered vertex is found through an open arc from $u$ or from a vertex discovered earlier in the exploration.
Induction over the discoveries therefore gives an open path from $u$ to each discovered vertex, with all vertices after $u$ outside the initial set $A\cup R$.
Every discovered vertex is pushed and is eventually popped, since the complete exploration has finitely many pushes.

Conversely, suppose that such a path leads to a vertex $y$ that is not discovered.
Let $x'$ be its first undiscovered vertex and $x$ its predecessor.
Then $x=u$ or $x$ is discovered during the exploration, so in either case the exploration eventually scans $x$.
At that scan $x'$ is still undiscovered, since it is never discovered, and $x'\notin R$.
The exploration therefore examines the open arc $(x,x')$ and discovers $x'$, a contradiction.
The choice between the two stacks changes the order of the scans, but not this argument, since every discovered vertex is eventually scanned.
\end{proof}

Consider the deferred exploration restarted from a snapshot $(A,u)$, with resolved set $R$ containing $u$ and layer $L\subseteq V\setminus R$ (\Cref{def:restarted}), and let $R'\defeq R\cup L$.
The pops of the vertices of $L$ cut this exploration into \emph{segments}.
The segment of $u$ starts with the initial scan of $u$, and the segment of a popped vertex $v\in L$ starts with the pop of $v$.
Each segment ends just before the next pop of a vertex of $L$, or at the end of the exploration if there is none.
Both stacks are empty at the start, and by Step~\ref{step:pop} of \Cref{alg:sample}, the second stack, which holds only vertices of $L$, is popped only when the first stack is empty.
Hence every segment starts with an empty first stack, and it ends as soon as a scan leaves the first stack empty, possibly already after its first scan.
A vertex $z\in L\cup\{u\}$ has a segment exactly when it is scanned, since $u$ is scanned first and a vertex of $L$ is scanned right after its pop.
For such $z$, let $D_0$ be the discovered set at the start of the segment of $z$.
During this segment, the exploration searches depth first from $z$ through vertices outside $D_0\cup R'$, and it sets aside on the second stack the vertices of $L$ that it discovers.
For example, in \Cref{fig:sample}, the segment of $u$ scans $u$, $x_3$, $z_4$ and $x_1$, and the segment of $v$ scans $v$, $z_1$ and $z_2$.
A \emph{history} of the scan of $z$ is the sequence of the arcs examined before this scan, in order, with their states, so for $z=u$ it is empty.

A snapshot $(D,v)$ stores only the discovered set $D$ and the vertex $v$, not the two stacks of the exploration at the pop of $v$.
\Cref{lem:restarting} implies that for $v\in L$ nothing is lost: at this pop the first stack is empty, and the vertices waiting on the second stack lie in $L$, so each of them is explored in its own segment and has its own snapshot.

\begin{lemma}\label{lem:restarting}
Let $z\in L\cup\{u\}$ and $y\in V$.
When $z$ is scanned, let $D_0$ be the discovered set at the start of the segment of $z$. Thus $D_0=A$ if $z=u$, and $(D_0,z)$ is the snapshot at the pop of $z$ if $z\in L$.
For every $\omega\in\{0,1\}^E$ on which $z$ is scanned, the segment of $z$ scans $z$ and exactly the vertices reachable from $z$ by open paths whose vertices after $z$ lie outside $D_0\cup R'$.

Under the product measure, conditional on any history of the scan of $z$, which determines $D_0$, the probability that the segment of $z$ scans $y$ is $\mathrm{Reach}_{R'}(z,D_0;y)$.
\end{lemma}

\begin{proof}
Fix $\omega\in\{0,1\}^E$ on which $z$ is scanned.
We first identify the vertices other than $z$ that the segment of $z$ scans.
By the definition of segments, each of them is popped during the segment and lies outside $L$, so it comes from the first stack.
The first stack is empty when the segment starts and when it ends.
Hence these vertices are exactly the vertices pushed onto the first stack during the segment.
All of them lie outside $R'$, since the first stack holds only vertices outside $R\cup L$.

Every vertex put on the first stack during the segment of $z$ is discovered through an open arc from $z$ or from a vertex scanned earlier in it.
Induction over these discoveries gives an open path from $z$ to it whose vertices after $z$ lie outside $D_0\cup R'$.
Conversely, suppose that such an open path leads to a vertex not scanned during the segment of $z$.
Take the first vertex $x'$ of the path not scanned in it, and its predecessor $x$.
Then $x=z$ or $x$ is scanned during it.
If $x'$ had been discovered before the scan of $x$, it would either belong to $D_0$, contrary to the choice of the path, or have been pushed onto the first stack during the segment of $z$.
In the latter case it would be popped and scanned before the segment of $z$ ends, again a contradiction.
Thus $x'$ is undiscovered when $x$ is scanned.
It is also outside $R'=R\cup L$, so the open arc $(x,x')$ is examined, and $x'$ is pushed onto the first stack and then scanned during the segment.
This contradiction proves the first statement.

We now prove the second statement.
For a fixed assignment of arc states, the exploration rules are deterministic.
Stop the exploration immediately before it scans $z$, if it reaches that scan.
Recall that a history of the scan of $z$ is the sequence of the arcs examined before this scan, in order, with their states.
For every possible history $\mathfrak h$ at this point, the assignments producing it prescribe exactly the arc states recorded in $\mathfrak h$.
Indeed, prescribing these states makes the exploration examine the same arcs in the same order and make the same pushes and choices between the stacks, by induction over its steps until this scan.
It therefore reaches the scan of $z$ with that same history, even if some intervening steps examine no arc.
Conversely, an execution with history $\mathfrak h$ has each of the prescribed states.

Let $\mathcal U$ be the set of vertices that have not been scanned before the scan of $z$.
The history determines $D_0$ and which vertices were scanned, and hence determines $\mathcal U$.
All recorded arc states have tails outside $\mathcal U$.
Thus conditioning on this history fixes only arc states with tails outside $\mathcal U$.
Independence of disjoint sets of coordinates leaves the arc states with tails in $\mathcal U$ independent with their original probabilities, conditional on the history.

By the first statement, the event in question is that an open path from $z$ to $y$ has all its vertices after $z$ outside $D_0\cup R'$.
The tails of its arcs are $z$ and vertices outside $D_0$.
None has been scanned before the scan of $z$: $z$ is scanned only once, and a vertex scanned earlier already belongs to $D_0$.
These tails lie in $\mathcal U$, so their arc states have the product law conditional on the history.
The conditional probability is therefore the defining probability $\mathrm{Reach}_{R'}(z,D_0;y)$.
\end{proof}

We recall what has to be proved.
When a resolving run installs its pools, the resolved set changes from $R$ to $R'=R\cup L$.
This changes the contribution of an old snapshot $(A,u)$ to the prediction for $y$: it drops from $\mathrm{Reach}_R(u,A;y)$ to $\mathrm{Reach}_{R'}(u,A;y)$, as the paths it counts may no longer visit $L$.
The snapshots at the pops of the vertices of $L$ have to make up for exactly this loss, as \Cref{lem:last-exit}(b) states.
Thus, they should count the vertices that the old snapshot can no longer reach, and not the vertices that it still reaches.

Fix a configuration, and consider the exploration restarted from $(A,u)$.
Let $C$ be the set of vertices it scans.
By \Cref{lem:reach}, these are $u$ and the vertices reachable from $u$ by open paths whose vertices after $u$ avoid $A\cup R$.
Let $C_u$ be the set of vertices scanned in the segment of $u$, and for every popped vertex $v\in L$, let $C_v$ be the set of vertices scanned in the segment of $v$.
Every vertex is scanned at most once, and every scan belongs to exactly one segment, so
\begin{align*}
  C=C_u\sqcup\bigsqcup_{v\in L\text{ popped}}C_v.
\end{align*}
We call this partition the \emph{last-exit decomposition}.

\Cref{lem:restarting} describes its parts.
With $z=u$, it says that $C_u$ consists of the vertices $y$ for which the event defining $\mathrm{Reach}_{R'}(u,A;y)$ occurs, so the old snapshot $(A,u)$ still accounts for these vertices after $L$ is resolved.
With $z=v$, it says that every $y\in C_v$ is reachable from $v$ by an open path whose vertices after $v$ lie outside $R'$.
As $v$ was discovered, there is an open path from $u$ to $v$.
Following it by the path from $v$ to $y$ gives an open path from $u$ to $y$ that leaves $R'$ for the last time at $v$, which explains the name.
A vertex $y$ may have several such paths, with different last vertices in $L$, and the order of the exploration decides which set $C_v$ contains it.

The priority of the first stack over the second is what prevents a vertex from being counted twice.
For example, let $L=\{v\}$, let $a,y\notin A\cup R'$, and suppose that the arcs of the paths $u\to a\to y$ and $u\to v\to y$ are open.
The path $u\to a\to y$ avoids $L$, so the old snapshot $(A,u)$ still accounts for $y$ after $L$ is resolved, and the snapshot at $v$ should not account for $y$ again.
The scan of $u$ pushes $a$ onto the first stack and $v$ onto the second.
By the priority of the first stack, the exploration pops $a$, and hence discovers $y$, before it pops $v$.
Thus, the snapshot $(D,v)$ at $v$ has $y\in D$.
The paths counted by $\mathrm{Reach}_{R'}(v,D;y)$ must avoid $D$ after $v$, so $\mathrm{Reach}_{R'}(v,D;y)=0$, and this snapshot contributes zero to the prediction for $y$.
If instead $v$ were popped before $a$, then $y$ would not be discovered at that pop, and $\mathrm{Reach}_{R'}(v,D;y)$ would be positive, through the arc $(v,y)$.
Then both snapshots would account for $y$.

We now use this decomposition to prove \Cref{lem:last-exit}.

\begin{proof}[Proof of \Cref{lem:last-exit}]
(a)
For every snapshot $(D,v)$, the paths counted by $\mathrm{Reach}_{R'}(v,D;y)$ form a subset of the open paths from $v$ to $y$ in the whole graph.
Thus $0\leq\mathrm{Reach}_{R'}(v,D;y)\leq\Pr_{G_p}[v\reach y]$ for every possible $D$.
A vertex of $L$ is popped at most once.
Adding these inequalities over the popped vertices gives the stated bound on $S$.
If no cut arc leaving $u$ is open, the initial scan pushes no vertex.
Both stacks then remain empty, so the exploration has no pop from $L$, and $S=0$.

(b)
For each $v\in L$, let $\mathcal E_v$ be the event that $v$ is popped and its segment scans $y$.
If $v$ is popped, let $(D_v,v)$ be its snapshot, so $D_v$ is the discovered set at the pop of $v$, namely the set $D_i$ in~\eqref{eq:contribution} when $v=v_i$.
It is random, as it depends on the arc states that the exploration reads before this pop.
Every vertex is popped at most once, so
\begin{align*}
  S=\sum_{v\in L}\ind\{v\text{ is popped}\}\,\mathrm{Reach}_{R'}(v,D_v;y).
\end{align*}
Let $\mathfrak H_v$ be the set of histories of the scan of $v$ (defined before \Cref{lem:restarting}) with positive probability.
For $\mathfrak h\in\mathfrak H_v$, let $\Pr[\mathfrak h]$ be the probability that $v$ is scanned with history $\mathfrak h$.
The arcs examined before the scan of $v$ and their states determine the vertices discovered before it, so $\mathfrak h$ determines $D_v$, and we denote this fixed set by $D_v(\mathfrak h)$.
Then
\begin{align*}
  \E\bigl[\ind\{v\text{ is popped}\}\,\mathrm{Reach}_{R'}(v,D_v;y)\bigr]
  &=\sum_{\mathfrak h\in\mathfrak H_v}\Pr[\mathfrak h]\,\mathrm{Reach}_{R'}\bigl(v,D_v(\mathfrak h);y\bigr)\\
  &=\sum_{\mathfrak h\in\mathfrak H_v}\Pr[\mathfrak h]\,\Pr[\mathcal E_v\mid\mathfrak h]
  =\Pr[\mathcal E_v].
\end{align*}
A vertex of $L$ is popped exactly when it is scanned, so, up to probability zero, the event that $v$ is popped is the disjoint union of the events that $v$ is scanned with history $\mathfrak h\in\mathfrak H_v$.
This gives the first equality, and also the last one, since $\mathcal E_v$ is contained in the event that $v$ is popped.
The middle equality is \Cref{lem:restarting} with $z=v$, for which $D_0=D_v(\mathfrak h)$.
Summing over the finite set $L$ gives $$\E S=\sum_{v\in L}\Pr[\mathcal E_v].$$

The segments are disjoint, and every vertex is scanned at most once, so the events $\mathcal E_v$ are disjoint.
Letting $\mathcal E$ be their union, we obtain $\E[S]=\Pr[\mathcal E]$.

On a fixed configuration, the scanned vertices are $u$ and the popped vertices, and by \Cref{lem:reach} these are exactly the vertices reachable from $u$ by open paths whose vertices after $u$ avoid $A\cup R$.
By \Cref{lem:restarting} with $z=u$, for which $D_0=A$, the segment of $u$ scans exactly those reachable by open paths whose vertices after $u$ also avoid $L$.
Every other scanned vertex is scanned in the segment of a popped vertex of $L$.
Thus $\mathcal E$ occurs if and only if $y$ is reachable from $u$ by an open path whose vertices after $u$ avoid $A\cup R$, but not by one whose vertices after $u$ also avoid $L$.
This is the difference between the two nested reachability events defining $\mathrm{Reach}_R(u,A;y)$ and $\mathrm{Reach}_{R'}(u,A;y)$, and its probability under the product measure is $\mathrm{Reach}_R(u,A;y)-\mathrm{Reach}_{R'}(u,A;y)$.
Since $\E S=\Pr[\mathcal E]$, this proves the stated mean.
\end{proof}

\section{Accuracy}\label{sec:accuracy}

In this section we prove \Cref{lem:accuracy}.
The proof outline is similar to that of \Cref{thm:alg-s} in \Cref{sec:s-analysis} for \Cref{alg:S}: 
given that the predictions are martingales (\Cref{sec:searches}), we bound their increments and then their conditional variances.
The main difference is that Freedman's inequality (\Cref{lem:freedman}) replaces Chebyshev's inequality.
Freedman's inequality requires bounding the differences of the martingale,
and to this end we introduce a monitored algorithm in \Cref{sec:monitored}.
This algorithm halts on two tests, and when the tests pass, some good properties hold and the differences of the martingale are bounded.
Lastly, in \Cref{sec:accuracy-proof}, we prove \Cref{lem:accuracy} by a union bound over the failure of the tests.

\subsection{The monitored algorithm}\label{sec:monitored}

First we introduce the monitored algorithm, which is a device for analysis only and will not be run.
We use the notation of \Cref{sec:sample-martingales} for the rounds.
If round $r$ is executed, then the sink $t\notin R_r$, as the algorithm returns as soon as $t$ is resolved, and we let
\begin{align*}
  q_*(r)\defeq\max_{v\notin R_r}q_v
  \qquad\text{and}\qquad
  Q_*(r)\defeq\max_{v\notin R_r}Q_r(v).
\end{align*}
Recall from \Cref{alg:pilot} that a pilot run stops once some unresolved vertex has been scanned $\theta$ times, and that its layer consists of the unresolved vertices scanned at least $\theta/16$ times.
Here $\theta=1024\lceil\log_2(n+m)\rceil=\Theta(\log(n+m))$, by~\eqref{eqn:parameters}.
As in \Cref{sec:sample-martingales}, if the pilot run of round $r$ is executed and stops, $T_r$ is its number of samples and $L_r$ is its layer.

\begin{definition}[Accurate rounds and good pilot runs]\label{def:good-pilot}
Assume round $r$ is executed.
It is \emph{accurate} if $\abs{Q_r(y)-q_y}\leq\frac13q_y$ for every vertex $y$.
If its pilot run is executed, so that $\widehat\Phi_r>0$, the pilot run is \emph{good} if $Q_*(r)>0$, the pilot run stops, and
\begin{align*}
  T_r>\frac{\theta\widehat\Phi_r}{2Q_*(r)}
  \qquad\text{and}\qquad
  \Bigl\{v\notin R_r:\ q_v\geq\frac{q_*(r)}2\Bigr\}
  \subseteq L_r\subseteq
  \Bigl\{v\notin R_r:\ q_v\geq\frac{q_*(r)}{128}\Bigr\}.
\end{align*}
\end{definition}

Basically, good pilot runs select all vertices whose reliabilities are close to the maximum but no vertices whose reliabilities are significantly smaller.
Moreover, since $T_r$ determines the number $M_r$ of samples drawn in the subsequent resolving run, it needs to be not too small to make sure the resolving run is accurate.
Nevertheless, $\E[T_r]\le \theta\widehat\Phi_r/Q_*(r)$ (\Cref{lem:pilot}(b)), so the run typically is not too long.
For $y\in R_r$, an accurate round implies that $\abs{\widehat q_y-q_y}\leq\frac13q_y$, by~\eqref{eq:prediction-resolved}.
Like the predictions, the number $Q_*(r)$ and the conditions of \Cref{def:good-pilot} are devices for analysis. The algorithm never computes them.

\begin{definition}[The monitored algorithm]\label{def:monitored}
The \emph{monitored algorithm} follows the uncapped algorithm, with the parameters~\eqref{eqn:parameters} and the same random choices, and adds the following two tests, each of which halts it.
\begin{enumerate}[label=(T\arabic*),leftmargin=*]
  \item\label{test:pilot} After the pilot run of a round, halt if it is not good.
  \item\label{test:accurate} After the resolving run of a round has installed its new pools, halt if some vertex $y$ has a prediction, for the new pools, with $\abs{Q(y)-q_y}>\frac13q_y$.
\end{enumerate}
At a halt, keep the resolved set and the pools reached so far.
When we count the steps of the monitored algorithm in \Cref{sec:work}, we count only the steps that the uncapped algorithm also performs, and not the cost of the tests.
This is harmless, since the monitored algorithm is never run.
\end{definition}

The monitored algorithm coincides with the uncapped algorithm until it halts.
A pilot run that does not stop is bad, and we regard the monitored algorithm as halting after it. This happens only on an event of probability zero (shown later in \Cref{lem:pilot}(a)).
The actual algorithm does not perform these tests: they use the unknown reliabilities and predictions, and are only used to define a process for the analysis.
Every pilot run and every resolving run finishes before the test that follows it, so the tests do not change the conditional law of a run given the history before it.
In the rest of this section, round-indexed notation refers to the monitored algorithm.

We first record a simple bound on the predictions.

\begin{lemma}\label{lem:prediction-outflow}
Fix pools for the vertices of a resolved set $R$ with $s\in R$, with outflow $\widehat\Phi$, possibly zero.
For $v\notin R$, the prediction satisfies $0\leq Q_R(v)\leq\widehat\Phi$.
\end{lemma}

\begin{proof}
For a snapshot $(A,u)$ with $u\in R$, the event defining $\mathrm{Reach}_R(u,A;v)$ requires an open path from $u$ to $v$ whose vertices after $u$ lie outside $A\cup R$.
As $u\in R$ and $v\notin R$, this path has a first arc, and the head of this arc lies outside $R$, so the first arc is an open cut arc leaving $u$.
Recall that $\mathrm{Out}_u$ is the event that some cut arc leaving $u$ is open.
Then, the event defining $\mathrm{Reach}_R(u,A;v)$ implies $\mathrm{Out}_u$, and so
$\mathrm{Reach}_R(u,A;v)\leq\gamma_u(R)$.
Multiplying by $\widehat q_u\widehat\pi_u(A)$ and summing over $u\in R$ and $A$ gives $0\leq Q_R(v)\leq\sum_{u\in R}\widehat q_u\gamma_u(R)=\widehat\Phi$.
\end{proof}

\begin{lemma}[The monitored algorithm]\label{lem:monitored}
On every execution of the monitored algorithm, the following hold.
\begin{enumerate}[label=(\alph*)]
  \item Every executed round $r$ is accurate, and $\widehat\Phi_r\geq Q_*(r)\geq Q_r(t)\geq\frac23q_t>0$.
  \item Every resolving run that the monitored algorithm executes follows a good pilot run of the same round, and this round is accurate.
  \item If rounds $r$ and $r+1$ are executed, then $q_*(r+1)<q_*(r)/2$.
  \item The executed rounds are the rounds $1,2,\ldots,r_{\max}$, for some $r_{\max}\leq n-1$. In particular, at most $n-1$ rounds are executed.
  \item The monitored algorithm uses finitely many random sources. (See \Cref{sec:sample-martingales}.) Suppose $i_*$ is the number used. Then, for every vertex $y$, the value $\widetilde Q_{i_*}(y)$ of the martingale of \Cref{lem:martingale} is the prediction for $y$ of the pools at the end of the monitored algorithm.
\end{enumerate}
\end{lemma}

\begin{proof}
(a)
Round $1$ is accurate: $R_1=\{s\}$ with the pool of $s$, so $Q_1(y)=q_y$ for every vertex $y$, by~\eqref{eq:initial-prediction}.
Let $r\geq1$, and suppose that round $r+1$ is executed.
The monitored algorithm reaches round $r+1$ only through round $r$, which therefore executed its resolving run, and the pools of $R_{r+1}$ are those after the new pools of this resolving run were installed.
They passed test~\ref{test:accurate}, since otherwise the monitored algorithm would have halted, so round $r+1$ is accurate.
For every executed round $r$, the vertex $t$ is unresolved, so \Cref{lem:prediction-outflow} and accuracy give
\begin{align*}
  \widehat\Phi_r\geq Q_*(r)\geq Q_r(t)\geq\frac23q_t,
\end{align*}
and $q_t>0$ by the standing assumption~\eqref{asm:standing}.

(b)
The resolving run of a round is executed only after its pilot run has passed test~\ref{test:pilot}, that is, after a good pilot run, and the round is accurate by~(a).

(c)
As round $r+1$ is executed, the resolving run of round $r$ was executed, so the pilot run of round $r$ is good, by~(b).
Hence every $v\notin R_r$ with $q_v\geq q_*(r)/2$ lies in $L_r\subseteq R_{r+1}$, and every vertex $v\notin R_{r+1}$ has $q_v<q_*(r)/2$.
As $t\notin R_{r+1}$, there is such a vertex, and $q_*(r+1)<q_*(r)/2$.

(d)
The rounds are executed in order.
If round $r$ is executed, then the rounds $1,\ldots,r-1$ have executed their resolving runs, and each of them has added a nonempty layer to the resolved set, as its pilot run stopped with a count equal to $\theta\geq\theta/16$ (Step~\ref{step:pilot-loop} of \Cref{alg:pilot}).
These layers are disjoint, so $\abs{R_r}\geq r$.
As $t\notin R_r$, $n-1\ge \abs{R_r}\ge r$.

(e)
The monitored algorithm uses the random sources of \Cref{sec:sample-martingales} in the same order as the uncapped algorithm, and it coincides with the uncapped algorithm until it ends, so \Cref{lem:martingale}(a) applies to its rounds.
By~(d), at most $n-1$ rounds are executed, and each of them uses at most one random source for its pilot run and $M_r$ random sources for its resolving run, so the number $i_*$ of random sources used is finite.
If the monitored algorithm ends after the resolving run of a round $r$, by a halt or a return, then $\widetilde Q_{i_*}(y)=Q_{r+1}(y)$ by \Cref{lem:martingale}(a), and these are the final pools.
Otherwise it ends in round $r$ before executing a resolving run, and then $\widetilde Q_{i_*}(y)=Q_r(y)$ by \Cref{lem:martingale}(a), while the pools of round $r$ are the final pools.
\end{proof}

\paragraph{Stopped martingales.}
Let $i_*$ be the number of random sources that the monitored algorithm uses (\Cref{lem:monitored}(e)), so that it uses $U_1,\ldots,U_{i_*}$.
Whether $i\leq i_*$ is determined by $\mathcal H_{i-1}$, as the same holds for the uncapped algorithm, and every test is applied after a pilot run or a resolving run has finished.
In the rest of this section, we stop the martingales of \Cref{lem:martingale} when the monitored algorithm ends: for every vertex $y$ and every $i>i_*$, we set $\widetilde Q_i(y)\defeq\widetilde Q_{i_*}(y)$ and $\Delta_i(y)\defeq0$.
The new differences are the old ones times the indicator of $i\leq i_*$, which is determined by $\mathcal H_{i-1}$, so they are still martingale differences for $(\mathcal H_i)_{i\geq0}$.
Nothing changes for $i\leq i_*$, so \Cref{lem:martingale}(a) still applies to the rounds of the monitored algorithm.
By \Cref{lem:monitored}(e), the final value $\widetilde Q_{i_*}(y)$ is the prediction for $y$ of the pools at the end of the monitored algorithm.

%\subsection{Freedman's inequality}\label{sec:freedman}
%
%\begin{lemma}[Freedman's inequality~{\cite[Proposition~(2.1)]{Freedman1975}}]\label{lem:freedman}
%Let $(\mathcal H_k)_{k\geq0}$ be a filtration, and let $\Delta_1,\Delta_2,\ldots$ be martingale differences for it: every $\Delta_k$ is $\mathcal H_k$-measurable, and $\E[\Delta_k\mid\mathcal H_{k-1}]=0$.
%Suppose that $\abs{\Delta_k}\leq c$ for every $k\geq1$, for a constant $c>0$.
%Let $S_0\defeq0$ and $V_0\defeq0$, and for $k\geq1$ let
%\begin{align*}
%  S_k\defeq\sum_{j=1}^k\Delta_j
%  \qquad\text{and}\qquad
%  V_k\defeq\sum_{j=1}^k\E\bigl[\Delta_j^2\bigm|\mathcal H_{j-1}\bigr].
%\end{align*}
%Then, for all $a>0$ and $\sigma^2>0$,
%\begin{align*}
%  \Pr\bigl[S_k\geq a\text{ and }V_k\leq\sigma^2\text{ for some }k\geq0\bigr]
%  \leq\exp\Bigl(-\frac{a^2}{2(\sigma^2+ca)}\Bigr).
%\end{align*}
%\end{lemma}
%
%This is Proposition~(2.1) of Freedman's paper, with $K=c$, $b=\sigma^2$ and the stopping time $\tau=\infty$.
%Freedman derives it from his Theorem~(1.6), the case $c=1$, by scaling the differences.
%A finite sequence of differences is extended by zeros.
%The differences $-\Delta_k$ satisfy the same assumptions, with the same $V_k$, so \Cref{lem:freedman} also bounds the lower tail.
%If $V_k\leq\sigma^2$ for every $k$ on every execution, \Cref{lem:freedman} bounds the probability that $S_k\geq a$ for some $k$.
%We apply \Cref{lem:freedman} to the samples of the resolving runs, one at a time (\Cref{lem:inherited}), and to the counts of the pilot runs (\Cref{lem:pilot}).

\subsection{Errors in a good run}\label{sec:inherited}

The predictions of the monitored algorithm carry the errors of all the resolving runs that it has executed.
To control them all at once, we use Freedman's inequality.

\begin{lemma}[Freedman's inequality~{\cite[Proposition~(2.1)]{Freedman1975}}]\label{lem:freedman}
Let $(\mathcal H_k)_{k\geq0}$ be a filtration, and let $\Delta_1,\Delta_2,\ldots$ be martingale differences for it: every $\Delta_k$ is $\mathcal H_k$-measurable, and $\E[\Delta_k\mid\mathcal H_{k-1}]=0$.
Suppose that $\abs{\Delta_k}\leq c$ for every $k\geq1$, for a constant $c>0$.
Let $S_0\defeq0$ and $V_0\defeq0$, and for $k\geq1$ let
\begin{align*}
  S_k\defeq\sum_{j=1}^k\Delta_j
  \qquad\text{and}\qquad
  V_k\defeq\sum_{j=1}^k\E\bigl[\Delta_j^2\bigm|\mathcal H_{j-1}\bigr].
\end{align*}
Then, for all $a>0$ and $b>0$,
\begin{align*}
  \Pr\bigl[S_k\geq a\text{ and }V_k\leq b\text{ for some }k\geq0\bigr]
  \leq\exp\Bigl(-\frac{a^2}{2(b+ca)}\Bigr).
\end{align*}
\end{lemma}

This is Proposition~(2.1) of Freedman's paper, with $K=c$ and the stopping time $\tau=\infty$.
Freedman derives it from his Theorem~(1.6), the case $c=1$, by scaling the differences.

We apply \Cref{lem:freedman} to $\widetilde Q_i(y)$, the sample-by-sample martingale of \Cref{lem:martingale} in \Cref{sec:sample-martingales}. 
To do that, we need to bound the differences $\Delta_i(y)=\widetilde Q_i(y)-\widetilde Q_{i-1}(y)$.

Recall from~\eqref{eqn:parameters} that $\rho=\lceil256n(\lceil\log_2(n+m)\rceil+\varepsilon^{-2})\rceil$, and from \Cref{sec:sample-martingales} that the resolving run of round $r$ has $M_r=\lceil256\rho T_r/\theta\rceil$ samples and that $\lambda_r=\widehat\Phi_r/M_r$ is the weight of each of its records.
Recall also from \Cref{sec:sample-martingales} that $Y_{r,k}(y)$ is the contribution of the $k$-th sample of this resolving run to the prediction for a vertex $y$, and that $\mu_r(y)=(Q_r(y)-Q_{\mathrm{old},r}(y))/M_r$, defined in~\eqref{eqn:Q_r-mu_r}, is the conditional mean of $Y_{r,k}(y)$ given the random sources used before the sample.
By~\eqref{eq:tilde-Q}, $\Delta_i(y)=Y_{r,k}(y)-\mu_r(y)$, where $i$ is the index of the random source $U_i$ used.

\begin{lemma}[Increments of a resolving run]\label{lem:increments}
  Suppose that the monitored algorithm executes the resolving run of round $r$. Let $y$ be a vertex, and let $d\defeq2q_y/\rho$.
\begin{enumerate}[label=(\alph*)]
  \item Every sample of the resolving run satisfies $0\leq Y_{r,k}(y)\leq d\abs{L_r}$, and $0\leq\mu_r(y)\leq d\abs{L_r}$.
  \item The conditional variances of the differences at the samples of the resolving run satisfy
  \begin{align*}
    \sum_{k=1}^{M_r}\E\bigl[\Delta_{i_r+k+1}(y)^2\bigm|\mathcal H_{i_r+k}\bigr]
    \leq d\abs{L_r}\,Q_r(y)
    \leq\frac43\,d\abs{L_r}\,q_y,
  \end{align*}
  where $i_r$ is the number of random sources used before round $r$, as in \Cref{lem:martingale}.
\end{enumerate}
\end{lemma}

\begin{proof}
By \Cref{lem:monitored}(b), round $r$ is accurate and its pilot run is good.

(a)
We first show that every $v\in L_r$ satisfies $\lambda_r\Pr_{G_p}[v\reach y]\leq d$.
As the pilot run is good, by \Cref{def:good-pilot}, $T_r>\theta\widehat\Phi_r/(2Q_*(r))$, and hence
\begin{align*}
  M_r\geq\frac{256\rho T_r}{\theta}>\frac{128\rho\widehat\Phi_r}{Q_*(r)},
  \qquad\text{so}\qquad
  \lambda_r=\frac{\widehat\Phi_r}{M_r}<\frac{Q_*(r)}{128\rho}.
\end{align*}
Accuracy gives $Q_r(x)\leq\frac43q_x\leq\frac43q_*(r)$ for every $x\notin R_r$, so $Q_*(r)\leq\frac43q_*(r)$ and $\lambda_r<q_*(r)/(96\rho)$.
Let $v\in L_r$.
As the pilot run is good, $q_v\geq q_*(r)/128>0$.
The events $s\reach v$ and $v\reach y$ are increasing, and together they imply $s\reach y$.
Harris's inequality (\Cref{lem:harris}) therefore gives
\begin{align*}
  q_v\Pr_{G_p}[v\reach y]\leq\Pr_{G_p}[s\reach v\text{ and }v\reach y]\leq q_y,
\end{align*}
so $\Pr_{G_p}[v\reach y]\leq q_y/q_v\leq128q_y/q_*(r)$.
Consequently
\begin{align*}
  \lambda_r\Pr_{G_p}[v\reach y]<\frac{q_*(r)}{96\rho}\cdot\frac{128q_y}{q_*(r)}=\frac{4q_y}{3\rho}\leq d.
\end{align*}
The contribution of a sample is $Y_{r,k}(y)=\lambda_rS$, where $S$ is the contribution~\eqref{eq:contribution} of its exploration.
By \Cref{lem:last-exit}(a) and the bound just shown,
\begin{align*}
  0\leq Y_{r,k}(y)\leq \lambda_r\sum_{v\in L_r}\Pr_{G_p}[v\reach y]\leq d\abs{L_r}.
\end{align*}
The conditional mean $\mu_r(y)$ of such a variable lies in the same interval.

(b)
Let $i=i_r+k+1$, so that $U_i$ is used by the $k$-th sample.
By \Cref{lem:martingale}(b), the difference $\Delta_i(y)=Y_{r,k}(y)-\mu_r(y)$ has conditional mean zero given $\mathcal H_{i-1}$, so $\E[\Delta_i(y)^2\mid\mathcal H_{i-1}]$ is the conditional variance of $Y_{r,k}(y)$, whose conditional mean is $\mu_r(y)$.
This variance is at most the conditional second moment, and $Y_{r,k}(y)^2\leq d\abs{L_r}\,Y_{r,k}(y)$ by~(a).
Thus $\E[\Delta_i(y)^2\mid\mathcal H_{i-1}]\leq d\abs{L_r}\,\mu_r(y)$.
Summing over the $M_r$ samples, and using $M_r\mu_r(y)=Q_r(y)-Q_{\mathrm{old},r}(y)\leq Q_r(y)$, as $Q_{\mathrm{old},r}(y)\geq0$, gives the first inequality.
Accuracy gives $Q_r(y)\leq\frac43q_y$, and hence the second.
\end{proof}

Now we can apply Freedman's inequality, \Cref{lem:freedman}.

\begin{lemma}[Error control]\label{lem:inherited}
Let $y$ be a vertex.
\begin{enumerate}[label=(\alph*)]
  \item For every $\zeta\in(0,1]$,
  \begin{align*}
    \Pr\Bigl[\abs{\widetilde Q_i(y)-q_y}>\zeta q_y\text{ for some }i\geq0\Bigr]
    \leq2\exp\Bigl(-\frac{\zeta^2\rho}{n(16/3+4\zeta)}\Bigr).
  \end{align*}
  \item For $\zeta=\frac13$, the bound in~(a) is $2e^{-\rho/(60n)}$.
  \item For $\zeta=\varepsilon$, the bound in~(a) is at most $2e^{-3\varepsilon^2\rho/(28n)}$.
\end{enumerate}
\end{lemma}

\begin{proof}
(a)
By the standing assumption~\eqref{asm:standing}, $q_y>0$.
Let $d\defeq2q_y/\rho$, as in \Cref{lem:increments}, and consider the differences $\Delta_i(y)$ of the stopped martingale $\widetilde Q_i(y)$.
They vanish except at the samples of the resolving runs that the monitored algorithm executes.
At a sample of the resolving run of round $r$, \Cref{lem:increments}(a) places both $Y_{r,k}(y)$ and $\mu_r(y)$ in $[0,d\abs{L_r}]$, so $\abs{\Delta_i(y)}\leq d\abs{L_r}$.
As $L_r\subseteq V\setminus R_r$ and $s\in R_r$, we have $\abs{L_r}\leq n-1$, so every difference satisfies
\begin{align*}
  \abs{\Delta_i(y)}\leq dn=\frac{2nq_y}\rho.
\end{align*}
By \Cref{lem:increments}(b), the conditional variances of the differences at the samples of the resolving run of round $r$ add up to at most $\frac43d\abs{L_r}q_y$.
The layers of different rounds are disjoint subsets of $V\setminus\{s\}$, so $\sum_r\abs{L_r}\leq n-1$, where the sum ranges over the rounds whose resolving runs are executed.
Hence, on every execution and for every $i$,
\begin{align*}
  V_i\defeq\sum_{j=1}^i\E\bigl[\Delta_j(y)^2\bigm|\mathcal H_{j-1}\bigr]
  \leq\frac43dq_y(n-1)
  \leq\frac{8nq_y^2}{3\rho}.
\end{align*}
The partial sum $\sum_{j\leq i}\Delta_j(y)$ equals $\widetilde Q_i(y)-q_y$.
Apply \Cref{lem:freedman} with $c=2nq_y/\rho$, $a=\zeta q_y$ and $b=8nq_y^2/(3\rho)$ to the differences $\Delta_i(y)$ and to their negatives.
As $V_i\leq b$ always, this gives
\begin{align*}
  \Pr\Bigl[\abs{\widetilde Q_i(y)-q_y}\geq\zeta q_y\text{ for some }i\Bigr]
  \leq2\exp\Bigl(-\frac{\zeta^2q_y^2}{2(b+c\zeta q_y)}\Bigr).
\end{align*}
Here $b+c\zeta q_y=(\frac83+2\zeta)nq_y^2/\rho$, so the exponent is $-\zeta^2\rho/(n(\frac{16}3+4\zeta))$, which proves~(a).

(b)
For $\zeta=\frac13$, we have $n(\frac{16}3+\frac43)=\frac{20}3n$, and the exponent is $-\frac19\cdot\frac{3\rho}{20n}=-\frac{\rho}{60n}$.

(c)
For $\zeta=\varepsilon\leq1$, we have $n(\frac{16}3+4\varepsilon)\leq\frac{28}3n$, and the exponent is at most $-3\varepsilon^2\rho/(28n)$.
\end{proof}

\subsection{Analysing pilot runs}\label{sec:pilot}

Recall from~\eqref{eqn:parameters} that $\theta=1024\lceil\log_2(n+m)\rceil$, and from \Cref{alg:pilot} that a pilot run stops once some unresolved vertex has been scanned $\theta$ times.
The counts of a pilot run depend on the probability that a sample scans an unresolved vertex, so we compute it first.
Fix pools for the vertices of a resolved set $R$ with $s\in R$, with positive outflow $\widehat\Phi$, and consider samples of these pools, as in a pilot run or a resolving run, each with some layer $L\subseteq V\setminus R$, possibly empty.

\begin{lemma}\label{lem:scans}
For $v\notin R$, a sample scans $v$ with probability $Q_R(v)/\widehat\Phi$, regardless of the layer.
\end{lemma}

\begin{proof}
Consider a sample with the start $(A,u)$, which satisfies $\gamma_u(R)>0$.
By \Cref{lem:deferred}, the sample can be executed on a configuration drawn from the product measure conditional on $\mathrm{Out}_u$, and neither this measure nor the start law~\eqref{eq:start-law} depends on the layer.
By \Cref{lem:reach}, on every configuration, the sample scans $v$ exactly when the event defining $\mathrm{Reach}_R(u,A;v)$ occurs, regardless of the layer.
This includes the case $v\in A$, in which $v$ is never pushed and the event is impossible.
Hence the probability that a sample scans $v$ does not depend on its layer, and we may take the layer $\{v\}$.

Then this lemma follows from \Cref{lem:sample-mean}.
This is because $v$ is popped exactly when it is scanned.
If $D$ is the discovered set at its pop, then $\mathrm{Reach}_{R'}(v,D;v)=1$, where $R'=R\cup\{v\}$.
Thus the contribution $S$ in~\eqref{eq:contribution} for $y=v$ is the indicator that $v$ is scanned.
Moreover, as $v\in R'$, $\mathrm{Reach}_{R'}(u',A';v)=0$ for every snapshot $(A',u')$ with $u'\in R$ (\Cref{sec:predictions}), and $Q_{\mathrm{old}}=0$ in \Cref{lem:sample-mean} with the layer $\{v\}$ and $y=v$.
Applying \Cref{lem:sample-mean} finishes the proof.
\end{proof}

The next lemma shows that most pilot runs are good and do not take too long.

\begin{lemma}\label{lem:pilot}
Let $r\geq1$ and $i\geq1$, and consider the event, determined by $\mathcal H_{i-1}$, that round $r$ of the monitored algorithm is executed and its pilot run uses the random source $U_i$.
On this event, conditional on $\mathcal H_{i-1}$, the following hold.
\begin{enumerate}[label=(\alph*)]
  \item The pilot run stops with probability one.
  \item The number $T_r$ of its samples satisfies $\E[T_r\mid\mathcal H_{i-1}]\leq\theta\widehat\Phi_r/Q_*(r)$.
  \item The pilot run is not good with probability at most $4n\,e^{-\theta/200}$.
\end{enumerate}
\end{lemma}

\begin{proof}
We condition on $\mathcal H_{i-1}$ and the event that $U_i$ is the random source of the pilot run of round~$r$.
Then $R_r$ and the pools are fixed, round $r$ is accurate, and $0<Q_*(r)\leq\widehat\Phi_r$, by \Cref{lem:monitored}(a).
For $v\notin R_r$, let $\beta_v\defeq Q_r(v)/\widehat\Phi_r$, and let $\beta_*\defeq Q_*(r)/\widehat\Phi_r=\max_{v\notin R_r}\beta_v$, so that $0<\beta_*\leq1$.
Consider the infinite sequence of independent samples supplied by $U_i$, and let $N_k(v)$ be the number of the first $k$ samples that scan $v$.
The pilot run stops after the least $k$ with $\max_{v\notin R_r}N_k(v)=\theta$. Then $T_r=k$ and $L_r=\{v\notin R_r:\ N_{T_r}(v)\geq \theta/16\}$.
If there is no such $k$, the pilot run does not stop.
The counts $N_k(v)$ do not decrease with $k$.

By \Cref{lem:scans}, a sample scans each $v\notin R_r$ with probability $\beta_v$, independently of the other samples.
Fix $v_*\notin R_r$ with $Q_r(v_*)=Q_*(r)$, so that $\beta_{v_*}=\beta_*$.
The index of the $\theta$-th sample that scans $v_*$ is a sum of $\theta$ independent geometric random variables with success probability $\beta_*>0$, each with mean $1/\beta_*$.
It is finite with probability one, and the pilot run stops no later, since the count of $v_*$ equals $\theta$ there.
Therefore the pilot run stops with probability one, and $\E[T_r]\leq \theta/\beta_*=\theta\widehat\Phi_r/Q_*(r)$, which proves~(a) and~(b).

For~(c), fix $v\notin R_r$ and an integer $k\geq1$.
The count $N_k(v)$ is binomial with $k$ trials and success probability $\beta_v$, so the Chernoff bound gives, for all $a>0$ and $b>0$ with $b\geq k\beta_v$,
\begin{align}\label{eq:count-tails}
  \Pr\bigl[N_k(v)-k\beta_v\geq a\bigr]\leq\exp\Bigl(-\frac{a^2}{2(b+a)}\Bigr)
  \quad\text{and}\quad
  \Pr\bigl[k\beta_v-N_k(v)\geq a\bigr]\leq\exp\Bigl(-\frac{a^2}{2(b+a)}\Bigr).
\end{align}

Let $k_-\defeq\lfloor \theta/(2\beta_*)\rfloor$ and $k_+\defeq\lceil2\theta/\beta_*\rceil$.
By~\eqref{eqn:parameters}, $\theta\geq2048$, so $k_-\geq1$.
Suppose that the pilot run stops with 
\begin{itemize}
  \item $k_-<T_r\leq k_+$;
  \item $N_{k_-}(v)\geq\theta/16$ for every $v\notin R_r$ with $q_v\geq q_*(r)/2$; and
  \item $N_{k_+}(v)<\theta/16$ for every $v\notin R_r$ with $q_v<q_*(r)/128$.
\end{itemize}
Then $T_r\geq k_-+1>\theta/(2\beta_*)=\theta\widehat\Phi_r/(2Q_*(r))$.
Every $v\notin R_r$ with $q_v\geq q_*(r)/2$ has $N_{T_r}(v)\geq N_{k_-}(v)\geq \theta/16$, and so lies in $L_r$.
Every $v\notin R_r$ with $q_v<q_*(r)/128$ has $N_{T_r}(v)\leq N_{k_+}(v)<\theta/16$, and so lies outside $L_r$.
As $Q_*(r)>0$, the pilot run is good.

Therefore, if the pilot run is not good, then one of the following four events occurs.
We bound their probabilities.
\begin{enumerate}[label=(\arabic*)]
  \item \emph{$T_r\leq k_-$.} Then $N_{k_-}(v)\geq \theta$ for some $v\notin R_r$.
  For every $v\notin R_r$, the mean of $N_{k_-}(v)$ is $k_-\beta_v\leq\frac{\theta}{2\beta_*}\,\beta_*=\frac \theta2$, so $N_{k_-}(v)\geq \theta$ implies $N_{k_-}(v)-k_-\beta_v\geq\frac \theta2$.
  By~\eqref{eq:count-tails} with $a=b=\frac \theta2$, this has probability at most $\exp(-(\theta/2)^2/(2\theta))=e^{-\theta/8}$.
  A union bound over the fewer than $n$ unresolved vertices bounds the probability of this event by $n\,e^{-\theta/8}$.
  \item \emph{$T_r>k_+$, or the pilot run does not stop.} Then $N_{k_+}(v_*)<\theta$.
  The mean of $N_{k_+}(v_*)$ is $\mu\defeq k_+\beta_*\geq2\theta$, and $N_{k_+}(v_*)<\theta$ implies $\mu-N_{k_+}(v_*)>\mu-\theta$.
  By~\eqref{eq:count-tails} with $a=\mu-\theta\geq \theta$ and $b=\mu$, the probability is at most $\exp(-(\mu-\theta)^2/(2(2\mu-\theta)))$.
  For every $\kappa>0$, the function $x\mapsto(x-\kappa)^2/(2x-\kappa)$ has the derivative $2x(x-\kappa)/(2x-\kappa)^2\geq0$ for $x\geq \kappa$, so
  \begin{align}\label{eq:monotone}
    \frac{(x-\kappa)^2}{2x-\kappa}\geq\frac{(x_0-\kappa)^2}{2x_0-\kappa}
    \qquad\text{for all }x\geq x_0\geq \kappa.
  \end{align}
  With $\kappa=\theta$ and $x_0=2\theta$, where the right side is $\theta/3$, the probability of this event is at most $e^{-\theta/6}$.
  \item \emph{$N_{k_-}(v)<\theta/16$ for some $v\notin R_r$ with $q_v\geq q_*(r)/2$.}
  For such $v$, accuracy gives $Q_r(v)\geq\frac23q_v\geq\frac13q_*(r)$.
  Every $x\notin R_r$ has $Q_r(x)\leq\frac43q_x\leq\frac43q_*(r)$, so $Q_*(r)\leq\frac43q_*(r)$, and $\beta_v\geq\frac{q_*(r)}{3\widehat\Phi_r}\geq\frac{Q_*(r)}{4\widehat\Phi_r}=\frac{\beta_*}4$.
  As $k_-\geq \theta/(2\beta_*)-1$ and $\beta_v\leq1$, the mean of $N_{k_-}(v)$ is
  \begin{align*}
    \mu\defeq k_-\beta_v\geq\frac{\theta}{2\beta_*}\cdot\frac{\beta_*}4-1=\frac \theta8-1\geq\frac{7\theta}{64},
  \end{align*}
  where the last step uses $\theta\geq64$.
  Then $N_{k_-}(v)<\theta/16$ implies $\mu-N_{k_-}(v)>\mu-\theta/16$.
  By~\eqref{eq:count-tails} with $a=\mu-\theta/16>0$ and $b=\mu$, the probability is at most $\exp(-(\mu-\theta/16)^2/(2(2\mu-\theta/16)))$.
  By~\eqref{eq:monotone} with $\kappa=\theta/16$ and $x_0=7\theta/64$, the exponent is at least $(3\theta/64)^2/(2\cdot10\theta/64)=9\theta/1280$.
  A union bound over the fewer than $n$ unresolved vertices bounds the probability of this event by $n\,e^{-9\theta/1280}$.
  \item \emph{$N_{k_+}(v)\geq \theta/16$ for some $v\notin R_r$ with $q_v<q_*(r)/128$.}
  For such $v$, accuracy gives $Q_r(v)\leq\frac43q_v<\frac{q_*(r)}{96}$.
  A vertex $x\notin R_r$ with $q_x=q_*(r)$ has $Q_r(x)\geq\frac23q_*(r)$, so $Q_*(r)\geq\frac23q_*(r)$, and $\beta_v<\frac{q_*(r)}{96\widehat\Phi_r}\leq\frac{3Q_*(r)}{2\cdot96\,\widehat\Phi_r}=\frac{\beta_*}{64}$.
  As $k_+\leq2\theta/\beta_*+1$ and $\beta_*\leq1$, the mean of $N_{k_+}(v)$ is
  \begin{align*}
    \mu\defeq k_+\beta_v\leq\Bigl(\frac{2\theta}{\beta_*}+1\Bigr)\frac{\beta_*}{64}\leq\frac \theta{32}+\frac1{64}.
  \end{align*}
  As $\theta\geq3$, we have $\frac \theta{32}-\frac \theta{40}=\frac{3\theta}{80}-\frac \theta{32}=\frac \theta{160}\geq\frac1{64}$, so $\frac \theta{16}-\mu\geq\frac \theta{40}$ and $\mu\leq\frac{3\theta}{80}$.
  Thus $N_{k_+}(v)\geq \theta/16$ implies $N_{k_+}(v)-\mu\geq \theta/40$, and~\eqref{eq:count-tails} with $a=\theta/40$ and $b=3\theta/80$ bounds its probability by
  \begin{align*}
    \exp\Bigl(-\frac{(\theta/40)^2}{2(3\theta/80+\theta/40)}\Bigr)=\exp\Bigl(-\frac{\theta^2/1600}{\theta/8}\Bigr)=e^{-\theta/200}.
  \end{align*}
  A union bound over the fewer than $n$ unresolved vertices bounds the probability of this event by $n\,e^{-\theta/200}$.
\end{enumerate}
The smallest of the exponents $\theta/8$, $\theta/6$, $9\theta/1280$ and $\theta/200$ is $\theta/200$, so one of the four events occurs with probability at most $(3n+1)e^{-\theta/200}\leq4n\,e^{-\theta/200}$, which proves~(c).
\end{proof}

\subsection{Proof of \texorpdfstring{\Cref{lem:accuracy}}{the accuracy lemma}}\label{sec:accuracy-proof}

\begin{proof}[Proof of \Cref{lem:accuracy}]
By \Cref{lem:monitored}(e), the monitored algorithm uses finitely many random sources, and every stopped martingale $\widetilde Q_i(y)$ is eventually constant.
Let $\widetilde Q_\infty(t)$ be the final value of $\widetilde Q_i(t)$, which is the prediction for $t$ of the pools at the end of the monitored algorithm.
Let
\begin{align}\label{eq:event-C}
  \mathcal{C}\defeq\bigl\{\text{neither test halts the monitored algorithm}\bigr\}\cap\left\{\abs{\widetilde Q_\infty(t)-q_t}\leq\varepsilon q_t\right\}.
\end{align}
It is determined by the random choices of the uncapped algorithm, since the monitored algorithm uses the same random choices.

For $r\geq1$, let $\mathcal B_r$ be the event that round $r$ is executed and its pilot run is not good.
By \Cref{lem:monitored}(d), a halt by test~\ref{test:pilot} occurs only if $\mathcal B_r$ occurs for some $r\leq n-1$.
Fix $r\leq n-1$, and for $i\geq1$, let $\mathcal A_{r,i}$ be the event that $U_i$ is the random source of the pilot run of round $r$.
It is determined by $\mathcal H_{i-1}$, and on it, \Cref{lem:pilot}(c) gives $\Pr[\mathcal B_r\mid\mathcal H_{i-1}]\leq4n\,e^{-\theta/200}$.
Round $r$ executes at most one pilot run, which uses a single random source, so the events $\mathcal A_{r,i}$ for $i\geq1$ are disjoint, and $\mathcal B_r\subseteq\bigcup_{i\geq1}\mathcal A_{r,i}$.
Therefore
\begin{align*}
  \Pr[\mathcal B_r]
  &=\sum_{i\geq1}\Pr[\mathcal B_r\cap\mathcal A_{r,i}]
  =\sum_{i\geq1}\E\bigl[\ind_{\mathcal A_{r,i}}\Pr[\mathcal B_r\mid\mathcal H_{i-1}]\bigr]\\
  &\leq4n\,e^{-\theta/200}\sum_{i\geq1}\Pr[\mathcal A_{r,i}]
  \leq4n\,e^{-\theta/200}.
\end{align*}
A union bound over $r\leq n-1$ bounds the probability of a halt by test~\ref{test:pilot} by $4n^2e^{-\theta/200}$.

A halt by test~\ref{test:accurate} after the resolving run of a round $r$ means that some vertex $y$ satisfies $\abs{Q_{r+1}(y)-q_y}>\frac13q_y$.
By \Cref{lem:martingale}(a), $Q_{r+1}(y)=\widetilde Q_i(y)$ for the index $i$ of the random source of the last sample of this resolving run, so $\abs{\widetilde Q_i(y)-q_y}>\frac13q_y$.
\Cref{lem:inherited}(b), and a union bound over the $n$ vertices, bound the probability of a halt by test~\ref{test:accurate} by $2n\,e^{-\rho/(60n)}$.
Finally, if $\abs{\widetilde Q_\infty(t)-q_t}>\varepsilon q_t$, then $\abs{\widetilde Q_i(t)-q_t}>\varepsilon q_t$ for some $i$, and \Cref{lem:inherited}(c) with $y=t$ bounds the probability of this event by $2e^{-3\varepsilon^2\rho/(28n)}$.
Consequently
\begin{align*}
  \Pr[\mathcal{C}^c]\leq4n^2e^{-\theta/200}+2n\,e^{-\rho/(60n)}+2e^{-3\varepsilon^2\rho/(28n)}.
\end{align*}
By~\eqref{eqn:parameters}, $\theta=1024\lceil\log_2(n+m)\rceil\geq2048$ and $\rho\geq256n(\lceil\log_2(n+m)\rceil+\varepsilon^{-2})$.
An easy calculation gives $\Pr[\mathcal{C}^c]<\frac1{64}$.

%We bound the three terms.
%Let $k\defeq\lceil\log_2(n+m)\rceil$.
%By~\eqref{eqn:parameters}, $\theta=1024k\geq2048$, so $k\geq2$, and $\rho\geq256n(k+\varepsilon^{-2})$.
%By the standing assumption~\eqref{asm:standing}, every vertex is reachable from $s$, so $m\geq n-1$, and the definition of $k$ gives $2n-1\leq n+m\leq2^k$.
%Thus $n\leq2^{k-1}+\frac12$, and as $n$ and $2^{k-1}$ are integers, $n\leq2^{k-1}$.
%Therefore $4n^2\leq4^k$ and $2n\leq2^k$.
%For the first term, $\theta/200=5.12k$ and $\ln4<1.39$, so, as $k\geq2$,
%\begin{align*}
%  4n^2e^{-\theta/200}\leq e^{(1.39-5.12)k}\leq e^{-7.46}<\frac1{256}.
%\end{align*}
%For the second, $\rho\geq256n(k+\varepsilon^{-2})\geq256n(k+1)$, as $\varepsilon\leq1$, so $\rho/(60n)\geq\frac{64}{15}(k+1)$, and $\ln2<0.7$ gives
%\begin{align*}
%  2n\,e^{-\rho/(60n)}\leq e^{0.7k-\frac{64}{15}(k+1)}\leq e^{1.4-\frac{64}5}<\frac1{256},
%\end{align*}
%where the middle step uses that the exponent decreases in $k$ and that $k\geq2$.
%For the third, $\varepsilon^2\rho\geq256n$, so
%\begin{align*}
%  2e^{-3\varepsilon^2\rho/(28n)}\leq2e^{-192/7}<\frac1{256}.
%\end{align*}
%Thus $\Pr[\mathcal{C}^c]<\frac3{256}<\frac1{64}$.

Now assume that $\mathcal{C}$ occurs.
Neither test halts the monitored algorithm, so it coincides with the uncapped algorithm of \Cref{def:monitored}.
By \Cref{lem:monitored}(a), every executed round is accurate, and every pilot run passed test~\ref{test:pilot}, so it is good.
In particular, every pilot run stops, and its layer is not empty, so every executed round resolves at least one vertex.
The algorithm does not return in Step~\ref{step:outflow}, since every executed round $r$ has $\widehat\Phi_r>0$, by \Cref{lem:monitored}(a).
By \Cref{lem:monitored}(d), at most $n-1$ rounds are executed, so the algorithm returns in Step~\ref{step:return} after at most $n-1$ rounds, which proves~(a).
Its output is the total weight of the pool of $t$ in the final pools, which is their prediction for $t$ by~\eqref{eq:prediction-resolved}, as $t$ is resolved.
By \Cref{lem:monitored}(e), this prediction is $\widetilde Q_\infty(t)$, and~\eqref{eq:event-C} gives $\abs{\widehat{R}-q_t}\leq\varepsilon q_t$, which proves~(b).
\end{proof}

\section{Running time}\label{sec:work}

In this section we prove \Cref{lem:runtime}.
We first fix an implementation of \Cref{alg:main}, and then bound the expected amount of work that it does.

\paragraph{Implementation.}
A direct implementation would waste work in three places.
\begin{itemize}[leftmargin=*]
  \item \emph{Snapshots and pools}.
  The problem is that a snapshot is shared but the exploration changes it.
  The same snapshot $(A,u)$ in a pool may be drawn by many samples, and each sample starts with discovered set $A$ and adds the vertices it discovers, while the pool must keep $A$ unchanged for later draws.
  A record must also store the discovered set $D$ at the moment of the pop, although the sample keeps changing $D$ afterwards.
  With ordinary arrays, both would require copying a set of $n$ bits, so every sample and every record would cost $\Omega(n)$ operations, even if the sample discovers only a few vertices.
  We store discovered sets as persistent bit arrays~\cite{DSST89}, in which reading a bit costs $O(\log n)$ operations, and setting a bit creates a new version in $O(\log n)$ operations and leaves the old versions unchanged.
  Thus a sample starts from the stored version of $A$ in $O(1)$ operations, each membership test or discovery costs $O(\log n)$ operations, and a record keeps the current version of $D$ in $O(1)$ operations.
  The discovered sets then cost a sample $O(\log n)$ operations per arc it examines, instead of $\Omega(n)$.
  Each pool is an array of snapshots, stored with its size and common entry weight, so its total weight $\widehat q_u$ is read in $O(1)$ operations.
  A resolving run reserves $M$ uninitialised slots per new pool in $O(1)$ operations; these suffice because each sample records any vertex at most once.
  Uniform draws and appends cost $O(1)$ operations. The pool of $s$ uses one slot.
  \item \emph{Scanning a vertex} (\Cref{sec:segment-trees}). A scan must find the open arcs leaving a vertex, but testing every arc costs its outdegree, even when almost all arcs are closed.
  We store the outgoing arcs of each vertex in a segment tree, which draws the open arcs directly: a scan that returns $J$ open arcs costs $O(1+J\log(n+m))$ operations. The same tree draws the first scan conditioned on some cut arc being open.
  \item \emph{Drawing a start} (\Cref{sec:start-tables}). Every sample draws $u\in R$ with probability proportional to $\widehat q_u\gamma_u(R)$, and every round needs the outflow $\widehat\Phi$, the sum of these weights. Drawing $u$ by a pass over $R$ costs up to $n$ operations per sample.
  We keep the weights in a binary tree over the vertices, the \emph{start table}, whose root stores $\widehat\Phi$. A start draw costs $O(\log n)$ operations, and keeping the trees up to date as $R$ grows costs $O(n^2+m\log(n+m))$ operations in total.
\end{itemize}
Outgoing arcs and their probabilities are stored in arrays in scan order.
The stacks are lists with $O(1)$ pushes and pops; $R$, $L$ and the counts of a pilot run are arrays with $O(1)$ access.

\subsection{Scanning without the closed arcs}\label{sec:segment-trees}

A scan must find the open arcs leaving a vertex whose heads lie outside $R$.
Testing these arcs one by one would cost the outdegree of the vertex, even when almost all of them are closed.
The segment tree below draws the open arcs directly, and pays only for the arcs that it returns.

We first describe the structure for a fixed resolved set $R$, and at the end of this subsection how it is built for $R=\{s\}$ and updated when $R$ grows.

For every vertex $z$ with an outgoing arc, store its outgoing list as the leaves of a binary tree, its \emph{segment tree}, from left to right in the fixed scan order.
Split every list of at least two arcs into two parts whose lengths differ by at most one, and recurse on the two parts.
The depth is at most $\lceil\log_2 m\rceil$.
At each node, store the product of $1-p_e$ over the arcs below it whose heads lie outside $R$, with an empty product equal to one.
The shape of the trees never changes; only these products depend on $R$.
For a resolved vertex $u$, the arcs leaving $u$ whose heads lie outside $R$ are exactly its cut arcs, so the product at the root of its tree is its \emph{cut product} $1-\gamma_u(R)$.

The chosen arcs are those whose heads lie outside $R$.
For a node $I$, let $b_I$ be the probability that some chosen arc below it is open.
This is one minus its stored product.
If $I$ has children $I_1$ and $I_2$, then $b_I=b_{I_1}+b_{I_2}-b_{I_1}b_{I_2}$.
The values $b_I$ can therefore be read with $O(1)$ operations at every visited node.

To draw an ordinary list, initialise an empty output list and toss a coin of bias $b_I$ at the root.
If the coin passes, call $\mathrm{Open}(I)$. Otherwise the list stays empty.
For a vertex with no outgoing arc, return the empty list without querying a tree.
To draw a nonempty list under its conditional law, assume $b_I>0$ at the root and call $\mathrm{Open}(I)$ directly.
All coins used below are fresh and independent.
The call $\mathrm{Open}(I)$, used only when $b_I>0$, does the following.
At a leaf, append its arc to the output list.
Otherwise let $I_1,I_2$ be its left and right children, and let $a_1=b_{I_1}$ and $a_2=b_{I_2}$.
Mark $I_1$ with probability $a_1/b_I$.
If $I_1$ is marked, mark $I_2$ with probability $a_2$. If $I_1$ is not marked, mark $I_2$.
Call $\mathrm{Open}$ on the marked children, the left one first.

\begin{lemma}[Drawing the open arcs]\label{lem:scan-sampler}
\begin{enumerate}[label=(\alph*)]
  \item For every node $I$ with $b_I>0$, $\mathrm{Open}(I)$ appends the open chosen arcs below $I$ under their product law conditioned on a nonempty list, in scan order.
  \item Consequently, the ordinary procedure returns the open chosen arcs under their unconditioned product law, in scan order.
  \item If either procedure returns $J$ arcs, it uses $O(1+J\log(n+m))$ operations and fresh coins.
\end{enumerate}
\end{lemma}

\begin{proof}
(a)
We use induction on the height of $I$.
At a leaf, conditioning on a nonempty chosen list says that its arc is chosen and open, so appending this arc has the required law.
At an internal node, the events that the two children have nonempty chosen lists are independent, with probabilities $a_1,a_2$.
Their union has probability $b_I$.
Conditional on that union, the probability that the left child is nonempty is $a_1/b_I$, exactly its marking probability.

If the left child is nonempty, the conditioning on the union is redundant.
The left child's states have their product law conditioned on nonemptiness, and the right child's states retain their independent original law.
By induction, the call on $I_1$ supplies the first of these laws.
The independent coin of bias $a_2$ and, if it passes, the call on $I_2$ supply the ordinary product law on the right.
If the left child is empty, the union requires the right child to be nonempty.
The procedure then makes only the call on $I_2$, which supplies its conditional law by induction, while every chosen arc on the left is closed.
These two cases prove the conditional law at $I$.
They also show that every marked child $I'$ has $b_{I'}>0$, so every recursive call has the required positive denominator.

(b)
The ordinary root coin has bias $b_I$, the probability that the list is not empty, so~(b) follows from~(a).

(c)
Every visited node of $\mathrm{Open}$ lies on a path to a returned leaf, since its conditional list is nonempty.
Each such path has at most $1+\lceil\log_2 m\rceil$ nodes.
The root coin and each visited node use $O(1)$ operations and coins, so the total is $O(1+J\log(n+m))$.
The left-first traversal visits at most two tree paths between consecutive returned leaves, and at most one path before the first or after the last.
Thus the returned arcs and one additional step can each be charged $O(\log(n+m))$ operations and coins as the traversal proceeds.
\end{proof}

An ordinary scan uses the ordinary procedure, with fresh randomness, and processes the returned list in scan order.
For the first scan of a sample from $(A,u)$, call $\mathrm{Open}$ at the root directly.
Since $u\in R$, the chosen arcs are exactly its cut arcs, and the root has $b_I=\gamma_u(R)>0$.
By \Cref{lem:scan-sampler}(a), this gives their product law conditioned on nonemptiness, as required by Step~\ref{step:first-scan} of \Cref{alg:sample}.
By \Cref{lem:scan-sampler}(b), every later scan has the ordinary product law on the arcs whose heads lie outside $R$, as required by Step~\ref{step:scan} of \Cref{alg:sample}.
States of arcs whose heads are already discovered are not used by the exploration, and every vertex is scanned at most once.
Thus this implementation of the scans has the law of \Cref{alg:sample}.

\paragraph{Initialisation and updates.}
For $R=\{s\}$, all segment trees are built bottom-up with $O(n+m)$ operations.
For every vertex, store also a list of its entering arcs, with pointers to their leaves in the segment trees.
The products change only when a resolving run installs the pools of a layer $L$ and $R$ becomes $R\cup L$.
Then the arcs whose heads lie in $L$ leave the product, and no other arc changes.
So visit the entering arcs of every $v\in L$, set the leaf factor of each to one, and recompute the products on the path from that leaf to the root, with $O(\log(n+m))$ operations per arc.
Every vertex is resolved at most once, so every arc is updated at most once, and all updates together cost $O(m\log(n+m))$ operations.
Within a round, the segment trees do not change until its resolving run installs its new pools.

\subsection{The start table}\label{sec:start-tables}

Every sample draws a resolved vertex $u$ with probability proportional to $\widehat q_u\gamma_u(R)$, and every round reads the outflow $\widehat\Phi$, the sum of these weights.
The weights change only when a layer is resolved, so we keep them in a tree that supports draws in $O(\log n)$ operations.

For a resolved vertex $u$, the probability $\gamma_u(R)$ that at least one of its cut arcs is open is read with $O(1)$ operations from the cut product at the root of its segment tree (\Cref{sec:segment-trees}).
For a vertex with no outgoing arc, $\gamma_u(R)=0$.
The \emph{start table} is a separate binary tree, over the vertices instead of the arcs, built by the same balanced splitting rule.
Its leaf for $u$ has weight $\widehat{q}_u\gamma_u(R)$ if $u\in R$, and weight zero otherwise.
Every internal node stores the total weight below it, so the root stores the outflow $\widehat\Phi$.
For a start draw, when $\widehat\Phi>0$, walk down from the root, choosing each child with probability proportional to its weight.
The probabilities along the path telescope, so the probability of drawing $u$ is $\widehat{q}_u\gamma_u(R)/\widehat\Phi$.
The draw costs $O(\log n)$ operations.
Only resolved vertices with $\gamma_u(R)>0$ and positive pool weight can be drawn, so both the pool draw and the conditional first scan are defined.

\paragraph{Initialisation and updates.}
For $R=\{s\}$, the start table is built bottom-up with $O(n)$ operations.
The weights change only when a resolving run installs the pools of a layer $L$.
After the segment trees have been updated (\Cref{sec:segment-trees}), recompute every leaf weight from the current pools and cut products, and every internal sum from its two children, working bottom-up.
This recomputes every stored start weight, including those of the newly resolved vertices, with $O(n)$ operations.
Within a round, the start table does not change until its resolving run installs its new pools.

In Step~\ref{step:outflow} of \Cref{alg:main}, the outflow $\widehat\Phi$ is read at the root of the start table.
Each sample of a pilot run or a resolving run draws a vertex $u$ from this tree and a snapshot uniformly from its pool, then performs the conditional first scan.
This gives exactly the start law~\eqref{eq:start-law}.
Drawing $u$ costs $O(\log n)$ operations.

\subsection{Proof of \texorpdfstring{\Cref{lem:runtime}}{the runtime lemma}}\label{sec:work-proof}

We bound the expected number of steps of the algorithm over all rounds.
Outside the pilot runs and the resolving runs, the algorithm only builds and updates the segment trees (\Cref{sec:segment-trees}) and the start table (\Cref{sec:start-tables}).
Building them and the pool of $s$ takes $O(n+m)$ steps, all updates of the segment trees take $O(m)$ steps, and recomputing the start table takes $O(n)$ steps per round, so $O(n^2)$ over the at most $n-1$ rounds.
A pilot run or a resolving run takes $O(n)$ steps besides its samples, to set up its counts or its new pools, so the main work of a round is its number of samples times the cost of a sample.
\begin{itemize}[leftmargin=*]
  \item \emph{Cost of a sample.} With pools for $R$, a sample takes $O(C_R/\widehat\Phi)$ steps in expectation, for a quantity $C_R$ defined below in~\eqref{eq:C-def} (\Cref{lem:run-work}).
  \item \emph{Cost of a round.} The pilot run of round $r$ draws $\E[T_r]\leq\theta\widehat\Phi_r/Q_*(r)$ samples (\Cref{lem:pilot}(b)), and its resolving run draws $M_r\approx256\rho T_r/\theta$ samples, so about $256\rho\widehat\Phi_r/Q_*(r)$ samples in expectation. Each sample costs $O(C_r/\widehat\Phi_r)$ steps in expectation, where $C_r\defeq C_{R_r}$. In the product of the number of samples and this cost, the factors $\widehat\Phi_r$ cancel, so the samples of round $r$ take $O\bigl((\theta+\rho)C_r/Q_*(r)\bigr)$ steps in expectation, and round $r$ takes $O\bigl(n+(\theta+\rho)C_r/Q_*(r)\bigr)$ steps in expectation.
  \item \emph{All rounds.} The ratios satisfy $\sum_rC_r/Q_*(r)\leq4(n+P)$ (\Cref{lem:round-sums}).
\end{itemize}
Together, the expected number of steps is $O\bigl(m+n^2+(\theta+\rho)(n+P)\bigr)$, which is at most $W_0$.
Below we define steps, state the two lemmas, and prove \Cref{lem:runtime} from them.
The two lemmas are proved in \Cref{sec:run-round-proofs}.

\paragraph{Steps.}
With the data structures above, every basic action of the algorithm costs at most $O(\log(n+m))$ operations.
We therefore measure the work in \emph{steps}, a unit that hides this factor, and count the steps as follows.
\begin{itemize}[leftmargin=*]
  \item Drawing an ordinary or conditioned list of $J$ open arcs counts as $1+J$ steps, which includes initialising a sample at its first list.
  \item Processing a returned arc in a scan counts as one step, including the update of $D$ and the push.
  \item A pop counts as one step, including its record in a resolving run and the update of its count in a pilot run.
  \item Every other basic action counts as one step, such as a draw from the start table or from a pool, an update of a segment tree at one arc or of the start table at one node, or processing one arc or one vertex elsewhere.
\end{itemize}

\begin{lemma}[Cost of a step]\label{lem:step-cost}
Every step costs $O(\log(n+m))$ operations.
\end{lemma}

\begin{proof}
The open-list procedure costs $O(\log(n+m))$ operations per charged step, by \Cref{lem:scan-sampler}(c).
Updates of the segment trees and of the start table cost $O(\log(n+m))$ operations per step, and a persistent array costs $O(\log n)$ per bit.
Drawing from a pool and computing $M$ cost $O(1)$ operations, as taking the integer part is an arithmetic operation.
All other actions cost $O(1)$ operations.
\end{proof}

\paragraph{Two lemmas.}
For a resolved set $R$ and a vertex $v$, let
\begin{align*}
  \sigma_v^+(R)\defeq\sum_{e=(v,z):\,z\notin R}p_e
\end{align*}
be the expected number of open arcs leaving $v$ whose heads lie outside $R$, and, for pools of the vertices of $R$, define
\begin{align}\label{eq:C-def}
  C_R\defeq\sum_{v\in V}Q_R(v)\sigma_v^+(R).
\end{align}
The quantity $C_R$ measures the cost of a sample: $C_R/\widehat\Phi$ is the expected number of open arcs returned by the scans of a sample.
Indeed, a sample scans an unresolved vertex $v$ with probability $Q_R(v)/\widehat\Phi$ (\Cref{lem:scans}), and this scan returns $\sigma_v^+(R)$ open arcs in expectation, which gives the term of $v$ divided by $\widehat\Phi$.
The first scan is at a resolved vertex $u$, and it gives the term of $u$ divided by $\widehat\Phi$, where $Q_R(u)=\widehat q_u$.
This number also pays for the pops and the other steps of the sample, as the next lemma shows.

\begin{lemma}[Work of a run]\label{lem:run-work}
There is an absolute constant $c_1 > 0$ with the following property.
Let $R$ with $s\in R$ and pools for its vertices with positive outflow $\widehat\Phi$ be given, and let $C_R$ be the quantity in~\eqref{eq:C-def}, for these pools.
\begin{enumerate}[label=(\alph*)]
  \item The expected number of steps of one sample, with any layer $L\subseteq V\setminus R$, possibly empty, and including its start draws, is at most $c_1C_R/\widehat\Phi$. This includes the samples of a pilot run.
  \item The expected number of steps of a resolving run with $M$ samples, apart from updates of the segment trees and of the start table, is at most
  \begin{align*}
    c_1\Bigl(n+\frac{MC_R}{\widehat\Phi}\Bigr).
  \end{align*}
  \item A pilot run uses at most $c_1n$ steps apart from its samples.
\end{enumerate}
\end{lemma}

The second lemma bounds the ratios $C_r/Q_*(r)$, summed over the rounds.

\begin{lemma}[Sums over the rounds]\label{lem:round-sums}
On every execution of the monitored algorithm, the following hold, where $C_r\defeq C_{R_r}$ uses the pools at the start of round $r$.
\begin{enumerate}[label=(\alph*)]
  \item Every executed round $r$ satisfies $Q_*(r)\leq\widehat\Phi_r\leq C_r$.
  \item The sum over the executed rounds satisfies
  \begin{align*}
    \sum_r\frac{C_r}{Q_*(r)}\leq4(n+P).
  \end{align*}
\end{enumerate}
\end{lemma}

\paragraph{Proof of \texorpdfstring{\Cref{lem:runtime}}{the runtime lemma}.}
Before the rounds, the algorithm builds the segment trees and the start table and creates the pool of $s$, with $O(n+m)$ steps.
All arc updates of the segment trees together cost $O(m)$ steps.
Apart from these updates, its pilot run and its resolving run, a round costs $O(n)$ steps: it reads the outflow, and after its resolving run it recomputes the start table.
Every round that the algorithm completes resolves at least one vertex, as its pilot run stops with a count equal to $\theta\geq\theta/16$ (Step~\ref{step:pilot-loop} of \Cref{alg:pilot}), so it executes at most $n-1$ rounds, and these parts of the algorithm use at most $c_2(n+m+n^2)$ steps on every execution, for an absolute constant $c_2$.

Consider the monitored algorithm (\Cref{def:monitored}), in which only the ordinary steps are counted.
Let $\mathcal W$ be $c_2(n+m+n^2)$ plus the total number of steps of the pilot runs and the resolving runs of the monitored algorithm.
By~\eqref{eq:event-C}, neither test halts the monitored algorithm on $\mathcal C$, so there it coincides with the uncapped algorithm.
Hence, conditional on the event $\mathcal C$, the number of steps after Step~\ref{step:params} is at most $\mathcal W$.

Fix $r\leq n-1$ and $i\geq1$, and work conditionally on $\mathcal H_{i-1}$, on the event that $U_i$ is the random source of the pilot run of round $r$.
Then the resolved set $R_r$, its pools, $\widehat\Phi_r$, $C_r$ and $Q_*(r)$ are fixed, and $Q_*(r)>0$, by \Cref{lem:monitored}(a).
For $k\geq1$, let $\mathcal W_k$ be the number of steps of the $k$-th sample of the infinite sequence of samples supplied by $U_i$.
These samples are independent, and $\E[\mathcal W_k]\leq c_1C_r/\widehat\Phi_r$ by \Cref{lem:run-work}(a).
The pilot run uses the first $T_r$ of these samples, so they take $\sum_{k=1}^{T_r}\mathcal W_k$ steps.
We cannot simply multiply $\E[T_r]$ by the cost of a sample, because $T_r$ depends on the samples.
However, whether the $k$-th sample is used depends only on the samples before it: the event $\{T_r\geq k\}$ is determined by the first $k-1$ samples, so it is independent of $\mathcal W_k$.
We therefore write the sum as $\sum_{k\geq1}\ind\{T_r\geq k\}\mathcal W_k$.
All its terms are nonnegative, so Tonelli's theorem (linearity of expectation for infinite sums of nonnegative terms) lets us exchange the expectation and the sum, and the independence then gives
\begin{align*}
  \E\biggl[\sum_{k=1}^{T_r}\mathcal W_k\biggr]
  =\sum_{k\geq1}\E\bigl[\ind\{T_r\geq k\}\mathcal W_k\bigr]
  =\sum_{k\geq1}\Pr[T_r\geq k]\,\E[\mathcal W_k].
\end{align*}
This is Wald's identity.
Now $\E[\mathcal W_k]\leq c_1C_r/\widehat\Phi_r$ for every $k$, $\sum_{k\geq1}\Pr[T_r\geq k]=\E[T_r]$, and $\E[T_r]\leq\theta\widehat\Phi_r/Q_*(r)$ by \Cref{lem:pilot}(b), so
\begin{align*}
  \E\biggl[\sum_{k=1}^{T_r}\mathcal W_k\biggr]
  \leq\E[T_r]\,\frac{c_1C_r}{\widehat\Phi_r}
  \leq\frac{c_1\theta C_r}{Q_*(r)}.
\end{align*}
With its at most $c_1n$ other steps (\Cref{lem:run-work}(c)), the pilot run of round $r$ has at most $c_1(n+\theta C_r/Q_*(r))$ expected steps.
If the monitored algorithm executes the resolving run of round $r$, then its samples use fresh random sources, and its number of samples is $M_r=\lceil256\rho T_r/\theta\rceil\leq1+256\rho T_r/\theta$.
Conditional on the pilot run, \Cref{lem:run-work}(b) therefore bounds the expected number of steps of the resolving run by $c_1(n+M_rC_r/\widehat\Phi_r)$.
Averaging over the pilot run, and using $\E[T_r]\leq \theta\widehat\Phi_r/Q_*(r)$ (\Cref{lem:pilot}(b)) and $C_r/\widehat\Phi_r\leq C_r/Q_*(r)$, as $Q_*(r)\leq\widehat\Phi_r$ (\Cref{lem:round-sums}(a)), bounds the expected number of steps of the resolving run of round $r$, counted only when it is executed, by
\begin{align*}
  c_1\Bigl(n+\Bigl(1+\frac{256\rho\widehat\Phi_r}{Q_*(r)}\Bigr)\frac{C_r}{\widehat\Phi_r}\Bigr)
  \leq c_1\Bigl(n+\frac{257\rho C_r}{Q_*(r)}\Bigr).
\end{align*}
Taking expectations over $\mathcal H_{i-1}$ and summing over $i$, we obtain, for every $r\leq n-1$,
\begin{align*}
  &\E\bigl[\#\text{steps of the pilot run and the resolving run of round $r$}\bigr]\\
  &\qquad\leq c_1\,\E\biggl[\ind\{\text{round $r$ is executed}\}\biggl(2n+\frac{(\theta+257\rho)C_r}{Q_*(r)}\biggr)\biggr],
\end{align*}
where the LHS counts zero steps if round $r$ is not executed.
By \Cref{lem:monitored}(d), only rounds $r\leq n-1$ are executed, so summing over them and applying \Cref{lem:round-sums}(b) on every execution yields
\begin{align*}
  \E\,\mathcal W
  \leq c_2(n+m+n^2)+c_1\bigl(2n^2+4(\theta+257\rho)(n+P)\bigr).
\end{align*}
By~\eqref{eqn:parameters}, $\theta\geq2048$ and $\rho\geq256n\lceil\log_2(n+m)\rceil=n\theta/4$.
As $n\geq2$ by~\eqref{asm:standing}, this gives $\theta\leq2\rho$ and $n\leq\rho$, and hence $n^2\leq(n+P)\rho$.
Moreover, by the standing assumption~\eqref{asm:standing}, every vertex is reachable from $s$, so $n\leq m+1\leq2m$.
Therefore the right-hand side is at most $3c_2m+(c_2+1038c_1)(n+P)\rho$, which is at most $W_0$ of~\eqref{eq:W0} for $c_*\defeq3c_2+1038c_1$.
\qed

\subsection{Proofs of \texorpdfstring{\Cref{lem:run-work,lem:round-sums}}{the two lemmas}}\label{sec:run-round-proofs}

\begin{proof}[Proof of \Cref{lem:run-work}]
(a)
Let $J$ be the total number of open arcs returned by all scans of one sample.
By~\eqref{eq:start-law}, the first scanned vertex is $u$ with probability $\widehat q_u\gamma_u(R)/\widehat\Phi$.
For a cut arc $e$ from $u$, being open already implies that the cut list is nonempty, so its conditional open probability is $p_e/\gamma_u(R)$.
The first list consists of the open cut arcs leaving $u$, so, given the starting vertex $u$, its expected length is $\sum_ep_e/\gamma_u(R)=\sigma^+_u(R)/\gamma_u(R)$, where $e$ ranges over the cut arcs leaving $u$.
Averaging over the starting vertex, the expected length of the first list is
\begin{align*}
  \sum_{u}\frac{\widehat q_u\gamma_u(R)}{\widehat\Phi}\cdot\frac{\sigma_u^+(R)}{\gamma_u(R)}=\frac{1}{\widehat\Phi}\sum_{u\in R}\widehat q_u\sigma_u^+(R),
\end{align*}
where the first sum ranges over the vertices $u\in R$ with $\widehat q_u\gamma_u(R)>0$. Every other vertex $u\in R$ has $\widehat q_u=0$, or $\gamma_u(R)=0$ and hence no cut arc, so in both cases $\widehat q_u\sigma^+_u(R)=0$.
For every $v\notin R$, the probability of a scan is $Q_R(v)/\widehat\Phi$, by \Cref{lem:scans}.
Whether this scan begins is determined before any outgoing state of $v$ is drawn.
Its returned list then has expected length $\sigma_v^+(R)$, conditional on that preceding history.
Every vertex is scanned at most once.
The tower property therefore gives
\begin{align*}
  \E J
  &=\frac{\sum_{u\in R}\widehat q_u\sigma_u^+(R)}{\widehat\Phi}
       +\frac{\sum_{v\notin R}Q_R(v)\sigma_v^+(R)}{\widehat\Phi}=\frac{C_R}{\widehat\Phi},
\end{align*}
where the second equality uses~\eqref{eq:prediction-resolved}.

Every popped vertex was pushed through a returned arc and added to $D$ at that push, so it is pushed only once and there are at most $J$ pops.
Apart from the initial scan, every scan follows one of these pops, so there are at most $1+J$ scans.
Drawing their lists costs at most $1+2J$ steps: one per scan and one per returned arc.
Processing the arcs adds $J$ steps, the pops and their records add at most $J$, and drawing the starting vertex and snapshot adds two.
The first list is conditioned to be nonempty, so $J\geq1$, even if all its heads were already discovered.
The total number of steps of the sample, including its start draws, is therefore at most
\begin{align*}
  (1+2J)+J+J+2
  =3+4J
  \leq7J.
\end{align*}
Taking expectations bounds this work by $7C_R/\widehat\Phi$, which proves~(a) for every $c_1\geq7$.
A sample of a pilot run is a sample with the empty layer, and the update of a count is part of the step of its pop, so the same bound holds for it.

(b) and~(c)
A resolving run computes $M$ and $\widehat\Phi/M$, initialises the new pools for $L$ and installs them with $O(n)$ steps, and a pilot run sets its counts to zero and selects its layer with $O(n)$ steps.
A resolving run makes $M$ samples, so, by~(a) with $c_1=7$, their explorations and start draws use at most $7MC_R/\widehat\Phi$ expected steps.
Recording is already charged to the pops.
Adding these $O(n)$ steps proves~(b) and~(c) for a suitable constant $c_1\geq7$.
\end{proof}

Two properties of the monitored algorithm keep the sum in \Cref{lem:round-sums}(b) small.
First, the predictions of every executed round are within a constant factor of the reliabilities (\Cref{lem:monitored}(a)), so $C_r$ is at most a constant times the same sum with the reliabilities in place of the predictions.
Second, the largest unresolved reliability more than halves at every round (\Cref{lem:monitored}(c)), so, for a fixed vertex $v$, the ratios $q_v/q_*(r)$ over the rounds form a geometric sum.
Part~(a) rests on the fact that the outflow is at most $C_R$.
For a resolved vertex $u$, $Q_R(u)=\widehat q_u$ by~\eqref{eq:prediction-resolved}, and the union bound gives $\gamma_u(R)\leq\sigma_u^+(R)$.
All terms in $C_R$ are nonnegative, so
\begin{align}\label{eq:outflow-work}
  \widehat\Phi=\sum_{u\in R}\widehat q_u\gamma_u(R)
  \leq\sum_{u\in R}\widehat q_u\sigma_u^+(R)
  \leq C_R.
\end{align}
We group all terms by their unresolved heads and use the incoming reliability bound~\eqref{eq:incoming-reliability} below.
For a vertex $v$, let $\sigma_v^-\defeq\sum_{e=(u,v)}p_e$ be its \emph{expected indegree}, the expected number of open arcs entering it.
Recall that $P=\sum_{e\in E}p_e$ is the total open probability. Every arc enters exactly one vertex, so $\sum_v\sigma_v^-=P$.

\begin{proof}[Proof of \Cref{lem:round-sums}]
First we show that, for every vertex $v$,
\begin{align}\label{eq:incoming-reliability}
  \sum_{e=(u,v)}p_eq_u\leq(1+\sigma_v^-)q_v.
\end{align}
As $\sum_{e=(u,v)}p_eq_v=\sigma_v^-q_v$, this inequality says that $\sum_{e=(u,v)}p_e(q_u-q_v)\leq q_v$.
The idea is as follows.
The difference $q_u-q_v$ comes from the configurations in which $s$ reaches $u$ but not $v$.
In such a configuration $e$ is closed, and opening it gives a new configuration, in which $s$ reaches $v$ and every open path from $s$ to $v$ uses $e$.
A new configuration has at most one arc entering $v$ with this property, so the new configurations obtained from different arcs are different, and their total probability is at most $q_v$.

For an arc $e=(u,v)$, let $\mathcal E_e$ be the event that, after $e$ is deleted, $s$ reaches $u$ but does not reach $v$.
This event uses only the states of the other arcs, so it is independent of the state of $e$.

We first bound the difference by $\Pr[\mathcal E_e]$.
If $s$ reaches $u$ but not $v$, then $e$ is closed, since an open $e$ would make $v$ reachable.
Deleting this closed arc changes no open path, so $\mathcal E_e$ occurs.
Thus, whenever $s$ reaches $u$, either $s$ reaches $v$ or $\mathcal E_e$ occurs, and the union bound gives
\begin{align*}
  q_u-q_v\leq\Pr[\mathcal E_e].
\end{align*}

Next we show that the events $\mathcal E_e\cap\{\omega_e=1\}$, over the arcs $e$ entering $v$, are disjoint and all imply $s\reach v$.
Suppose that this event occurs for $e=(u,v)$, and choose an open path from $s$ to $u$ in the graph with $e$ deleted.
It avoids $v$, because $v$ is not reachable in that graph.
Appending the open arc $e$ gives an open path to $v$, so $s\reach v$, and the only arc of this path entering $v$ is $e$.
Deleting a different arc $f$ entering $v$ leaves this path intact, so $s$ still reaches $v$, and $\mathcal E_f$ does not occur.
By the independence of $\mathcal E_e$ and the state of $e$, and this disjointness,
\begin{align*}
  \sum_{e=(u,v)}p_e\Pr[\mathcal E_e]
  =\sum_{e=(u,v)}\Pr[\mathcal E_e\cap\{\omega_e=1\}]
  \leq q_v.
\end{align*}

Finally, multiplying $q_u-q_v\leq\Pr[\mathcal E_e]$ by $p_e$ and summing over all arcs entering $v$ gives
\begin{align*}
  \sum_{e=(u,v)}p_eq_u-\sigma_v^-q_v
  \leq\sum_{e=(u,v)}p_e\Pr[\mathcal E_e]\leq q_v.
\end{align*}
Adding $\sigma_v^-q_v$ proves~\eqref{eq:incoming-reliability}.

We now prove the two parts of the lemma.
Part~(a) is the bound $Q_*(r)\leq\widehat\Phi_r\leq C_r$, and part~(b) is the bound $\sum_rC_r/Q_*(r)\leq4(n+P)$, which uses~\eqref{eq:incoming-reliability}.

(a)
Fix an executed round $r$.
By \Cref{lem:monitored}(a), $Q_*(r)\leq\widehat\Phi_r$, and~\eqref{eq:outflow-work} gives $\widehat\Phi_r\leq C_r$.

(b)
Every executed round $r$ is accurate, by \Cref{lem:monitored}(a), so $Q_r(u)\leq\frac43q_u$ for every vertex $u$.
Every arc in the definition of $C_r$ has its head outside $R_r$.
Regrouping by these heads and applying~\eqref{eq:incoming-reliability} gives
\begin{align*}
  C_r
  &=\sum_{v\notin R_r}\sum_{e=(u,v)}p_eQ_r(u)\leq\frac43\sum_{v\notin R_r}\sum_{e=(u,v)}p_eq_u
   \leq\frac43\sum_{v\notin R_r}(1+\sigma_v^-)q_v.
\end{align*}
A vertex $x\notin R_r$ with $q_x=q_*(r)$ has $Q_r(x)\geq\frac23q_*(r)$, so $Q_*(r)\geq\frac23q_*(r)$, and
\begin{align*}
  \frac{C_r}{Q_*(r)}\leq2\sum_{v\notin R_r}(1+\sigma_v^-)\frac{q_v}{q_*(r)}.
\end{align*}

Fix a vertex $v$ and consider the executed rounds in which $v$ is unresolved.
The resolved set only grows, so, if there is any such round, these rounds are the rounds $1,2,\ldots,r_v$, where $r_v$ is the last one.
Then $q_v\leq q_*(r_v)$.
For $r<r_v$, the rounds $r,r+1,\ldots,r_v$ are executed, so \Cref{lem:monitored}(c) gives $q_*(r)>2^{r_v-r}q_*(r_v)$.
Hence
\begin{align*}
  \sum_{r:\,v\notin R_r}\frac{q_v}{q_*(r)}
  \leq\sum_{r=1}^{r_v}\frac{q_*(r_v)}{q_*(r)}
  \leq\sum_{r=1}^{r_v}2^{-(r_v-r)}
  <\sum_{j\geq0}2^{-j}=2,
\end{align*}
where the sum on the left ranges over the executed rounds in which $v$ is unresolved.
If there is no such round, this sum is zero and the same bound holds.
Applying this geometric bound with the nonnegative weight $1+\sigma_v^-$ gives
\begin{align*}
  \sum_r\frac{C_r}{Q_*(r)}
  &\leq2\sum_v(1+\sigma_v^-)\sum_{r:\,v\notin R_r}\frac{q_v}{q_*(r)}\leq4\sum_v(1+\sigma_v^-)=4(n+P),
\end{align*}
where all the round sums range only over executed rounds.
\end{proof}

\section{Undirected graphs}\label{sec:undirected}

In this section we prove \Cref{thm:undirected} by applying \Cref{thm:expdeg}.
We first transform undirected graphs to bi-directed graphs, using a result of Ball \cite[Corollary~1]{Ball80}.
Then we prove a sparsification lemma, \Cref{lem:sparsification},
which allows us to construct a new graph that preserves the two-terminal reliability while bounding the total open probability.

To avoid a blow-up of notations, we use $H$, $T$, $\delta$, $\alpha_e$, $U_e$, $a(\cdot,\cdot)$, $\mathcal E$ and $\mathcal Q$ locally, unrelated to their uses in the previous sections.
For an undirected graph $G$ whose edges are open independently, each with its own probability, and two vertices $x$ and $y$, let $\Rel_G(x,y)$ be the probability that $x$ and $y$ are connected by a path of open edges, with $\Rel_G(x,x)=1$.
When $G$ is clear from the context, we may omit the subscript.

\subsection{From undirected to bi-directed}\label{sec:bidirection}

Let $G$ be a simple undirected graph in which every edge $f$ has an open probability $p_f$.
Let $\vec G$ be the directed graph on the same vertex set that has, for every edge $f$ with endpoints $u$ and $v$, the two arcs $(u,v)$ and $(v,u)$, each with the open probability $p_f$ of the edge.
In $\vec G$ the two arcs of an edge are open independently, while in $G$ the edge is a single random variable.
Ball \cite{Ball80} showed that going from $G$ to $\vec G$ preserves the two-terminal reliability.

\begin{lemma}[{\cite[Corollary~1]{Ball80}}]\label{lem:bidirected}
For every two distinct vertices $s$ and $t$ of $G$, the probability that $s$ and $t$ are connected by a path of open edges in $G$ equals the probability that $s$ reaches $t$ by a directed path of open arcs in $\vec G$.\footnote{This is part~(1) of \cite[Corollary~1]{Ball80}.
Its part~(2) states that, for every vertex $r$, the graph $G$ is connected with the same probability as $r$ reaches all vertices in $\vec G$.
Guo and Jerrum~\cite[Section~5]{GuoJerrum2019} give a coupling proof of this all-terminal result, by an exploration from $r$ that never examines both arcs of an edge.
For a common open probability of all edges, Linusson~\cite[Lemma~2.1]{Linusson2009} gives the stronger statement that, for every vertex $s$, the set of vertices connected to $s$ via an open path in $G$ and the set of vertices reached from $s$ in $\vec G$ have the same distribution, and attributes it to McDiarmid~\cite{McDiarmid1981}.
}
\end{lemma}

\begin{remark}[Multigraphs]\label{rem:bidirected-multigraph}
\Cref{lem:bidirected} extends to undirected graphs $G$ with parallel edges without change.
This is because effectively we can replace the parallel edges $f_1,\ldots,f_k$ between two vertices by one edge with open probability $1-\prod_{i=1}^k(1-p_{f_i})$ while preserving the two-terminal reliability.
The same construction holds with $\vec G$ too.
\end{remark}

\subsection{Proof of \texorpdfstring{\Cref{thm:undirected}}{Theorem 2}}\label{sec:undirected-proof}

Let $G=(V,E)$ be an undirected graph with $n\geq 2$ vertices and $m$ edges, with rational open probabilities $\boldsymbol p\in(0,1]^E$.
Here we allow parallel edges and self-loops, and $G$ need not be connected.
A partition $\+P$ partitions $V$ into \emph{blocks}.
The \emph{quotient} $G/\mathcal P$ is the undirected multigraph whose vertices are the blocks.
For each edge $e$ of $G$ between two distinct blocks, it has one edge between these blocks with the original probability $p_e$.
Edges within a block are deleted, and parallel edges between blocks are kept separately.
For a vertex $x$, let $[x]$ be the block that contains $x$.
%The partition is obtained from the connected components of independent random subgraphs, as described .

The key lemma of this section is the following one, proved in \Cref{sec:sparsification-proof}.
\begin{lemma}[Sparsification]\label{lem:sparsification}
Let $\delta\in(0,1)$ and let $T\defeq\lceil\log_2(64n^2/\delta)\rceil$.
There is a randomised algorithm that, given $G$, $\boldsymbol p$ and $T$, in $O(T(m+n\log n))$ time returns either a failure symbol or a partition $\mathcal P$ of $V$ such that the edges of $G/\mathcal P$ have total open probability at most $64T^2(n-1)$.
Moreover, with probability at least $1-\frac1{32}$, it returns $\+P$ for which every two vertices $s$ and $t$ satisfy
\begin{align}\label{eqn:sparsification}
  (1-\delta)\Rel_{G/\mathcal P}([s],[t])\leq\Rel_G(s,t)\leq\Rel_{G/\mathcal P}([s],[t]).
\end{align}
\end{lemma}

%The comparison is relative: it holds however small the reliabilities are.
What is sparsified here is the total open probability, not the number of edges.
In fact, the quotient $G/\+P$ may have as many edges as $G$.
Only their total open probability is controlled, which is what the running-time bound of \Cref{thm:expdeg} needs.
Unlike the directed reduction, this preprocessing may fail, and its failure probability will be included in the proof of \Cref{thm:undirected}.
For all-terminal reliability, Feng, Guo and Zhang~\cite[Section~2 and Theorem~19 of Appendix~A]{FengGuoZhang2026} use a similar idea of sparsification.
This is achieved by contracting highly connected blocks in the sense of strengths by Bencz\'ur and Karger~\cite{BK15}.
A strength-based deterministic contraction can also work for \Cref{lem:sparsification},
but we present a simpler randomised construction instead, in \Cref{sec:sparsification-proof}.

\begin{proof}[Proof of \Cref{thm:undirected}]
Let the input be an undirected graph $G$ with $n$ vertices and $m$ edges, rational open probabilities in $(0,1]$, terminals $s$ and $t$, and $\varepsilon\in(0,1]$.
If $s=t$, return $1$.
Otherwise $n\geq2$, and we apply \Cref{lem:sparsification} to $G$ with $\delta\defeq\varepsilon/8$ so that $T=\lceil\log_2(512n^2/\varepsilon)\rceil$.

If the algorithm of \Cref{lem:sparsification} returns the failure symbol, return $0$.
Otherwise construct the quotient $H\defeq G/\mathcal P$, which costs $O(n+m)$ operations.
%, by listing the blocks as its vertices and scanning every input edge once, retaining it exactly when its endpoints lie in distinct blocks.
It has $n_H\defeq\abs{\mathcal P}\leq n$ vertices, $m_H\leq m$ edges, and total open probability $P_H\leq64T^2(n-1)$, by \Cref{lem:sparsification}.
Let $q\defeq\Rel_G(s,t)$ and $q_H\defeq\Rel_H([s],[t])$.
Let $\mathcal S$ be the event, of probability at least $1-1/32$, that $\+P$ is returned and \eqref{eqn:sparsification} holds,
namely, $(1-\delta)q_H\leq q\leq q_H$.
In the following we condition on $\+S$.

%In the accuracy argument below we work on $\mathcal S$.
%We use that, as $\varepsilon\leq 1$,
%\begin{align}\label{eq:eps-product}
%  (1+\varepsilon)\Bigl(1-\frac{\varepsilon}{8}\Bigr)=1+\frac78\varepsilon-\frac18\varepsilon^2\geq 1+\frac34\varepsilon\geq 1+\frac{\varepsilon}{4}.
%\end{align}

If $[s]=[t]$, return $1$.
It is easy to verify that when $\+S$ holds this is accurate.
%When $\mathcal S$ holds this is accurate: then $q_H=1$, so $1-\delta\leq q\leq 1$, and $(1-\varepsilon)q\leq 1\leq q/(1-\delta)\leq(1+\varepsilon)q$ by~\eqref{eq:eps-product}.

Otherwise, let $\vec H$ be the bi-directed graph obtained from $H$ as in \Cref{lem:bidirected}.
It has $n_H$ vertices, $2m_H$ arcs and total open probability $2P_H$.
By \Cref{lem:bidirected} and \Cref{rem:bidirected-multigraph}, applied to $H$ and its distinct vertices $[s]$ and $[t]$, the probability that $[s]$ reaches $[t]$ in $\vec H$ is $q_H$.
Run the algorithm of \Cref{thm:expdeg} on $\vec H$ with the terminals $[s]$ and $[t]$ and the error parameter $\varepsilon/4$.
Return its output $\widehat{R}$.
By~\eqref{eq:success}, $\widehat{R}$ is accurate with probability at least $1-5/64$.
Thus, conditional on $\+S$, with probability at least $59/64$ we have $(1-\varepsilon/4)q_H\leq\widehat{R}\leq(1+\varepsilon/4)q_H$, and then
\begin{align*}
  (1-\varepsilon)\,q\leq\Bigl(1-\frac{\varepsilon}{4}\Bigr)q_H\leq\widehat{R}\leq\Bigl(1+\frac{\varepsilon}{4}\Bigr)q_H\leq(1+\varepsilon)\,q.
\end{align*}
%The preprocessing fails to give $\mathcal S$ with probability at most $1/32$, and the subsequent directed estimate fails with probability at most $5/64$.
By the union bound, the returned estimate is within $1\pm\varepsilon$ of $q$ with probability at least $3/4$.

For the running time, \Cref{thm:expdeg}, with the error parameter $\varepsilon/4$, runs on $\vec H$ in time $\widetilde{O}(m_H+n_H(n_H+2P_H)\varepsilon^{-2})$.
Whenever the quotient is constructed, $n_H\leq n$, $n_H+2P_H\leq n+128T^2(n-1)$ and $T=O(\log(n/\varepsilon))$, so this time is $\widetilde{O}(m+n^2\varepsilon^{-2})$, where $\widetilde{O}$ hides the factor $T^2$.
The sparsification uses $O(T(m+n\log n))$ operations, by \Cref{lem:sparsification}, and constructing the quotient and $\vec H$ costs $O(n+m)$ operations.
This finishes the proof.
\end{proof}

\subsection{Proof of \texorpdfstring{\Cref{lem:sparsification}}{the sparsification lemma}}\label{sec:sparsification-proof}

Throughout this subsection, $G$, $\boldsymbol p$, $\delta$ and $T$ are as in \Cref{lem:sparsification}.
The proof has three parts: the construction of the partition, the bound on its total open probability, and the comparison of reliabilities.

\begin{proof}[Proof of \Cref{lem:sparsification}]

  We first construct the partition by random subgraphs.
  Let $\alpha_e\defeq p_e/T$ for every edge $e$.
Draw $T$ independent random subgraphs $G_1,\ldots,G_T$ in which each edge $e$ is present with probability $\alpha_e$, independently of the other edges.
In each $G_i$, compute its connected components and give each of them a distinct integer label.
The partition $\mathcal P$ puts vertices $u$ and $v$ in the same block exactly when their component labels agree in every $G_i$.
This is an equivalence relation: agreement of all $T$ labels is reflexive, symmetric and transitive.
It is computed by sorting the vertices by their $T$-tuples of labels.

Let $E_{\times}$ be the set of edges of $G$ between distinct blocks of $\mathcal P$.
If $\sum_{e\in E_{\times}}p_e>64T^2(n-1)$, return the failure symbol. Otherwise return $\mathcal P$.
Generating the $T$ subgraphs and finding their components costs $O(T(n+m))$ operations.
Comparing two tuples costs $O(T)$ operations, so sorting them costs $O(Tn\log n)$.
One further pass over the edges computes the sum over $E_{\times}$.
The total is $O(T(m+n\log n))$, which is the claimed bound.

We still need to bound the probability of success and the reliability bound.
The algorithm succeeds when the total open probability is bounded, an event we analyse next.

\paragraph{Success probability.}
Consider one of the random subgraphs, $G_i$, and fix an edge $e=uv$.
Let $\mathcal D_e$ be the event that $u$ and $v$ are disconnected after $e$ is deleted from $G_i$.
This event depends only on the other edges and is independent of the state of $e$.
The two endpoints $u$ and $v$ are disconnected in the random subgraph exactly when $\mathcal D_e$ occurs and $e$ is absent.
The edge $e$ is a bridge of the random subgraph exactly when $\mathcal D_e$ occurs and $e$ is present.
Thus
\begin{align}
  \alpha_e\Pr[u\not\leftrightarrow v\text{ in }G_i]
  &=\alpha_e(1-\alpha_e)\Pr[\mathcal D_e]\notag\\
  &=(1-\alpha_e)\Pr[e\text{ is a bridge of }G_i]
  \leq\Pr[e\text{ is a bridge of }G_i].\label{eq:bridge}
\end{align}
No cycle contains a bridge.
In particular, a self-loop is never a bridge and neither is an edge with a parallel copy in $G_i$.
Thus the bridges form a forest and there are at most $n-1$ of them.
Summing~\eqref{eq:bridge} over the edges, the right-hand sides add up to the expected number of bridges of $G_i$, which is at most $n-1$.
As $p_e=T\alpha_e$, this gives\footnote{If $\alpha_e=\alpha$ for all $e$, this says that the expected number of edges of $G$ whose endpoints lie in different components of the random subgraph is at most $(n-1)/\alpha$.
These edges are light with respect to a minimum spanning forest of the random subgraph, so in this case the bound, with $n$ in place of $n-1$, also follows from the sampling lemma of Karger, Klein and Tarjan~\cite[Lemma~2.1]{KKT95}, which bounds the expected number of such light edges by $n/\alpha$.
Inequality~\eqref{eq:cut-edges} extends this special case to edge-dependent probabilities.}
\begin{align}\label{eq:cut-edges}
  \sum_{e=uv}p_e\Pr[u\not\leftrightarrow v\text{ in }G_i]\leq T(n-1).
\end{align}
An edge lies in $E_{\times}$ exactly when its endpoints are disconnected in at least one of $G_1,\ldots,G_T$.
The union bound and~\eqref{eq:cut-edges} therefore imply
\begin{align*}
  \E\biggl[\sum_{e\in E_{\times}}p_e\biggr]
  \leq\sum_{i=1}^T\sum_{e=uv}p_e\Pr[u\not\leftrightarrow v\text{ in }G_i]
  \leq T^2(n-1).
\end{align*}
By Markov's inequality,
\begin{align}\label{eq:failure-symbol}
  \Pr[\text{the algorithm returns the failure symbol}]\leq\frac1{64}.
\end{align}
The edges of $G/\mathcal P$ are precisely the edges of $E_{\times}$, each with its original probability.
Hence, whenever the algorithm returns a partition, the total open probability of the quotient equals $\sum_{e\in E_{\times}}p_e\leq64T^2(n-1)$, as claimed.

\paragraph{The reliability bound.}
For two vertices $u,v\in V$, let $a(u,v)$ be the probability that they are connected in a random subgraph of $G$ in which every edge $e$ is present with probability $\alpha_e$, independently of the other edges, as in each $G_i$.
Define the event
\begin{align*}
  \mathcal Q\defeq\bigl\{a(u,v)\geq1/2\text{ for all }u,v\in V\text{ with }[u]=[v]\bigr\}.
\end{align*}
It depends only on the $T$ subgraphs used to construct the partition.
We claim that
\begin{align}\label{eq:Q-bound}
  \Pr[\mathcal Q^c]\leq\frac1{64}.
\end{align}
To see the claim, let $u$ and $v$ be two vertices with $a(u,v)<1/2$.
Independence of $G_1,\ldots,G_T$ gives
\begin{align*}
  \Pr\bigl[[u]=[v]\bigr]=a(u,v)^T\leq2^{-T}.
\end{align*}
There are at most $n^2$ such pairs, and $n^22^{-T}\leq\delta/64\leq1/64$ by the choice of $T$.
The union bound therefore gives~\eqref{eq:Q-bound}.

Fix $\mathcal P$ belonging to $\mathcal Q$.
For $u,v\in V$, let $u\leftrightarrow v$ be the event, under the original product measure on $G$, that $u$ and $v$ are connected by a path of open edges of $G$.
Define the event
\begin{align*}
  \mathcal E\defeq\bigl\{u\leftrightarrow v\text{ for all }u,v\in V\text{ with }[u]=[v]\bigr\}.
\end{align*}
The paths from $u$ to $v$ may leave the block $[u]$.
We claim that
\begin{align}\label{eq:E-bound}
  \Pr[\mathcal E^c]\leq\delta.
\end{align}
To see the claim, fix $u,v\in V$ with $[u]=[v]$.
For the analysis only, let $G'_1,\ldots,G'_T$ be independent random subgraphs of $G$ in which each edge $e$ is present with probability $\alpha_e$, and let $W$ be the union of their edges. 
The algorithm does not generate them.
The events $\{e\in W\}$ for distinct edges $e$ depend on disjoint sets of independent variables, so they are independent.
Moreover, $\Pr[e\in W]=1-(1-\alpha_e)^T\leq T\alpha_e=p_e$.
Therefore we can couple $W$ with a random open subgraph of $G$ with the probabilities $p_e$, such that all of $W$ is open.
This can be done by, for example, first drawing a uniform number $U_e$ for every edge $e$, and then including $e$ in $W$ when $U_e\leq\Pr[e\in W]$, and opening $e$ when $U_e\leq p_e$.
%This realises $W$ and the open subgraph of $G$ with their laws, and $W$ is contained in the open subgraph.

As connectivity is increasing, $\Rel_G(u,v)\geq\Pr[u\leftrightarrow v\text{ in }W]$.
If $u$ and $v$ are connected in some $G'_i$, they are connected in $W$, and the $T$ events that $u$ and $v$ are connected in $G'_i$ are independent, each with probability $a(u,v)\geq\frac12$.
Hence,
\begin{align*}
  1-\Rel_G(u,v)\leq\Pr[u\not\leftrightarrow v\text{ in }W]\leq(1-a(u,v))^T\leq2^{-T}.
\end{align*}
There are at most $n^2$ pairs to check, so the union bound and the choice of $T$ give
\begin{align*}
  \Pr[\mathcal E^c]\leq n^2 2^{-T}\leq\delta/64\leq\delta.
\end{align*}

Now for $s$ and $t$, couple the configuration on $G/\mathcal P$ with that on $G$ by giving every quotient edge the state of its corresponding edge of $G$.
Distinct quotient edges correspond to distinct input edges, so these states are independent and have the required original probabilities.
Let $\mathcal X\defeq\{[s]\leftrightarrow[t]\text{ in the random subgraph of }G/\mathcal P\}$ and $\mathcal Y\defeq\{s\leftrightarrow t\text{ in the random subgraph of }G\}$.
Then, $\Pr[\mathcal X]=\Rel_{G/\mathcal P}([s],[t])$ and $\Pr[\mathcal Y]=\Rel_G(s,t)$.

If $\mathcal Y$ occurs, then the blocks of the vertices of an open path from $s$ to $t$, with repeated consecutive blocks removed, form a walk from $[s]$ to $[t]$ along open edges of $G/\mathcal P$.
Therefore $\mathcal Y\subseteq\mathcal X$, which proves the upper bound.

Conversely, suppose that $\mathcal X$ and $\mathcal E$ occur, and let $[s]=B_0,B_1,\ldots,B_j=[t]$ be a path of open edges of $G/\mathcal P$.
For $1\leq i\leq j$, let $x_iy_i$ be the input edge corresponding to the $i$-th edge of this path, which is open in the coupling, with $x_i\in B_{i-1}$ and $y_i\in B_i$, and let $y_0\defeq s$ and $x_{j+1}\defeq t$.
For $0\leq i\leq j$, the vertices $y_i$ and $x_{i+1}$ lie in the block $B_i$.
On $\mathcal E$ they are equal or connected by a path of open edges of $G$.
Concatenating these paths and the edges $x_iy_i$ gives a walk of open edges of $G$ from $s$ to $t$, so $\mathcal X\cap\mathcal E\subseteq\mathcal Y$.

The events $\mathcal X$ and $\mathcal E$ are increasing, since opening edges of $G$ does not destroy paths of open edges.
They need not be independent, because paths joining vertices of a block may leave that block.
Harris's inequality (\Cref{lem:harris}) and~\eqref{eq:E-bound} therefore give
\begin{align*}
  \Rel_G(s,t)=\Pr[\mathcal Y]\geq\Pr[\mathcal X\cap\mathcal E]
  \geq\Pr[\mathcal X]\Pr[\mathcal E]
  \geq(1-\delta)\Rel_{G/\mathcal P}([s],[t]).
\end{align*}
%This holds for all $s,t$ for every partition obtained on $\mathcal Q$.
By~\eqref{eq:failure-symbol}, the algorithm fails with probability at most $1/64$.
Moreover, $\Pr[\mathcal Q^c]\leq1/64$ by~\eqref{eq:Q-bound},
which implies that the reliability bounds hold with probability at least $1-1/32$.
\end{proof}

\begin{remark}\label{rem:backbone}
  The contraction cannot be replaced by a sparse spanning subgraph with the same probabilities,
  in the style of the backbone construction \cite{ChechikEtAl2012}.
Let $G$ be the complete graph on $n\geq 3$ vertices, with every edge open with probability $p=n^{-3}$.
Let $u$ and $v$ be two vertices, and $G'$ the graph $G$ without the edge $uv$.
Then $\Rel_G(u,v)\geq p$.
A path from $u$ to $v$ in $G'$ has some length $l\geq 2$, there are at most $n^{l-1}$ such paths of length $l$, and each of them is open with probability $p^l$.
Thus, by the union bound,
\begin{align*}
  \Rel_{G'}(u,v)\leq\sum_{l\geq 2}n^{l-1}p^l=\frac{np^2}{1-np}
  \qquad\text{and}\qquad
  \frac{\Rel_{G'}(u,v)}{\Rel_G(u,v)}\leq\frac{np}{1-np}=\frac{1}{n^2-1}.
\end{align*}
Hence, for large $n$, a spanning subgraph of $G$ with the same probabilities whose two-terminal reliabilities are within a constant factor of those of $G$ must keep all $\binom n2$ edges.
Instead, for our sparsification, this graph already has total open probability less than $1$, and the total open probability of its bi-directed version is $n(n-1)p<1$ as well.
%Deleting the edges within blocks cannot increase the total open probability.
%Thus the algorithm has the required small total open probability on this example whether or not the random partition contracts any vertices. 
It does not need to produce a quotient with few edges.
\end{remark}

\section{Hardness for unreliability}\label{sec:hardness}

Finally, we prove the complementary hardness result for $s$-$t$ unreliability in \Cref{thm:unrel-hard}.
Let us define some notation first.
Let $G=(V,E)$ be an undirected graph.
We consider the case where all edges open with the same probability independently.
This only makes our hardness result stronger.
Let $G_p$ be the random subgraph where each edge $e\in E$ is open with probability $p$ independently.
Then the $s$-$t$ unreliability is defined as follows, where $s,t\in V$:
\begin{align*}
  \UnRel_G(s,t;p) \defeq \Pr_{G_p}[s\not\leftrightarrow t] = \sum_{\substack{F \subseteq E:\\ s\not\leftrightarrow t \text{ in } (V,F)}} p^{\abs{F}}(1-p)^{\abs{E\setminus F}}.
\end{align*}

Now we are ready to prove \Cref{thm:unrel-hard}. We note that the main construction here is due to Provan and Ball \cite{PB83}.
In this section we use fresh notation: the symbols $G$, $H$, $n$, $m$, $q$, $L$, $R$ and $Z$ have the meanings given here, unrelated to their meanings in the previous sections.

\begin{proof}[Proof of \Cref{thm:unrel-hard}]
  Let $G=(L\cup R,E)$ be a bipartite graph, where $L$ and $R$ are the two vertex classes.
  Each isolated vertex contributes a known factor of two to the number of independent sets, so we remove these vertices and restore their factors at the end.
  If no edge remains, we return the exact count.
  We may therefore assume that $G$ contains no isolated vertices and has at least one edge.
  We first construct a new graph $H=(V,E_H)$ as follows.
  The vertex set is $V=L\cup R\cup\{s,t\}$.
  For each edge $\{u,v\}\in E$, where $u\in L$ and $v\in R$,
  add one copy of each of the edges $\{s,u\}$, $\{u,v\}$, and $\{v,t\}$ to $E_H$.
  Note that this will introduce parallel edges, and $\abs{E_H}=3\abs{E}$.
  An illustration is given in \Cref{fig:bis-to-mincut}.

 \begin{figure}[H]
    \centering
    \begin{tikzpicture}[
      vertex/.style={circle,draw,fill=white,inner sep=1.5pt,minimum size=5mm},
      >=stealth
    ]
      % The input bipartite graph G.
      \node[vertex] (u1) at (0,1) {$u_1$};
      \node[vertex] (u2) at (0,-1) {$u_2$};
      \node[vertex] (v1) at (1.8,1) {$v_1$};
      \node[vertex] (v2) at (1.8,-1) {$v_2$};
      \draw (u1) -- (v1);
      \draw (u1) -- (v2);
      \draw (u2) -- (v2);
      \node at (0,1.65) {$L$};
      \node at (1.8,1.65) {$R$};
      \node at (0.9,-1.65) {$G$};

      % Transformation arrow.
      \draw[->,thick] (2.5,0) -- (3.7,0);

      % The graph H. Each line style tracks one edge of G.
      \node[vertex] (s) at (4.4,0) {$s$};
      \node[vertex] (hu1) at (5.9,1) {$u_1$};
      \node[vertex] (hu2) at (5.9,-1) {$u_2$};
      \node[vertex] (hv1) at (7.7,1) {$v_1$};
      \node[vertex] (hv2) at (7.7,-1) {$v_2$};
      \node[vertex] (t) at (9.2,0) {$t$};

      \draw (s) to[bend right=14] (hu1);
      \draw (s) to[bend left=14] (hu1);
      \draw (s) -- (hu2);

      \draw (hu1) -- (hv1);
      \draw (hu1) -- (hv2);
      \draw (hu2) -- (hv2);

      \draw (hv1) -- (t);
      \draw (hv2) to[bend left=14] (t);
      \draw (hv2) to[bend right=14] (t);

      \node at (5.9,1.65) {$L$};
      \node at (7.7,1.65) {$R$};
      \node at (6.8,-1.65) {$H$};
    \end{tikzpicture}
    \caption{The transformation from $G$ to $H$. Each edge $\{u,v\}\in E$, with
      $u\in L$ and $v\in R$, contributes one copy of each of
      $\{s,u\}$, $\{u,v\}$, and $\{v,t\}$.}
    \label{fig:bis-to-mincut}
  \end{figure}  

  We call a subset $S$ of vertices an \emph{$s$-set} if $s\in S$ and $t\not\in S$.
  Let $\delta_H(S)$ denote the cut set induced by $S$ in $H$, and let $I(S)=(L\setminus S) \cup (S\cap R)$. 
  We claim that
  \begin{align}\label{eqn:cut-S}
    \abs{\delta_H(S)} = \abs{E} + 2e_H(I(S)),
  \end{align}
  where $e_H(I(S))$ denotes the number of edges in $H$ between vertices of $I(S)$.
  Note that by definition, $I(S)\subseteq L\cup R$ and hence $e_H(I(S))=e_G(I(S))$.
  To see the claim, consider each edge $\{u,v\}\in E$.
  It contributes $3$ to the cut if both $u$ and $v$ are in $I(S)$,
  and contributes $1$ otherwise.
  By~\eqref{eqn:cut-S}, every $s$-set $S$ satisfies $\abs{\delta_H(S)}\geq\abs{E}$, with equality exactly when $I(S)$ is an independent set in $G$.
  The set $S=\{s\}$ attains equality, since $I(\{s\})=L$ is independent.
  Therefore the minimum $s$-$t$ cuts of $H$ have $\abs{E}$ edges.
  Moreover, $S\mapsto I(S)$ is a bijection from the $s$-sets to the subsets of $L\cup R$, with inverse $I\mapsto\{s\}\cup(L\setminus I)\cup(R\cap I)$.
  Thus $S$ defines a minimum $s$-$t$ cut in $H$ if and only if $I(S)$ is an independent set in $G$, and this is a bijection between the sets defining minimum cuts and the independent sets of $G$.

  Every such minimum-cut set $S$ is also connected to $s$ within the subgraph of $H$ induced by $S$.
  Indeed, each vertex of $S\cap L$ is adjacent to $s$.
  If $v\in S\cap R$, then independence of $I(S)$ forces every neighbour of $v$ into $S\cap L$, and at least one neighbour exists because there are no isolated vertices.
  Thus opening every edge that is not in the cut makes $S$ exactly the set of vertices reachable from $s$.

  Now, suppose the number of independent sets in $G$ is $Z$, and the open probability is $q\in(0,1)$.
  Let $n\defeq\abs{L\cup R}$ and $m\defeq\abs{E}$.
  We have that
  \begin{align}    \label{eqn:unrel-bis}
    Z (1-q)^{m}q^{2m}\le\UnRel_H(s,t;q) \le Z (1-q)^{m} + 2^n(1-q)^{m+2}.
  \end{align}
  For each $F\subseteq E_H$ such that $s$ cannot reach $t$ in $(V,F)$, let $S_F$ be the set of vertices reachable from $s$ in $(V,F)$.
  It is an $s$-set.
  On the event $S_F=S$, every edge of $\delta_H(S)$ is closed, since otherwise $s$ would reach a vertex outside $S$.
  Thus the event has probability at most $(1-q)^{\abs{\delta_H(S)}}$.
  If moreover $s$ reaches every vertex of $S$ within the subgraph of $H$ induced by $S$, then the configuration in which exactly the edges of $\delta_H(S)$ are closed has $S_F=S$, so the event has probability at least $(1-q)^{\abs{\delta_H(S)}}q^{3m-\abs{\delta_H(S)}}$.
  For the lower bound in~\eqref{eqn:unrel-bis}, sum the last bound over the $Z$ sets $S$ that define minimum cuts.
  They all have $\abs{\delta_H(S)}=m$ and, as shown above, $s$ reaches all of $S$ inside $S$.
  For the upper bound, $\UnRel_H(s,t;q)=\sum_S\Pr[S_F=S]\leq\sum_S(1-q)^{\abs{\delta_H(S)}}$ over all $2^n$ $s$-sets $S$.
  By~\eqref{eqn:cut-S}, $\abs{\delta_H(S)}=m$ for the $Z$ sets with $I(S)$ independent, and $\abs{\delta_H(S)}\geq m+2$ for the others.
  For a requested relative error $\eps\in(0,\frac12]$, choose $q$ so close to one that
  \begin{align*}
    2m(1-q)\leq\frac{\eps}{4}\qquad\text{and}\qquad 2^n(1-q)^2\leq\frac{\eps}{4}.
  \end{align*}
  Since $q^{2m}\geq 1-2m(1-q)$ and $Z\geq 1$, the estimates \eqref{eqn:unrel-bis} give
  \begin{align*}
    \left(1-\frac{\eps}{4}\right)Z\leq\frac{\UnRel_H(s,t;q)}{(1-q)^{m}}\leq\left(1+\frac{\eps}{4}\right)Z.
  \end{align*}

  The remaining question is how to implement an exponentially small closing probability with the fixed open probability $p$.
  Let $H'$ be obtained from $H$ by replacing each edge by $\ell$ parallel paths of length $2$, each with its own new internal vertex, and let every edge of $H'$ be open with probability $p$.
  A new internal vertex has degree two, so the endpoints of an edge of $H$ are connected through its paths exactly when one of these $\ell$ paths is fully open, which has probability $1-(1-p^2)^{\ell}$, independently over the edges of $H$.
  Hence $\UnRel_{H'}(s,t;p)=\UnRel_H(s,t;q)$ with $1-q=(1-p^2)^{\ell}$, and $H'$ has no parallel edges.
  Let $\ell$ be the least integer with $2m(1-p^2)^{\ell}\leq\frac{\eps}{4}$ and $2^n(1-p^2)^{2\ell}\leq\frac{\eps}{4}$.
  With $c_p\defeq-\log(1-p^2)>0$, this gives $\ell\leq 1+c_p^{-1}\max\{\log(8m/\eps),\,\frac12(n\log 2+\log(4/\eps))\}=O(n+\log(1/\eps))$ for fixed $p$, as $m\leq n^2$.
  All these quantities are computed exactly in rational arithmetic, with numbers of polynomial bit length.
  Let $a$ and $b$ be integers with $0<a<b$ and $p=a/b$, which depend only on $p$, so that $1-p^2=(b^2-a^2)/b^2$, and take the requested error $\eps$ rational, as it is part of the input.
  For $j=1,2,\ldots$, compute $(1-p^2)^j$ as the quotient of the integers $(b^2-a^2)^j$ and $b^{2j}$, and compare $2m(1-p^2)^j$ and $2^n(1-p^2)^{2j}$ with $\frac{\eps}{4}$. The first $j$ that passes both comparisons is $\ell$.
  For $j\leq\ell$, these integers have $O(\ell\log b)=O(n+\log(1/\eps))$ bits for fixed $p$, so $\ell$ is found in time polynomial in $n$ and the encoding length of $\eps$.
  Likewise, $(1-q)^m=(1-p^2)^{\ell m}$ is the quotient of the integers $(b^2-a^2)^{\ell m}$ and $b^{2\ell m}$, which have $O(\ell m\log b)=O(n^2(n+\log(1/\eps)))$ bits, so it is computed exactly in polynomial time.
  Call the assumed unreliability FPRAS on $H'$ with relative error $\eps/4$, and divide its output by $(1-q)^m$.
  On success, the resulting estimate lies between $(1-\eps/4)^2Z$ and $(1+\eps/4)^2Z$, hence within $(1\pm\eps)Z$, as $(1+\eps/4)^2=1+\frac{\eps}{2}+\frac{\eps^2}{16}\leq 1+\eps$.
  The graph $H'$ has size polynomial in $n$ and $\log(1/\eps)$, which finishes the reduction.
\end{proof}

\section*{Acknowledgements}

Weiming Feng thanks Tianyi Zhang for extensive and insightful discussions on the two-terminal reliability problem during his time at ETH Zürich in 2024. He also acknowledges support from the Hong Kong Research Grants Council under ECS Grant No.~27202725.

\section*{AI methodology}

The first version of the algorithm was found autonomously by ChatGPT-5.6 Sol Ultra first, and was simplified and streamlined by the authors.
See \cite{FFG26v1} for details.
Subsequently, the authors spent a lot of effort to accelerate the algorithm, eventually completely getting rid of the first algorithm and arriving at the current fast DFS-type estimator.
In this process, ChatGPT-6 Astra Pro and Ultra were heavily involved, although the authors provided some high-level directions, including suggesting using a graph exploration process like DFS.
The authors then streamlined the argument again and wrote the proofs. 
ChatGPT-6 Astra and Claude Opus 5.5 were also used in the writing and polishing process.
For consistency with standard mathematical exposition, the words ``we'' and ``our'' are used throughout the paper, including when presenting ideas that originate in the output of the model.

\bibliographystyle{alpha}
\bibliography{ref}

\begin{thebibliography}{AvBGM25}

\bibitem[ACJR21]{ACJR21}
Marcelo Arenas, Luis~Alberto Croquevielle, Rajesh Jayaram, and Cristian
  Riveros.
\newblock {\#}{NFA} admits an {FPRAS:} efficient enumeration, counting, and
  uniform generation for logspace classes.
\newblock {\em J. {ACM}}, 68(6):48:1--48:40, 2021.

\bibitem[AvBGM25]{AvBGM25}
Antoine Amarilli, Timothy van Bremen, Octave Gaspard, and Kuldeep~S. Meel.
\newblock Approximating queries on probabilistic graphs.
\newblock {\em Log. Methods Comput. Sci.}, 21(4), 2025.

\bibitem[Bal80]{Ball80}
Michael~O. Ball.
\newblock Complexity of network reliability computations.
\newblock {\em Networks}, 10(2):153--165, 1980.

\bibitem[Bal86]{Ball86}
Michael~O. Ball.
\newblock Computational complexity of network reliability analysis: An
  overview.
\newblock {\em IEEE Trans. Reliab.}, 35(3):230--239, 1986.

\bibitem[BK15]{BK15}
Andr{\'{a}}s~A. Bencz{\'{u}}r and David~R. Karger.
\newblock Randomized approximation schemes for cuts and flows in capacitated
  graphs.
\newblock {\em {SIAM} J. Comput.}, 44(2):290--319, 2015.

\bibitem[BP83]{BP83}
Michael~O. Ball and J.~Scott Provan.
\newblock Calculating bounds on reachability and connectedness in stochastic
  networks.
\newblock {\em Networks}, 13(2):253--278, 1983.

\bibitem[CEPP12]{ChechikEtAl2012}
Shiri Chechik, Yuval Emek, Boaz Patt{-}Shamir, and David Peleg.
\newblock Sparse reliable graph backbones.
\newblock {\em Inf. Comput.}, 210:31--39, 2012.

\bibitem[CGZZ24]{CGZZ24}
Xiaoyu Chen, Heng Guo, Xinyuan Zhang, and Zongrui Zou.
\newblock Near-linear time samplers for matroid independent sets with
  applications.
\newblock In {\em RANDOM}, pages 32:1--32:12, 2024.

\bibitem[CHLP24]{CHLP24}
Ruoxu Cen, William He, Jason Li, and Debmalya Panigrahi.
\newblock Beyond the quadratic time barrier for network unreliability.
\newblock In {\em {SODA}}, pages 1542--1567. {SIAM}, 2024.

\bibitem[CL26]{CL26}
Xiaoyu Chen and Kuikui Liu.
\newblock An {FPRAS} for counting common bases of two matroids.
\newblock {\em arXiv}, abs/2610.06724, 2026.

\bibitem[CLP25]{CenLiPanigrahi2025}
Ruoxu Cen, Jason Li, and Debmalya Panigrahi.
\newblock Network unreliability in almost-linear time.
\newblock In {\em {STOC}}, pages 120--131. {ACM}, 2025.

\bibitem[Col87]{Col87}
Charles~J. Colbourn.
\newblock {\em The Combinatorics of Network Reliability}, volume~4 of {\em The
  International Series of Monographs on Computer Science}.
\newblock Oxford University Press, New York, 1987.

\bibitem[DGGJ04]{DGGJ03}
Martin~E. Dyer, Leslie~Ann Goldberg, Catherine~S. Greenhill, and Mark Jerrum.
\newblock The relative complexity of approximate counting problems.
\newblock {\em Algorithmica}, 38(3):471--500, 2004.

\bibitem[DSST89]{DSST89}
James~R. Driscoll, Neil Sarnak, Daniel~D. Sleator, and Robert~Endre Tarjan.
\newblock Making data structures persistent.
\newblock {\em J. Comput. Syst. Sci.}, 38(1):86--124, 1989.

\bibitem[FFG26]{FFG26v1}
Weiming Feng, Yucheng Fu, and Heng Guo.
\newblock Approximating two-terminal network reliability.
\newblock {\em arXiv}, abs/2608.02523v1, 2026.

\bibitem[FG24]{FG24}
Weiming Feng and Heng Guo.
\newblock An {FPRAS} for two terminal reliability in directed acyclic graphs.
\newblock In {\em {ICALP}}, volume 297 of {\em LIPIcs}, pages 62:1--62:19.
  Schloss Dagstuhl - Leibniz-Zentrum f{\"{u}}r Informatik, 2024.

\bibitem[FGZ26]{FengGuoZhang2026}
Weiming Feng, Heng Guo, and Yiyao Zhang.
\newblock An {$\widetilde{O}(n^2)$}-time sampler for zero-field ferromagnetic
  {Ising} models.
\newblock {\em arXiv}, abs/2609.26197, 2026.

\bibitem[Fre75]{Freedman1975}
David~A. Freedman.
\newblock On tail probabilities for martingales.
\newblock {\em Ann. Probab.}, 3(1):100--118, 1975.

\bibitem[GGL16]{GGL16}
Leslie~Ann Goldberg, Rob Gysel, and John Lapinskas.
\newblock Approximately counting locally-optimal structures.
\newblock {\em J. Comput. System Sci.}, 82(6):1144--1160, 2016.

\bibitem[GJ19]{GuoJerrum2019}
Heng Guo and Mark Jerrum.
\newblock A polynomial-time approximation algorithm for all-terminal network
  reliability.
\newblock {\em SIAM J. Comput.}, 48(3):964--978, 2019.

\bibitem[GM20]{GoharshadyMohammadi2020}
Amir~Kafshdar Goharshady and Fatemeh Mohammadi.
\newblock An efficient algorithm for computing network reliability in small
  treewidth.
\newblock {\em Reliab. Eng. Syst. Saf.}, 193:106665, 2020.

\bibitem[Har60]{Harris1960}
Theodore~E. Harris.
\newblock A lower bound for the critical probability in a certain percolation
  process.
\newblock {\em Proc. Cambridge Philos. Soc.}, 56(1):13--20, 1960.

\bibitem[Jer81]{Jerrum1981}
Mark Jerrum.
\newblock {\em On the Complexity of Evaluating Multivariate Polynomials}.
\newblock {Ph.D.} dissertation, Department of Computer Science, University of
  Edinburgh, 1981.
\newblock Technical Report CST-11-81.

\bibitem[JSV04]{JSV04}
Mark Jerrum, Alistair Sinclair, and Eric Vigoda.
\newblock A polynomial-time approximation algorithm for the permanent of a
  matrix with nonnegative entries.
\newblock {\em J. ACM}, 51(4):671--697, 2004.

\bibitem[Kan94]{Kannan1994}
Ravi Kannan.
\newblock Markov chains and polynomial time algorithms.
\newblock In {\em FOCS}, pages 656--671, 1994.

\bibitem[Kar01]{Kar01}
David~R. Karger.
\newblock A randomized fully polynomial time approximation scheme for the
  all-terminal network reliability problem.
\newblock {\em {SIAM} Rev.}, 43(3):499--522, 2001.

\bibitem[Kar16]{Kar16}
David~R. Karger.
\newblock A fast and simple unbiased estimator for network (un)reliability.
\newblock In {\em {FOCS}}, pages 635--644. {IEEE} Computer Society, 2016.

\bibitem[KKT95]{KKT95}
David~R. Karger, Philip~N. Klein, and Robert~Endre Tarjan.
\newblock A randomized linear-time algorithm to find minimum spanning trees.
\newblock {\em J. {ACM}}, 42(2):321--328, 1995.

\bibitem[KL85]{KarpLuby1985}
Richard~M. Karp and Michael Luby.
\newblock {Monte-Carlo} algorithms for the planar multiterminal network
  reliability problem.
\newblock {\em J. Complexity}, 1(1):45--64, 1985.

\bibitem[Lin09]{Linusson2009}
Svante Linusson.
\newblock A note on correlations in randomly oriented graphs.
\newblock {\em arXiv}, abs/0905.2881, 2009.

\bibitem[McD81]{McDiarmid1981}
Colin McDiarmid.
\newblock General percolation and random graphs.
\newblock {\em Adv. in Appl. Probab.}, 13(1):40--60, 1981.

\bibitem[MCM24]{MCM24}
Kuldeep~S. Meel, Sourav Chakraborty, and Umang Mathur.
\newblock A faster {FPRAS} for {\#}{NFA}.
\newblock {\em Proc. {ACM} Manag. Data}, 2(2):112, 2024.

\bibitem[MdC25]{MdC25}
Kuldeep~S. Meel and Alexis de~Colnet.
\newblock Towards practical {FPRAS} for {\#}{NFA}: Exploiting the power of
  dependence.
\newblock {\em Proc. {ACM} Manag. Data}, 3(2):116:1--116:23, 2025.

\bibitem[MdC26]{MdC26}
Kuldeep~S. Meel and Alexis de~Colnet.
\newblock {\#}{CFG} and {\#}{DNNF} admit {FPRAS}.
\newblock In {\em {SODA}}, pages 5978--6010. {SIAM}, 2026.

\bibitem[PB83]{PB83}
J.~Scott Provan and Michael~O. Ball.
\newblock The complexity of counting cuts and of computing the probability that
  a graph is connected.
\newblock {\em SIAM J. Comput.}, 12(4):777--788, 1983.

\bibitem[Pro86]{Provan1986}
J.~Scott Provan.
\newblock The complexity of reliability computations in planar and acyclic
  graphs.
\newblock {\em SIAM J. Comput.}, 15(3):694--702, 1986.

\bibitem[SW85]{SatyanarayanaWood1985}
A.~Satyanarayana and R.~Kevin Wood.
\newblock A linear-time algorithm for computing {$K$}-terminal reliability in
  series-parallel networks.
\newblock {\em SIAM J. Comput.}, 14(4):818--832, 1985.

\bibitem[Val79]{Valiant1979}
Leslie~G. Valiant.
\newblock The complexity of enumeration and reliability problems.
\newblock {\em SIAM J. Comput.}, 8(3):410--421, 1979.

\bibitem[ZL11]{ZenklusenLaumanns2011}
Rico Zenklusen and Marco Laumanns.
\newblock High-confidence estimation of small {$s$-$t$} reliabilities in
  directed acyclic networks.
\newblock {\em Networks}, 57(4):376--388, 2011.

\end{thebibliography}

\end{document}